\documentclass[11pt]{article}
\usepackage[margin=1in]{geometry}
\usepackage[T1]{fontenc}
\usepackage{lmodern}
\usepackage{amsmath,amssymb,amsthm,mathtools,bm}
\usepackage{booktabs,array,graphicx,float,setspace,longtable,pdflscape,caption}
\usepackage[super,comma,sort&compress]{natbib}
\usepackage[hidelinks]{hyperref}
\usepackage[table]{xcolor}
\definecolor{yescell}{RGB}{207,226,243}
\definecolor{nocell}{RGB}{244,244,244}
\newcommand{\Yes}{\cellcolor{yescell}Yes}
\newcommand{\No}{\cellcolor{nocell}\textcolor{black!55}{No}}
\newcommand{\E}{\mathbb{E}}
\newcommand{\Var}{\mathrm{Var}}
\newcommand{\Cov}{\mathrm{Cov}}
\newcommand{\tr}{\mathrm{trace}}
\newcommand{\diag}{\mathrm{diag}}
\newcommand{\plim}{\xrightarrow{p}}
\newcommand{\dlim}{\xrightarrow{d}}
\newcommand{\betahat}{\hat{\beta}}
\newcommand{\betastar}{\hat{\beta}^*}
\newcommand{\VarAR}{\hat{V}_{AR}}

\newtheorem{theorem}{Theorem}
\newtheorem{lemma}{Lemma}
\newtheorem{corollary}{Corollary}
\newtheorem{proposition}{Proposition}

\graphicspath{{figures/}}
\newcommand{\proofheading}[1]{\subsubsection*{#1}}
\hypersetup{pdfauthor={Awan Afiaz and M. Shafiqur Rahman},pdftitle={Overstuffed sandwiches and separation anxiety: finite-sample variance estimation for penalized GEE with near-separated binary data}}

\title{Overstuffed sandwiches and separation anxiety: finite-sample variance estimation for penalized GEE with near-separated binary data}

\author{
Awan Afiaz\textsuperscript{1,2,*} and
M.\ Shafiqur Rahman\textsuperscript{2,\ensuremath{\dagger}}
\\[6pt]
{\small\textsuperscript{1}Department of Biostatistics, University of Washington}\\
{\small\textsuperscript{2}Institute of Statistical Research and Training, University of Dhaka}
}

\date{}

\begin{document}

\maketitle
\begingroup
\renewcommand{\thefootnote}{\fnsymbol{footnote}}
\footnotetext[1]{Emails: \href{mailto:aafiaz@uw.edu}{aafiaz@uw.edu}; \href{mailto:aafiaz@isrt.ac.bd}{aafiaz@isrt.ac.bd}}
\footnotetext[2]{Email: \href{mailto:shafiq@isrt.ac.bd}{shafiq@isrt.ac.bd}}
\endgroup
\setcounter{footnote}{0}

\begin{abstract}
Penalized generalized estimating equations (PGEE) stabilize point estimation for longitudinal binary data under near-separation, but inference still depends on how the sandwich variance is corrected. Applications to toenail trial and contagious bovine pleuropneumonia data show that the existing estimators can disagree on significance in small samples, and that several cannot be computed when subjects have unequal numbers of observations. We study which variance estimator to use for PGEE in small longitudinal studies with rare events. We establish first-order asymptotics for PGEE along convergent interior-root sequences and derive a matrix characterization of the parameter-specific overcorrection induced by full leverage adjustment. We compare thirteen existing estimators with a proposed estimator that combines full inverse leverage correction, score centering and finite-sample scaling. With ten subjects, only the proposed estimator and Morel's control the type I error rate without assuming that all subjects share a correlation structure, and the proposed estimator is closer to the nominal level. In small samples, its scaling removes the excess rejections left by leverage correction and its centering offsets part of the resulting conservatism. In larger samples the leverage-corrected estimators behave alike. Estimators that average residual covariance across subjects also control the type I error rate, in small and larger samples, but they require a shared correlation structure and equal numbers of observations per subject, which may not always be feasible.
\end{abstract}


\section{Introduction}
\label{sec:intro}

Longitudinal and clustered binary outcomes are common in biomedical research. A typical repeated-measures design records a binary response (a disease indicator, a treatment failure, an adverse event) on each subject at several follow-up times together with baseline or time-varying covariates \cite{fitzmaurice2012applied}. Repeated measurements within a subject are correlated and valid inference on the regression coefficients requires accounting for this within-subject dependence. The generalized estimating equations (GEE) framework of Liang and Zeger\cite{liang1986longitudinal} specifies a marginal mean model for the outcome without requiring the full joint distribution, and is widely used when the scientific question concerns marginal (population-averaged) effects. Under standard regularity conditions and a correctly specified mean model, the GEE estimator is consistent and asymptotically normal, and the sandwich (empirical) covariance estimator is consistent for the asymptotic variance even when the working correlation structure is misspecified.

These large-sample approximations can be inaccurate when the number of subjects $N$ is small. The sandwich estimator is downward biased in finite samples \cite{mancl2001covariance}, so Wald tests and confidence intervals built from it are anticonservative. The finite-sample type I error rate can exceed the nominal level, sometimes by a factor of two or more. Small-sample inference matters in settings where GEE is routinely applied, including cluster-randomized trials with a handful of clusters per arm, pediatric and rare-disease studies where recruitment is capped, and pilot or feasibility trials expected to deliver credible standard errors under limited resources. In such studies, the sandwich bias can materially inflate false-positive rates.

The finite-sample separation problem worsens when events are rare or when covariates partition the response space unevenly. Albert and Anderson\cite{albert1984existence} identified two separation regimes for logistic regression. The first is complete separation, where a linear combination of covariates perfectly predicts the outcome and the maximum likelihood estimator has no finite solution, and the second is quasicomplete separation, where the separation holds up to ties on the boundary. Mondol and Rahman\cite{mondol2019bias} described a third regime, near-to-quasi-complete separation, in which covariate patterns have few but nonzero events, a situation also described as sparsity in the broader logistic-regression literature \cite{greenland2016sparse}. Although GEE is not a likelihood method, the same finite-sample estimation instability transfers to the score equations. A quasi-score built around a logit link inherits the same convergence failures and parameter explosions as the MLE when events are scarce and subjects are few \cite{mondol2019bias}. Near-separation is promoted by low event rates, small numbers of subjects, and within-subject correlation that concentrates events in a subset of subjects. Conditions of this type arise often in small clinical trials of rare endpoints.

Firth\cite{firth1993bias} proposed a modification of the score equation that removes the first-order bias of the maximum likelihood estimator in regular exponential families. For binary regression, the Firth penalty also guarantees finite parameter estimates under separation. Mondol and Rahman\cite{mondol2019bias} adapted this construction to GEE by adding a Firth-type penalty to the GEE score equation, producing penalized GEE (PGEE). Their simulations showed that PGEE reduces coefficient bias and improves convergence relative to standard GEE under separation and near-separation. PGEE can often restore stable point estimates in small-sample, rare-event fits, but stable point estimation does not resolve inference. Valid Wald tests and confidence intervals still depend on how the covariance matrix is estimated, and whether existing GEE-oriented corrections remain appropriate under the penalty remains open.

Existing small-sample corrections address residual shrinkage through leverage adjustment \cite{kauermann2001note,mancl2001covariance}, add positive-definite stabilizers \cite{morel2003small,rogers2015modification}, or pool residual covariance estimates across subjects \cite{pan2001robust,gosho2014robust,wang2011modified}. Hybrids combine these mechanisms \cite{ford2018comparison}. Their finite-sample performance has been studied for standard GEE in cluster-randomized and longitudinal designs \cite{pengli2015small,wang2016covariance,ford2018comparison}. We evaluate these corrections for PGEE.

What is known about PGEE variance estimation is limited. Mondol and Rahman\cite{mondol2019bias} applied Morel's correction because it guarantees a positive definite estimate and its adjustment vanishes as $N$ increases, and they noted that further study was required to estimate the standard error and type I error rate correctly under PGEE. Geroldinger et~al.\cite{geroldinger2022investigation} compared PGEE to an augmented-GEE alternative for point estimation but again restricted their variance analysis to the Morel estimator. Gosho et~al.\cite{gosho2023comparison} compared nine covariance estimators under GEE, bias-corrected GEE and PGEE in simulations and found Morel's correction the most reasonable choice, but the paper did not explain why. Ishii et~al.\cite{ishii2024geessbin} released the \texttt{geessbin} R package for modified GEE methods, including GEE, bias-corrected GEE, and PGEE with multiple bias-adjusted covariance estimators. The contribution is computational rather than theoretical. Geroldinger et~al.\cite{geroldinger2022investigation} also showed that, contrary to the original claim of Mondol and Rahman, the Firth penalty does not guarantee finite estimates for PGEE under complete separation. The existing recommendations rest on simulation alone, and variance estimation for PGEE remains poorly understood in three respects. First, no asymptotic theory has been established for PGEE and subsequent work assumes without proof that PGEE inherits the large-sample properties of GEE. Second, no explanation has been given for why Morel's correction outperforms the alternatives in this setting. Third, existing bias corrections were derived under GEE assumptions that the penalty may violate, and there is no PGEE-specific bias framework to guide practitioners in the high-leverage regime near separation.

We answer the call of Mondol and Rahman\cite{mondol2019bias} for further study and address these three gaps. We establish first-order asymptotic theory for PGEE along convergent interior-root sequences (Lemma~\ref{lem:penalty-order} and Theorem~\ref{thm:pgee-asymptotics}) and characterize parameter-specific leverage overcorrection under correctly specified working covariance (Theorem~\ref{thm:overcorrection}). The balanced intercept-treatment case identifies the role of the smaller treatment arm (Corollary~\ref{cor:overcorrection}). We derive a new variance estimator for PGEE in small samples, $\VarAR$, which combines full inverse leverage correction, score centering and finite-sample scaling, and give bias expressions and consistency under estimator-specific conditions (Propositions~\ref{prop:bias} and~\ref{prop:consistency}). We compare it with the thirteen existing estimators in simulations, in order to give practical recommendations for PGEE analyses of small studies. We first motivate these aims with two data examples.

\section{Motivating examples}
\label{sec:motivating}

We now discuss two examples that demonstrate why the choice of variance estimator for PGEE needs careful consideration. They show what the gaps identified above mean in practice, and why practitioners need guidance on this choice. In the first, the existing estimators are applied to the same small sample and disagree on whether the treatment effect is significant. In the second, five of them cannot be computed because herd sizes are unequal.

In each example we fit the PGEE model with exchangeable working correlation and apply the thirteen existing estimators reviewed in the Methods. For every computable estimator, $p$-values come from the same two-sided Wald $t$-test with $N-p$ degrees of freedom, where $p$ is the number of regression coefficients, including the intercept. Table~\ref{tab:examples} gives the standard errors and $p$-values for both examples. Web Appendix~\ref{app:applied-full} gives the results for all coefficients and for further toenail samples.

\begin{table}[H]
\centering\small
\caption{Standard errors and $p$-values in the two motivating examples under the thirteen existing variance estimators. The toenail columns refer to the treatment coefficient $\hat\beta_1=-3.16$ in ten patients, and the CBPP columns to the herd-size coefficient $\hat\beta_1=-0.41$ in fifteen herds of unequal size. A dash marks an estimator that cannot be computed, and $^{*}$ marks $p<0.05$.}\label{tab:examples}
\begin{tabular}{llcccc}
\toprule
 & & \multicolumn{2}{c}{Toenail} & \multicolumn{2}{c}{CBPP} \\
\cmidrule(lr){3-4}\cmidrule(lr){5-6}
Group & Estimator & SE & $p$ & SE & $p$ \\
\midrule
Baseline & $\hat V_{LZ}$ & 0.97 & 0.014$^{*}$ & 0.44 & 0.373 \\
Baseline & $\hat V_{DF}$ & 1.16 & 0.029$^{*}$ & 0.49 & 0.424 \\
Leverage & $\hat V_{KC}$ & 1.25 & 0.039$^{*}$ & 0.48 & 0.410 \\
Leverage & $\hat V_{MD}$ & 1.66 & 0.098 & 0.52 & 0.447 \\
Leverage & $\hat V_{FG}$ & 1.26 & 0.041$^{*}$ & 0.50 & 0.426 \\
Additive & $\hat V_{MBN}$ & 1.62 & 0.092 & 0.51 & 0.440 \\
Pooling & $\hat V_{\mathrm{Pan}}$ & 2.42 & 0.232 & --- & --- \\
Pooling & $\hat V_{GST}$ & 2.89 & 0.310 & --- & --- \\
Hybrids & $\hat V_{RS}$ & 2.76 & 0.290 & --- & --- \\
Hybrids & $\hat V_{WL}$ & 2.81 & 0.298 & --- & --- \\
Hybrids & $\hat V_{WB}$ & 2.60 & 0.264 & --- & --- \\
Hybrids & $\hat V_{FW}$ & 1.47 & 0.068 & 0.50 & 0.429 \\
Hybrids & $\hat V_{FZ}$ & 1.46 & 0.067 & 0.49 & 0.416 \\
\bottomrule
\end{tabular}
\end{table}

\subsection{Toenail onychomycosis trial}

The toenail data originate from a 12-week double-blind randomized trial of $372$ patients comparing terbinafine ($250$ mg/day) against itraconazole ($200$ mg/day) for dermatophyte toe onychomycosis \cite{debacker199612}. Because a suitable dataset from a small study could not be obtained, we use a chosen subsample of this public trial to illustrate disagreement among variance estimators. From the version distributed in the \texttt{geesmv} R package \cite{wang2015geesmv}, we use a chosen subsample of ten patients who completed all seven visits, giving 70 observations. The selection procedure and the seed are given in Web Appendix~\ref{app:applied-full}.

The binary outcome is infection severity, coded as $0$ for none or mild and $1$ for moderate or severe. We fit the marginal model
\begin{equation*}
\text{logit}(\mu_{ij}) = \beta_0 + \beta_1 x_i + \beta_2 t_{ij},
\quad i = 1, \ldots, N, \; j = 1, \ldots, 7,
\end{equation*}
where $x_i$ is the treatment indicator and $t_{ij}$ is the fractional-month time covariate. The selected sample has eight patients in one treatment arm and two in the other, with no moderate or severe infections in the two-patient arm, producing quasi-complete separation.

At the PGEE treatment estimate $\hat\beta_1=-3.16$, standard errors across the thirteen variance estimators range from 0.97 to 2.89, and $p$-values range from 0.014 to 0.310 (Table~\ref{tab:examples}). Four estimators give a significant treatment effect at the 0.05 level and nine do not. The existing corrections were developed for ordinary GEE or for linear-model sandwich inference, without the PGEE penalty. It is therefore not clear which of them should be used for treatment-effect inference in this PGEE analysis.

\subsection{Contagious bovine pleuropneumonia}

The CBPP data record incidence of contagious bovine pleuropneumonia across $15$ cattle herds in the Ethiopian highlands over four time periods \cite{lesnoff2004cbpp}. The data are available in the \texttt{lme4} R package \cite{bates2015lme4} in aggregated form (number infected per herd-period). We expand to individual-level binary outcomes. Cluster sizes range from $26$ to $96$ observations per herd, giving $n^* = 842$ total observations with an overall event rate of $11.8\%$.

We define $h_i$ equal to one if herd $i$'s period-1 size exceeds the median of 19 animals across the 15 herds, and zero otherwise. Six herds have $h_i=1$. We fit
\begin{equation*}
\text{logit}(\mu_{ij}) = \beta_0 + \beta_1 h_i + \beta_2 t_{ij},
\end{equation*}
where $t_{ij}$ indexes period. The fitted coefficient for the herd-size indicator is $\hat\beta_1=-0.409$.

The direct covariance-pooling estimators average residual covariance matrices across herds, which requires matrices of the same dimensions. With unequal herd sizes, the Pan and Gosho--Sato--Takeuchi estimators and the Wang--Long, Westgate--Burchett and Rogers--Stoner hybrids therefore cannot be computed in the pooled forms compared here (Table~\ref{tab:examples}). The analyst must choose among the eight estimators that remain computable.

\section{Methods}
\label{sec:pgee-theory}

\subsection{Notation}
\label{sec:notation}

Consider $N$ independent subjects, of which $N_1$ have a subject-level treatment indicator $x_i = 1$ and $N_0 = N - N_1$ have $x_i = 0$. Write $N_{\min} = \min(N_0, N_1)$. We use ``subject'' and ``cluster'' interchangeably. In a repeated-measures design the cluster is the subject, and $n_i$ denotes the number of repeated observations on subject or cluster $i$. Unless a component is indexed explicitly, $\beta$ denotes the full $p$-dimensional parameter vector. For subject $i$ ($i = 1, \ldots, N$), the response vector $y_i$ is $n_i \times 1$ and the covariate matrix $X_i$ is $n_i \times p$. The total number of observations is $n^* = \sum_i n_i$. Let $\mu_{ij} = \E(Y_{ij}\mid X_{ij})$ denote the marginal mean for observation $j$ in subject $i$, let $\phi$ denote the dispersion parameter, and let $R_i(\alpha)$ denote the working correlation matrix indexed by parameter vector $\alpha$.

The variance function matrix is $W_i = \diag\{\mu_{ij}(1-\mu_{ij})\}$, and the working covariance matrix is $V_i = \phi W_i^{1/2} R_i(\alpha) W_i^{1/2}$, both of dimension $n_i \times n_i$. For the canonical logistic link, the derivative matrix is $D_i = W_i X_i$, of dimension $n_i \times p$. Define the sensitivity matrix $I_0 = \sum_i D_i^T V_i^{-1} D_i$, its inverse $\Delta = I_0^{-1}$, and the per-subject information $A_i = D_i^T V_i^{-1} D_i$, each of dimension $p \times p$. Thus $I_0 = \sum_i A_i$. The hat matrix block for cluster $i$ is $H_{ii} = D_i \Delta D_i^T V_i^{-1}$, of dimension $n_i \times n_i$.

The fitted marginal mean is $\hat{\mu}_i = \mu_i(\betahat)$, the fitted residual is $r_i = y_i - \hat{\mu}_i$, and the true residual is $e_i = y_i - \mu_i(\beta_0)$. The target variability matrix is $M_0 = \sum_i D_i^T V_i^{-1} \Cov(y_i) V_i^{-1} D_i$. Regularity conditions (R0)--(R7) for fixed $p$ and growing $N$ are given in Web Appendix~\ref{app:proofs}, Section~\ref{sec:regularity}.

\subsection{GEE and sandwich variance}

The GEE estimator $\betahat$ is any root of the estimating equation
\begin{equation}
U(\beta) = \sum_{i=1}^{N} D_i(\beta)^T V_i(\beta,\alpha)^{-1}\{y_i - \mu_i(\beta)\} = 0.
\label{eq:gee}
\end{equation}
For the canonical logistic mean model in the Notation subsection, $D_i(\beta) = W_i(\beta)X_i$ and $V_i(\beta,\alpha) = \phi W_i(\beta)^{1/2}R_i(\alpha)W_i(\beta)^{1/2}$. A first-order expansion of $U(\betahat)$ about $\beta_0$ gives the covariance target $\Cov(\betahat) \approx \Delta M_0 \Delta$.

The Liang-Zeger sandwich estimator replaces $\Cov(y_i)$ in $M_0$ with $r_i r_i^T$,
\begin{equation}
\hat{V}_{LZ} = \Delta \left[\sum_{i=1}^{N} D_i^T V_i^{-1} r_i r_i^T V_i^{-1} D_i\right] \Delta.
\label{eq:lz}
\end{equation}
This is consistent as $N \to \infty$ but downward biased in finite samples.

\subsection{Penalized GEE}
\label{sec:asymptotics}

The PGEE modifies the GEE score by adding a Firth-type penalty \cite{firth1993bias},
\begin{equation}
U^*(\beta) = U(\beta) + b(\beta) = 0, \quad b_r(\beta) = \frac{1}{2}\tr\!\left(I_0^{-1} \frac{\partial I_0}{\partial \beta_r}\right).
\label{eq:pgee}
\end{equation}
Mondol and Rahman\cite{mondol2019bias} adapted this construction from penalized likelihood to GEE, applying the adjustment to a quasi-score. The penalty stabilizes estimation under near-separation, but Geroldinger et~al.\cite{geroldinger2022investigation} showed that PGEE estimates can still diverge under complete separation. The penalty in \eqref{eq:pgee} uses the working sensitivity matrix $I_0$. In maximum likelihood, $I_0$ equals the Fisher information and the Firth penalty removes the first-order bias of the MLE. In the GEE setting, $I_0$ coincides with the true information only when the working correlation is correctly specified. Under misspecification, the penalty preserves the $O(1)$ bound of Lemma~\ref{lem:penalty-order}.

At $\beta_0$, $\E[U^*(\beta_0)] = b(\beta_0) \neq 0$. The PGEE estimating equation is biased by design. Lemma~\ref{lem:penalty-order} shows the bias is bounded and Theorem~\ref{thm:pgee-asymptotics} shows it does not affect the limiting distribution.

\begin{lemma}[Firth penalty order]
\label{lem:penalty-order}
Under regularity conditions (R0)--(R4), the Firth penalty satisfies $\|b(\beta)\| = O(1)$ and $\|\partial b / \partial \beta^T\|_{\text{sp}} = O(1)$ uniformly on $\Theta$, where $\|\cdot\|_{\text{sp}} = \sigma_{\max}(\cdot)$ denotes the spectral norm.
\end{lemma}

The proof is given in Web Appendix~\ref{app:proofs}, Section~\ref{sec:proofs}. Lemma~\ref{lem:penalty-order} extends the result of Kosmidis and Firth\cite{kosmidis2009bias}, who established $A_t = O_p(1)$ for the bias-reduction adjustment within the exponential family, to estimating equations where no loglikelihood exists.

\begin{theorem}[PGEE asymptotic distribution]
\label{thm:pgee-asymptotics}
Under (R0)--(R7), let $\{\betastar_N\}$ be any sequence of PGEE roots with $\betastar_N \in \mathrm{int}(\Theta)$ for all sufficiently large $N$. Then
\begin{enumerate}
\item[(a)] $\betastar_N \plim \beta_0$;
\item[(b)] $\sqrt{N}(\betastar_N - \beta_0) \dlim N(0, \Gamma_0^{-1}\Sigma_0\Gamma_0^{-1})$,
\end{enumerate}
where $\Gamma_0 = \lim_{N \to \infty} N^{-1} I_0$ and $\Sigma_0 = \lim_{N \to \infty} N^{-1} M_0$.
\end{theorem}

The proof is given in Web Appendix~\ref{app:proofs}, Section~\ref{sec:proofs}.

PGEE and GEE solve different equations and differ in finite samples but share the same limiting distribution along any convergent interior-root sequence. The practical consequence is that GEE variance estimators target the correct quantity under PGEE in that asymptotic regime. Define $I_0^* = I_0 - \partial b / \partial \beta^T$ as the PGEE sensitivity matrix and $H_{ii}^* = D_i(I_0^*)^{-1}D_i^TV_i^{-1}$ as the PGEE hat matrix. The difference $H_{ii}^* - H_{ii} = O(N^{-2})$, so $H_{ii}$ evaluated at $\betastar$ is a sufficient first-order approximation for variance estimation.

Regularity condition (R0) requires that the marginal probabilities remain bounded away from 0 and 1. Near-separation is a finite-sample phenomenon. Under a fixed, interior $\beta_0$, the law of large numbers guarantees that complete or quasi-complete separation occurs with probability tending to zero as $N \to \infty$. Theorem~\ref{thm:pgee-asymptotics} describes this stable limiting regime along convergent interior-root sequences. It does not resolve the finite-sample paths on which PGEE fails to converge. It shows only that when an interior root is selected, the penalty does not alter the first-order limit.

\subsection{Existing variance corrections for PGEE}
\label{sec:existing}

\subsubsection{Correction mechanisms and requirements}

We evaluate thirteen GEE variance estimators as plug-in covariance estimators for PGEE. Table~\ref{tab:assumptions} groups them by correction mechanism and required assumptions. Complete formulas and explanations are in Web Appendix~\ref{app:existing}. The baseline estimators are the Liang--Zeger sandwich ($\hat V_{LZ}$) and the scalar degrees-of-freedom correction ($\hat V_{DF}$) \cite{liang1986longitudinal,mackinnon1985some}. Leverage corrections adjust for the shrinkage of fitted residuals. Kauermann--Carroll ($\hat V_{KC}$) uses $(I-H_{ii})^{-1/2}$, whereas Mancl--DeRouen ($\hat V_{MD}$) uses $(I-H_{ii})^{-1}$ \cite{kauermann2001note,mancl2001covariance}. These residual-level corrections require invertibility. Fay--Graubard ($\hat V_{FG}$) instead uses clipped diagonal score-level adjustments \cite{fay2001small}. The Morel--Bokossa--Neerchal estimator ($\hat V_{MBN}$) centers the score contributions, applies finite-sample multipliers, and adds a positive-definite stabilizer \cite{morel2003small}. Pan ($\hat V_{\mathrm{Pan}}$) and Gosho--Sato--Takeuchi ($\hat V_{GST}$) instead pool standardized residual outer products across subjects \cite{pan2001robust,gosho2014robust}. Pooling assumes a common correlation structure, and the direct unstructured implementation requires equal numbers of observations per subject.

\begin{table}[H]
\centering\small
\caption{Assumptions required by the thirteen existing variance estimators, grouped by correction mechanism.}\label{tab:assumptions}
\resizebox{\ifdim\width>\linewidth\linewidth\else\width\fi}{!}{%
\begin{tabular}{llccc}
\toprule
Group & Estimator(s) & Common correlation & Equal subject sizes & $(I-H_{ii})$ invertible \\
\midrule
Baseline & $\hat V_{LZ}$, $\hat V_{DF}$ & \No & \No & \No \\
Leverage corrections & $\hat V_{KC}$, $\hat V_{MD}$ & \No & \No & \Yes \\
Leverage corrections & $\hat V_{FG}$ & \No & \No & \No \\
Additive stabilizers & $\hat V_{MBN}$ & \No & \No & \No \\
Pooling estimators & $\hat V_{\mathrm{Pan}}$, $\hat V_{GST}$ & \Yes & \Yes & \No \\
Hybrids & $\hat V_{RS}$ (pooling + additive) & \Yes & \Yes & \No \\
Hybrids & $\hat V_{WL}$, $\hat V_{WB}$ (pooling + leverage) & \Yes & \Yes & \Yes \\
Hybrids & $\hat V_{FW}$, $\hat V_{FZ}$ & \No & \No & \Yes \\
\bottomrule
\end{tabular}}
\end{table}

Common correlation means a common within-subject correlation structure across subjects. Equal subject sizes are required by the direct unstructured pooled-correlation implementation used here. The last column identifies the residual-level leverage condition for the inverse correction. $\hat{V}_{FG}$ uses a score-level diagonal clipping rather than a residual-level inverse and therefore does not require $(I-H_{ii})$ invertibility.

The hybrids combine these mechanisms. Rogers--Stoner ($\hat V_{RS}$) adds a Morel-type stabilizer to a pooled estimator \cite{rogers2015modification}. Wang--Long ($\hat V_{WL}$) and Westgate--Burchett ($\hat V_{WB}$) combine pooling with leverage correction \cite{wang2011modified,westgate2016improving}. Ford--Westgate ($\hat V_{FW}$) averages $\hat V_{KC}$ and $\hat V_{MD}$ \cite{ford2018comparison}, while Fan--Zhang--Zhang ($\hat V_{FZ}$) subtracts estimated cross-subject contamination from the $\hat V_{MD}$ correction \cite{fan2012robust}. The pooled hybrids retain the common-correlation and equal-size requirements, and $\hat V_{WL}$ and $\hat V_{WB}$ also require $(I-H_{ii})$ to be invertible. $\hat V_{FW}$ and $\hat V_{FZ}$ retain the invertibility requirement without requiring pooling.

\subsubsection{First-order effect of full inverse leverage correction}

Fitting the mean model shrinks the residuals, which is why the leverage corrections discussed above inflate them. In the first-order residual second moment, full leverage correction removes this shrinkage but leaves a positive semidefinite cross-subject contribution (Web Appendix~\ref{app:proofs}, Section~\ref{sec:bias-problem}). Theorem~\ref{thm:overcorrection} characterizes the first-order overcorrection matrix for the full inverse correction. Write $\hat M_{MD}$ for the uncentered middle matrix of the full inverse correction, so that $\hat V_{MD}=\Delta\hat M_{MD}\Delta$.

\begin{theorem}[Direction-dependent overcorrection]
\label{thm:overcorrection}
Under correctly specified working covariance ($V_i = \Cov(y_i)$) and the first-order Taylor expansion of the fitted residuals, the overcorrection matrix of any estimator applying the full $(I - H_{ii})^{-1}$ correction to the score contributions is
\begin{equation}
B_{\mathrm{lev}} = \sum_{i=1}^N A_i (I_0 - A_i)^{-1} A_i,
\label{eq:blev}
\end{equation}
where $A_i = D_i^T V_i^{-1} D_i$ is the per-subject contribution to $I_0 = \sum_i A_i$.
\end{theorem}

\begin{corollary}[Scalar overcorrection bound]
\label{cor:overcorrection}
Under the conditions of Theorem~\ref{thm:overcorrection} with $p = 2$ (intercept and a subject-level binary covariate), $N_{\min} \ge 2$, and common per-subject information contributions within each treatment group (as when subjects within each group have the same number of repeated observations and the same working covariance), the maximum eigenvalue of $M_0^{-1}\E[\hat{M}_{MD}]$ is
\begin{equation}
\lambda_{\max}\!\left(M_0^{-1} \E\!\left[\hat{M}_{MD}\right]\right) = \frac{N_{\min}}{N_{\min} - 1},
\label{eq:scalar-bound}
\end{equation}
where $N_{\min} = \min(N_0, N_1)$. With $N_{\min} = 3$ (a common clinical trial configuration under PGEE), the overcorrection is 50\%.
\end{corollary}

The residual expansion \eqref{eq:residual-expansion} is given in Web Appendix~\ref{app:proofs}, Section~\ref{sec:bias-problem}, and the proofs of Theorem~\ref{thm:overcorrection} and Corollary~\ref{cor:overcorrection} in Section~\ref{sec:proofs}. A coefficient-level diagnostic derived from this first-order result, together with its values for the applied examples, is provided in Web Appendix~\ref{app:applied-full}, Section~\ref{sec:coefficient-diagnostic}. These results characterize mean overcorrection under the first-order approximation. In small samples with few events the correction is nevertheless too small in a typical dataset. With ten subjects and a $10\%$ event rate, the median standard error of $\hat V_{MD}$ is $0.79$ of the simulation standard error (Web Figure~\ref{fig:se-ratio}), and its test rejects too often.

\subsubsection{Components used by the existing estimators}

Table~\ref{tab:ingredients} shows which components each estimator uses. Leverage correction adjusts the residuals or score contributions for the shrinkage caused by fitting the mean model. Score centering subtracts the mean score contribution before the outer products are formed. Scalar finite-sample scaling multiplies the whole middle matrix by a degrees-of-freedom factor, such as $N/(N-p)$ in $\hat V_{DF}$ and $\hat V_{GST}$. Pooling averages standardized residual outer products across subjects. An additive stabilizer adds a positive-definite term to the sandwich.

\begin{table}[H]
\centering\small
\caption{Components used by the thirteen existing variance estimators and by the proposed estimator. Only the proposed estimator combines leverage correction, score centering and scaling.}\label{tab:ingredients}
\resizebox{\ifdim\width>\linewidth\linewidth\else\width\fi}{!}{%
\begin{tabular}{llccccc}
\toprule
Group & Estimator & Leverage correction & Score centering & Scalar finite-sample scaling & Pooling & Additive stabilizer \\
\midrule
Baseline & $\hat V_{LZ}$ & \No & \No & \No & \No & \No \\
Baseline & $\hat V_{DF}$ & \No & \No & \Yes & \No & \No \\
Leverage corrections & $\hat V_{KC}$ & \Yes & \No & \No & \No & \No \\
Leverage corrections & $\hat V_{MD}$ & \Yes & \No & \No & \No & \No \\
Leverage corrections & $\hat V_{FG}$ & \Yes & \No & \No & \No & \No \\
Additive stabilizers & $\hat V_{MBN}$ & \No & \Yes & \Yes & \No & \Yes \\
Pooling estimators & $\hat V_{\mathrm{Pan}}$ & \No & \No & \No & \Yes & \No \\
Pooling estimators & $\hat V_{GST}$ & \No & \No & \Yes & \Yes & \No \\
Hybrids & $\hat V_{RS}$ & \No & \No & \No & \Yes & \Yes \\
Hybrids & $\hat V_{WL}$ & \Yes & \No & \No & \Yes & \No \\
Hybrids & $\hat V_{WB}$ & \Yes & \No & \No & \Yes & \No \\
Hybrids & $\hat V_{FW}$ & \Yes & \No & \No & \No & \No \\
Hybrids & $\hat V_{FZ}$ & \Yes & \No & \No & \No & \No \\
\midrule
Proposed & $\VarAR$ & \Yes & \Yes & \Yes & \No & \No \\
\bottomrule
\end{tabular}}
\end{table}

\subsection{The proposed estimator}
\label{sec:proposed}

\subsubsection{Leverage correction, score centering and finite-sample scaling}

The proposed estimator $\VarAR$ combines full inverse leverage correction with score centering and the finite-sample multipliers used by Morel et al.\cite{morel2003small}. All fitted quantities below are evaluated at the same PGEE root $\betastar$. The full inverse correction removes residual self-shrinkage in the first-order expansion, with the remaining cross-subject contribution characterized by Theorem~\ref{thm:overcorrection}. This is the same leverage step used by $\hat V_{MD}$. Define the corrected score contribution
\begin{equation}
f_i = D_i^T V_i^{-1}(I - H_{ii})^{-1}r_i.
\label{eq:corrected-score}
\end{equation}

Score centering replaces the empirical second moment of these contributions with their empirical covariance. Let $d_i=D_i^TV_i^{-1}r_i$ denote the uncorrected score contribution. At a PGEE root, the estimating equation gives
\begin{equation}
\sum_{i=1}^N d_i=-b(\betastar).
\label{eq:pgee-score-sum}
\end{equation}
The uncorrected contributions sum to zero at a GEE root, but need not do so at a PGEE root. When $(I-H_{ii})$ is invertible for every subject, $f_i-d_i=A_i\Delta f_i$, and the mean corrected contribution satisfies
\begin{equation}
\bar f=-\frac{b(\betastar)}{N}+\frac{1}{N}\sum_{i=1}^N A_i\Delta f_i.
\label{eq:corrected-score-mean}
\end{equation}
This mean contains a penalty part and a leverage part, the latter of which can also be present at a GEE root. Writing $\pi=-b(\betastar)/N$ and $\lambda=N^{-1}\sum_i A_i\Delta f_i$, the amount removed by centering is
\begin{equation}
\begin{aligned}
\sum_i f_if_i^T-\sum_i(f_i-\bar f)(f_i-\bar f)^T
&=N\bar f\bar f^T\\
&=N\pi\pi^T+N\lambda\lambda^T+N(\pi\lambda^T+\lambda\pi^T).
\end{aligned}
\label{eq:centering-components}
\end{equation}
The derivations of these identities are given in Web Appendix~\ref{app:proofs}, Section~\ref{sec:score-centering-identities}.
These identities give a covariance-based motivation for centering under PGEE. Under the stated regularity conditions, centering preserves first-order consistency, as in Proposition~\ref{prop:consistency}.

For finite-sample scaling, we use Morel's multiplier
\begin{equation}
c_N=\frac{n^*-1}{n^*-p}\cdot\frac{N}{N-1},\qquad N>1,\quad n^*>p.
\label{eq:finite-sample-multiplier}
\end{equation}
The first factor adjusts for the regression dimension using the total number of observations, and the second is the Bessel-like adjustment associated with centering across subjects. Both tend to one under the fixed-$p$, growing-$N$ regime. The choice of $c_N$ follows Morel's construction, and its effect on test calibration is evaluated alongside that of centering in the simulations.

Combining the three components gives
\begin{equation}
\begin{aligned}
\hat{M}_{AR} &= \frac{n^*-1}{n^*-p} \cdot \frac{N}{N-1} \sum_{i=1}^N (f_i - \bar{f})(f_i - \bar{f})^T, \\
\VarAR &= \Delta\, \hat{M}_{AR}\, \Delta,
\end{aligned}
\label{eq:var-ar}
\end{equation}
where $\bar{f} = N^{-1}\sum_{i=1}^N f_i$. Proposition~\ref{prop:bias} characterizes the resulting first-order bias. The estimator retains the invertibility requirement on $(I-H_{ii})$ and adds no common-correlation or equal-size pooling requirement.

\subsubsection{First-order bias and consistency}

Proposition~\ref{prop:bias} gives the first-order bias of the middle matrix for the uncorrected sandwich, the full inverse correction and the proposed estimator.

\begin{proposition}[Finite-sample bias comparison]
\label{prop:bias}
Under (R0)--(R7), with fixed $p$ and correctly specified working covariance $V_i = \Cov(y_i)$, let $\hat M$ denote a linearized middle matrix based on the first-order residual expansions, with matrices evaluated at $\beta_0$. Define $B(\hat{M}) = \E[\hat{M}(\beta_0)] - M_0$ and $B_{\mathrm{lev}}^{(0)} = \sum_i A_i\Delta A_i$. Then
\begin{enumerate}
\item[(i)] $\hat{V}_{LZ}$: $\E[\hat{M}_{LZ}] - M_0 = -B_{\mathrm{lev}}^{(0)} + O(N^{-1})$, negative semidefinite at leading order.
\item[(ii)] $\hat{V}_{MD}$: $\E[\hat{M}_{MD}] - M_0 = +B_{\mathrm{lev}}^{(0)} + O(N^{-1})$, positive semidefinite at leading order.
\item[(iii)] $\VarAR$: $\E[\hat{M}_{AR}] - \E[\hat{M}_{MD}] = (c_N - 1)\E[\hat{M}_{MD}] + O(N^{-1})$, positive definite at leading order, where $c_N = \frac{n^\star - 1}{n^\star - p}\cdot\frac{N}{N-1}$ is the product of the FPC and Bessel factors.
\end{enumerate}
In each case, $B = O(1)$ while $M_0 = O(N)$, so the relative bias is $O(N^{-1})$.
\end{proposition}

Scaling raises the mean of the middle matrix by the factor $c_N$, and centering leaves it unchanged at leading order, because its expected contribution is $O(N^{-1})$ against a middle matrix of order $N$. Test calibration in small samples also depends on the variability of the variance estimate, so the simulation study measures what each component contributes, using the combinations defined below.

\begin{proposition}[Consistency]
\label{prop:consistency}
Under (R0)--(R7), with $p$ fixed, additionally assume equal subject sizes and a common true standardized covariance $W_i(\beta_0)^{-1/2}\Cov(y_i)W_i(\beta_0)^{-1/2}=C_0$ for the pooling estimators $\hat V_{\mathrm{Pan}}$, $\hat V_{GST}$, $\hat V_{WL}$, $\hat V_{WB}$ and $\hat V_{RS}$. Then, for each candidate $\hat{V} = I_0(\betahat)^{-1}\hat{M}(\betahat)I_0(\betahat)^{-1}$ under its respective conditions, $N \hat{V} \plim \Gamma_0^{-1}\Sigma_0\Gamma_0^{-1}$. In particular, $\VarAR$ is consistent under (R0)--(R7).
\end{proposition}

The residual expansions \eqref{eq:residual-expansion}--\eqref{eq:pgee-residual} are given in Web Appendix~\ref{app:proofs}, Section~\ref{sec:bias-problem}, and the proofs of Propositions~\ref{prop:bias} and~\ref{prop:consistency} in Section~\ref{sec:proofs}. Under their respective conditions, the candidates converge to the same target.

\subsubsection{Combinations used in the component comparison}

To distinguish the effects of centering and scaling, we apply four combinations of centering and scaling to each of the uncorrected, half-power and full inverse score forms,
\begin{equation}
g_i^{\mathrm{unc}}=d_i, \qquad
g_i^{\mathrm{half}}=D_i^TV_i^{-1}(I-H_{ii})^{-1/2}r_i, \qquad
g_i^{\mathrm{full}}=f_i.
\label{eq:ingredient-score-forms}
\end{equation}
For any of these forms, write $\bar g=N^{-1}\sum_i g_i$ and define
\begin{equation}
\begin{aligned}
\hat V_g^{(0,0)}&=\Delta\left[\sum_i g_i g_i^T\right]\Delta,\\
\hat V_g^{(1,0)}&=\Delta\left[\sum_i(g_i-\bar g)(g_i-\bar g)^T\right]\Delta,\\
\hat V_g^{(0,1)}&=c_N\hat V_g^{(0,0)},\\
\hat V_g^{(1,1)}&=c_N\hat V_g^{(1,0)}.
\end{aligned}
\label{eq:ingredient-versions}
\end{equation}
The first index indicates centering and the second scaling. Without centering or scaling, the three score forms recover $\hat V_{LZ}$, $\hat V_{KC}$ and $\hat V_{MD}$, respectively. The full inverse form with both components is $\VarAR$. All combinations use the same PGEE fit, and the same $c_N$ when scaling is applied.

\subsection{Software availability}
\label{sec:software}

The analyses were conducted in \textsf{R} 4.5.3 \cite{rcore2026r}. The PGEE implementation is available as the public GitHub package \texttt{pgeeVar} \cite{afiaz2026pgeevar} (\url{https://github.com/awanafiaz/pgeeVar}). The package provides PGEE fitting, correlated binary-data generation, and all fourteen variance estimators considered in this paper (the thirteen literature estimators plus $\VarAR$). The current implementation is scoped to repeated-measure subjects with at least two observations per subject, so singleton subjects are excluded. The core package has no external system dependencies. At the R level it relies only on \texttt{MASS} and \texttt{utils}, both included in a standard \textsf{R} installation.

\section{Simulation and results}
\label{sec:simulation}

\subsection{Design}
\label{sec:sim-design}

We generate correlated binary longitudinal data using the conditional linear family (CLF) method of Qaqish\cite{qaqish2003family}. Each dataset follows the marginal model
\begin{equation*}
\text{logit}(\mu_{ij}) = \beta_0 + \beta_1 x_i + \beta_2 t_{ij},
\quad i = 1, \ldots, N,\; j = 1, \ldots, n_i,
\end{equation*}
where $x_i$ is a subject-level treatment indicator, $t_{ij} = 0.2j$ is time, and the number of treated subjects is fixed at $\operatorname{round}(\gamma N)$ within each scenario. Under the power alternative, $\beta_1 = \log(2)$ and $\beta_2 = 0.2$. Under the null, the relevant coefficient is set to zero. For balanced null designs, the intercept $\beta_0$ is calibrated for each target event rate and number of repeated observations per subject. Power and unbalanced designs use the intercept from the corresponding balanced null design with four visits. True within-subject correlation is exchangeable or AR(1), with matched or deliberately mismatched working correlation. All models are fit using PGEE with the Firth penalty.

The design comprises 148 scenarios organized into nine groups. The primary evaluation uses 72 null scenarios for $\beta_1$, crossing 3 sample sizes ($N \in \{10, 30, 50\}$), 3 event rates ($10\%, 20\%, 30\%$), 4 true-correlation levels ($\rho \in \{0.05, 0.10, 0.20, 0.30\}$), and 2 correlation structures (exchangeable and AR(1)) ($3 \times 3 \times 4 \times 2 = 72$). Twelve scenarios test robustness to working-correlation misspecification ($2 \times 3 \times 2 = 12$ across event rates, sample sizes, and misspecification directions). Six scenarios use unbalanced numbers of repeated observations per subject ($3 \times 2 = 6$ across sample sizes and imbalance patterns). Eight scenarios at $5\%$ event rate with $N \in \{30, 50\}$ and eight at $N = 5$ with $n_i \in \{6, 8\}$ probe the boundaries of the method. Six scenarios with both $\beta_1 = \beta_2 = 0$ evaluate type I error for the time covariate.

Twelve power scenarios cross 3 sample sizes, 2 true-correlation levels, and 2 correlation structures ($3 \times 2 \times 2 = 12$) under $\beta_1 = \log(2)$ (odds ratio 2). We also ran 18 scenarios with two alternative values of the treated-subject proportion $\gamma$ ($\gamma = 0.2$ and $\gamma = 0.5$ instead of the core value $0.3$) across three event rates and three sample sizes to check whether the results change materially under treatment-allocation imbalance. Six scenarios fit the reduced model with $p = 2$ (intercept and treatment only, no time covariate) as a check that $\VarAR$ performs comparably in the simpler setting.

We evaluate fourteen variance estimators, comprising the thirteen literature estimators reviewed in the Methods ($\hat{V}_{LZ}$, $\hat{V}_{DF}$, $\hat{V}_{KC}$, $\hat{V}_{MD}$, $\hat{V}_{FG}$, $\hat{V}_{MBN}$, $\hat{V}_{\text{Pan}}$, $\hat{V}_{GST}$, $\hat{V}_{WL}$, $\hat{V}_{WB}$, $\hat{V}_{RS}$, $\hat{V}_{FW}$, $\hat{V}_{FZ}$) plus our proposed $\VarAR$. Five of the fourteen pool residual covariance across subjects ($\hat{V}_{\text{Pan}}$, $\hat{V}_{GST}$, $\hat{V}_{WL}$, $\hat{V}_{WB}$ and $\hat{V}_{RS}$) and nine do not. The primary metric is the type I error rate, estimated as the number of replications rejecting $H_0$ divided by the number of converged PGEE fits, using a two-sided Wald $t$-test at nominal level $0.05$ with $N-p$ degrees of freedom. Under the alternative, the same proportion estimates power. We use $N-p$ degrees of freedom following the majority of the comparison literature \cite{ford2018comparison,gosho2023comparison,pengli2015small}. In a few converged fits (at most $22$ per estimator over the whole study) an estimator fails numerically and no test can be computed. Such a fit counts as a non-rejection for that estimator, so that every estimator has the same denominator. Structurally inapplicable estimators are marked as unavailable rather than assigned a rate. We also report the median SE/SimSE ratio, the median across replications of the estimated standard error divided by the simulation standard error of $\betastar$, with a target of $1.0$. We use the median to reduce sensitivity to extreme standard errors from the $(I - H_{ii})^{-1}$ correction in near-separated datasets. Estimated-SE summaries use finite, positive standard errors, while simulation standard errors use all converged coefficient estimates. For the null scenarios, 95\% Wald coverage is exactly the complement of the type I error rate, so we do not report coverage separately.

Each scenario uses $B = 5{,}000$ replications. For a type I error rate near $0.05$, the Monte Carlo standard error is approximately $0.0031$ before exclusions and at most $0.0032$ with the observed convergence counts. An estimator whose true type I error rate is $0.05$ can therefore show an estimated rate somewhat above $0.05$ by chance. We count a scenario as above nominal only when its estimated rate exceeds $0.0562$, the nominal rate plus two Monte Carlo standard errors, and we refer to this value as the cutoff.

The component comparison uses the three score forms and the combinations of centering and scaling defined in the Methods. We also compare $N-p$ degrees of freedom with the coefficient-specific degrees of freedom of Pan and Wall\cite{pan2002small}, using the same three score forms. For a coefficient indexed by $s$, write $q_i=(e_s^T\Delta g_i)^2$ for its cluster contribution to the variance, and use
\begin{equation}
\hat V_s=\sum_i q_i,\qquad
\hat\tau_s=\frac{N}{N-1}\sum_i(q_i-\bar q)^2,\qquad
\hat d_s=\frac{2\hat V_s^2}{\hat\tau_s},
\label{eq:pan-wall-reference}
\end{equation}
where $e_s$ is the coordinate vector and $\bar q=N^{-1}\sum_i q_i$. These tests retain the estimated variance and change the reference distribution to $t_{\hat d_s}$.

\subsection{Existing estimators and the proposed estimator}
\label{sec:sim-results}

Table~\ref{tab:simulation-primary} gives the main result. With ten subjects, $\VarAR$ and Morel's estimator are the only estimators that do not pool whose type I error stays below the cutoff in every primary scenario. The uncorrected sandwich is anti-conservative throughout the $N=10$ grid, and the scalar degrees-of-freedom correction reduces but does not remove the excess. Leverage corrections bring the average closer to nominal, although their performance differs across event rates. Morel's estimator ($\hat V_{MBN}$) and the pooling-based hybrids are more conservative at $N=10$. Among the pooling estimators, Pan's unscaled form ($\hat V_{\mathrm{Pan}}$) is closest to nominal on average at each sample size in this grid.

\begin{table}[H]
\centering\small
\caption{Type I error rate for the treatment coefficient in the primary null grid, computed over converged fits, using $N-p$ degrees of freedom. For each number of subjects, the type I error is averaged over the $24$ scenarios, and Above cutoff is the number of those scenarios in which it exceeds $0.0562$, the nominal $0.05$ plus two Monte Carlo standard errors.}\label{tab:simulation-primary}
\begin{tabular}{llcccccc}
\toprule
 & & \multicolumn{2}{c}{$N=10$} & \multicolumn{2}{c}{$N=30$} & \multicolumn{2}{c}{$N=50$} \\
\cmidrule(lr){3-4}\cmidrule(lr){5-6}\cmidrule(lr){7-8}
Group & Estimator & Type I error & Above cutoff & Type I error & Above cutoff & Type I error & Above cutoff \\
\midrule
Baseline & $\hat V_{LZ}$ & 0.1048 & 24 & 0.0675 & 14 & 0.0565 & 9 \\
Baseline & $\hat V_{DF}$ & 0.0670 & 16 & 0.0557 & 7 & 0.0495 & 3 \\
Leverage & $\hat V_{KC}$ & 0.0725 & 16 & 0.0572 & 8 & 0.0503 & 4 \\
Leverage & $\hat V_{MD}$ & 0.0426 & 8 & 0.0490 & 7 & 0.0442 & 2 \\
Leverage & $\hat V_{FG}$ & 0.0473 & 8 & 0.0532 & 7 & 0.0479 & 2 \\
Additive & $\hat V_{MBN}$ & 0.0185 & 0 & 0.0491 & 7 & 0.0451 & 2 \\
Pooling & $\hat V_{Pan}$ & 0.0334 & 0 & 0.0418 & 0 & 0.0454 & 0 \\
Pooling & $\hat V_{GST}$ & 0.0154 & 0 & 0.0325 & 0 & 0.0385 & 0 \\
Hybrids & $\hat V_{RS}$ & 0.0132 & 0 & 0.0319 & 0 & 0.0383 & 0 \\
Hybrids & $\hat V_{WL}$ & 0.0110 & 0 & 0.0291 & 0 & 0.0368 & 0 \\
Hybrids & $\hat V_{WB}$ & 0.0191 & 0 & 0.0354 & 0 & 0.0407 & 0 \\
Hybrids & $\hat V_{FW}$ & 0.0552 & 9 & 0.0525 & 7 & 0.0470 & 2 \\
Hybrids & $\hat V_{FZ}$ & 0.0611 & 14 & 0.0562 & 8 & 0.0498 & 3 \\
Proposed & $\VarAR$ & 0.0336 & 0 & 0.0450 & 7 & 0.0412 & 2 \\
\bottomrule
\end{tabular}
\end{table}

\begin{figure}[tbp]
\centering
\includegraphics[width=\textwidth]{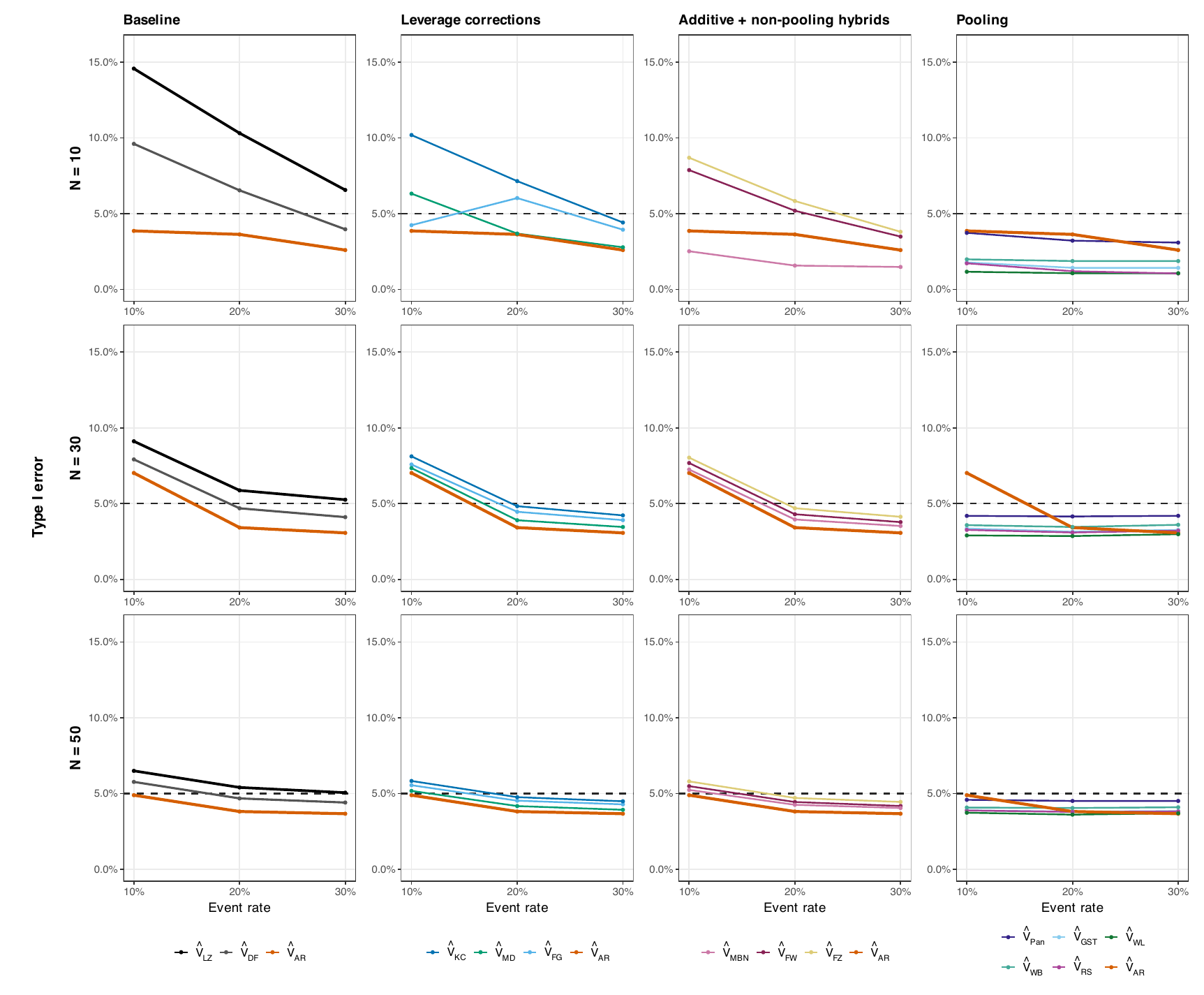}
\caption{Type I error for $\beta_1$ (treatment effect) under $H_0$, by number of subjects $N$ and event rate. Rows fix $N \in \{10, 30, 50\}$; columns group estimators into baseline ($\hat{V}_{LZ}$, $\hat{V}_{DF}$), leverage corrections ($\hat{V}_{KC}$, $\hat{V}_{MD}$, $\hat{V}_{FG}$), additive and non-pooling hybrids ($\hat{V}_{MBN}$, $\hat{V}_{FW}$, $\hat{V}_{FZ}$), and pooling / pooling-based hybrids ($\hat{V}_{\text{Pan}}$, $\hat{V}_{GST}$, $\hat{V}_{WL}$, $\hat{V}_{WB}$, $\hat{V}_{RS}$). $\VarAR$ (orange) appears in every panel for direct comparison. Each point averages the type I error rates of eight scenarios with equal weight, across four true-correlation levels $\rho \in \{0.05, 0.10, 0.20, 0.30\}$ and two correlation structures (exchangeable, AR(1)). Dashed horizontal line marks the nominal level $0.05$.}
\label{fig:sim-typeI}
\end{figure}

All treatment tests in this grid were available after PGEE convergence. Figure~\ref{fig:sim-typeI} shows the same type I error rates by event rate. At $N=10$, seven of the nine estimators that do not pool exceed the cutoff in at least $8$ of the $24$ primary scenarios. $\VarAR$ is closer to nominal than $\hat V_{MBN}$, with a mean type I error rate of $0.0336$ against $0.0185$. At $N=10$ and the $10\%$ event rate, the averages are $0.1457$ for $\hat V_{LZ}$, $0.1019$ for $\hat V_{KC}$, $0.0632$ for $\hat V_{MD}$, $0.0424$ for $\hat V_{FG}$ and $0.0386$ for $\VarAR$. By $N=50$, the corrected-estimator averages are closer to nominal, although $\hat V_{LZ}$ remains above it at $0.0565$. In the primary grid at $N=30$ and $50$, every exceedance of $\hat V_{MD}$, $\hat V_{FG}$, $\hat V_{MBN}$ and $\VarAR$ occurs at the $10\%$ event rate, in $7$ of $8$ scenarios at $N=30$ and $2$ of $8$ at $N=50$, and none of the five pooling-based estimators exceeds the cutoff. Our exploratory analysis links these exceedances to datasets in which one treatment arm has no events (Discussion and Web Appendix~\ref{app:empty-arm-rule}).

\subsection{Contributions of centering, scaling and the reference distribution}

Centering and scaling act in opposite directions at $N=10$ for each score form (Table~\ref{tab:ingredient-rejection}). Centering raises the type I error rate, while scaling lowers it. With full inverse correction, centering increases the rate from $0.0426$ to $0.0544$, scaling alone reduces it to $0.0236$, and both together give $0.0336$. The difference between both components and scaling alone is $0.0100$, with a paired Monte Carlo standard error of $0.0003$.

\begin{table}[H]
\centering\small
\caption{Mean type I error rates for the three score forms and component combinations at $N=10$ in the primary null grid.}\label{tab:ingredient-rejection}
\resizebox{\ifdim\width>\linewidth\linewidth\else\width\fi}{!}{%
\begin{tabular}{lcccc}
\toprule
Score form & Form alone & With centering & With scaling & With both \\
\midrule
Uncorrected & 0.1048 & 0.1210 & 0.0879 & 0.0965 \\
Half power & 0.0725 & 0.0860 & 0.0561 & 0.0687 \\
Full inverse & 0.0426 & 0.0544 & 0.0236 & 0.0336 \\
\bottomrule
\end{tabular}}
\end{table}

All entries use the same PGEE fits, $N-p$ degrees of freedom and the same finite-sample multiplier when scaling is applied. The full inverse form with scaling alone also remains below the cutoff in every primary $N=10$ scenario, with a mean rate of $0.0236$. Adding centering, which gives $\VarAR$, brings the rate closer to nominal, at $0.0336$. At $N=30$ and $50$, centering changes the full inverse rate by only $0.0002$ and $0.0001$ when scaling is already present. Scaling still changes the type I error rate, from $0.0490$ to $0.0448$ at $N=30$ and from $0.0442$ to $0.0412$ at $N=50$ before centering.

The Pan and Wall degrees of freedom increase the type I error rate for all three uncentered score forms at $N=10$ and decrease it at the larger sample sizes (Web Table~\ref{tab:reference-comparison}). For the full inverse form at $N=10$, the average of the scenario-specific median degrees of freedom is $11.5548$, compared with $N-p=7$. The type I error rate rises from $0.0426$ to $0.0563$, and the number of primary scenarios above the cutoff increases from $8$ to $16$ of $24$. At $N=30$ and $50$, the corresponding average median degrees of freedom are $20.3280$ and $28.2994$, below $N-p=27$ and $47$.

\subsection{Working correlation, unequal subject sizes and power}

Under the two working-correlation misspecification directions, every pooled estimator remains below the cutoff in all $12$ scenarios. The proposed estimator exceeds it in the two $N=30$, $10\%$-event scenarios, with a mean type I error rate of $0.0738$, compared with $0.0702$ over the correctly specified primary scenarios at that sample size and event rate. The other misspecified scenarios give mean rates below nominal for $\VarAR$. Among the estimators that do not pool, the number of misspecified scenarios above the cutoff is $9$ for $\hat V_{LZ}$, $7$ for $\hat V_{DF}$, $\hat V_{KC}$ and $\hat V_{FZ}$, $4$ for $\hat V_{MD}$, $\hat V_{FG}$ and $\hat V_{FW}$, and $2$ for $\hat V_{MBN}$ and $\VarAR$.

Unequal numbers of repeated observations per subject remove the five pooling estimators from the comparison. Web Table~\ref{tab:unequal-rejection} reports the leverage and additive corrections under the two unequal-size patterns, with cluster sizes obtained by rounding uniform draws on $[2,6]$ or $[3,8]$, an event-rate target of $20\%$ and $\rho=0.2$. The proposed estimator is conservative at $N=10$ and moves toward nominal as $N$ increases. Fay--Graubard ($\hat V_{FG}$) is closer to nominal in these settings, with mean type I error rates of $0.0482$, $0.0459$ and $0.0463$, but at ten subjects it exceeds the cutoff in $8$ of the $24$ primary scenarios (Table~\ref{tab:simulation-primary}), whereas $\VarAR$ exceeds it in none.

Table~\ref{tab:power-comparison} reports power under $\beta_1=\log(2)$, averaged over the true-correlation levels and structures in the power design. The proposed estimator has power $0.0454$, $0.1858$ and $0.3344$ at $N=10$, $30$ and $50$, respectively. Pan's estimator ($\hat V_{\mathrm{Pan}}$), which is below the cutoff throughout the primary null grid, has corresponding power $0.1010$, $0.2512$ and $0.3852$. Morel's estimator ($\hat V_{MBN}$) has similar power to $\VarAR$ at $N=10$ and somewhat higher power at the larger sample sizes. The highest power, that of $\hat V_{LZ}$, reflects its inflated type I error (Table~\ref{tab:simulation-primary}). Power is comparable only among estimators that control the type I error.

\begin{table}[H]
\centering\small
\caption{Treatment-effect power in the twelve alternative scenarios, with $N-p$ degrees of freedom.}\label{tab:power-comparison}
\resizebox{\ifdim\width>\linewidth\linewidth\else\width\fi}{!}{%
\begin{tabular}{lccc}
\toprule
Estimator & $N=10$ & $N=30$ & $N=50$ \\
\midrule
$\hat V_{LZ}$ & 0.1306 & 0.2542 & 0.3854 \\
$\hat V_{KC}$ & 0.0914 & 0.2257 & 0.3656 \\
$\hat V_{MD}$ & 0.0568 & 0.1987 & 0.3451 \\
$\hat V_{FG}$ & 0.0672 & 0.2115 & 0.3562 \\
$\hat V_{MBN}$ & 0.0443 & 0.2048 & 0.3512 \\
$\hat V_{Pan}$ & 0.1010 & 0.2512 & 0.3852 \\
$\hat V_{GST}$ & 0.0595 & 0.2202 & 0.3626 \\
$\hat V_{WL}$ & 0.0456 & 0.2102 & 0.3553 \\
$\VarAR$ & 0.0454 & 0.1858 & 0.3344 \\
\bottomrule
\end{tabular}}
\end{table}

\subsection{Additional settings, standard-error accuracy and convergence}

At the $5\%$ event-rate target, corrections that do not pool remain anti-conservative. For $N=30$, the mean type I error rate is $0.1708$ for $\VarAR$, $0.1727$ for $\hat V_{MD}$, $0.1705$ for $\hat V_{FG}$ and $0.1561$ for $\hat V_{MBN}$, compared with $0.0500$ for $\hat V_{\mathrm{Pan}}$. The Gosho--Sato--Takeuchi estimator ($\hat V_{GST}$) has mean rates of $0.0419$ and $0.0412$ at $N=30$ and $50$, with no exceedance in the eight $5\%$-event scenarios. At $N=50$, the rate for $\VarAR$ is still $0.1101$ in this regime. This excess again comes from datasets with an empty treatment arm. With thirty subjects, $19.2\%$ of the datasets have one, and $\VarAR$ rejects in $75.1\%$ of them against $3.2\%$ of the datasets with events in both arms. With fifty subjects the figures are $6.8\%$, $94.7\%$ and $4.9\%$ (Web Appendix~\ref{app:empty-arm-rule}). Across all $40$ null scenarios with $N=10$, $\VarAR$ exceeds the cutoff in one and $\hat V_{MBN}$ in none, compared with $12$ for $\hat V_{MD}$, $14$ for $\hat V_{FG}$ and $28$ for $\hat V_{KC}$. We also varied the share of treated subjects, using $20\%$ or $50\%$ in place of $30\%$. $\VarAR$ stayed below the cutoff in $15$ of these $18$ scenarios, and the $3$ exceptions were all at the $10\%$ event rate, as in the main grid. No other estimator that does not pool did better. $\hat V_{MBN}$ also exceeded the cutoff in $3$ scenarios, and the others in $4$ to $13$. With only $N=5$ subjects, $N-p=2$ degrees of freedom produce conservative tests across the estimator groups. With the two-parameter model, $\VarAR$ has type I error rates of $0.0455$, $0.0343$ and $0.0386$ at $N=10$, $30$ and $50$ and exceeds the cutoff in none of the six scenarios. For the time coefficient, in the six scenarios where both coefficients are zero, $\VarAR$ also stays below the cutoff (Web Table~\ref{tab:time-null}).

Across all null scenarios at $N=10$, the median of the scenario-specific median SE/SimSE ratios is $0.7299$ for $\hat V_{LZ}$, $0.8496$ for $\hat V_{KC}$ and $1.0896$ for $\VarAR$. In the primary $10\%$-event settings, the mean of the scenario-specific median SE/SimSE ratios for $\VarAR$ is below one at every sample size, at $0.8291$, $0.8692$ and $0.9276$. PGEE converged in $737{,}077$ of $740{,}000$ generated datasets, with no generator failures. The minimum scenario convergence rate was $93.42\%$, and $141$ of $148$ scenarios exceeded $97\%$. Web Appendix~\ref{app:simulation-full} gives the type I error rate in every scenario, the standard-error summaries and figures by design factor.

\section{Discussion and recommendations}
\label{sec:discussion}

This paper contributes first-order asymptotic grounding for PGEE along convergent interior-root sequences, a first-order matrix characterization of parameter-specific leverage overcorrection under correctly specified working covariance, a variance estimator for small samples that requires no pooling assumption, and a comparison of fourteen variance estimators for PGEE. The comparisons separate what an estimator assumes from how its finite-sample tests behave. Pooling gives a useful default when its common-correlation assumption is appropriate and subject sizes permit it. Among estimators that do not pool, the choice depends on the number of subjects and the event rate. In the toenail sample of the motivating example, the four estimators that gave a significant treatment effect are all among those whose type I error exceeds the cutoff with ten subjects (Table~\ref{tab:simulation-primary}).

The first-order theory says that the full inverse correction overcorrects on average, yet with ten subjects and few events its standard errors are too small in a typical dataset, which is why the proposed estimator scales it up. In the proposed estimator, centering raises the type I error rate at about ten subjects and offsets some of the conservatism introduced by scaling. At the larger sample sizes, centering adds little while scaling still lowers the rate. With about ten subjects, $\VarAR$ and Morel's estimator ($\hat V_{MBN}$) were the only two of the nine estimators that do not pool whose type I error stayed below the cutoff throughout the main simulations, and $\VarAR$ was closer to the nominal level than $\hat V_{MBN}$. Earlier PGEE studies recommended Morel's correction \cite{mondol2019bias,gosho2023comparison} without explaining why it performs well. Our component comparison gives the reason. Morel's estimator applies centering and scaling to the uncorrected scores and adds a positive-definite stabilizer. Without the stabilizer, that combination has a type I error of $0.0965$ at ten subjects (Table~\ref{tab:ingredient-rejection}), and with it $0.0185$. Its control therefore comes from the stabilizer, which makes it conservative, whereas $\VarAR$ reaches control through leverage correction with centering and scaling and stays closer to the nominal level.

Our exploratory analysis examined where the excess type I error in designs with few events comes from. In each scenario we separated the simulated datasets in which one treatment arm had no events from those in which both arms had events. With thirty subjects and a $10\%$ event rate, only $4.7\%$ of the datasets had an empty arm, yet each of the nine estimators that do not pool rejected in at least $94\%$ of them, against at most $4.8\%$ of the datasets in which both arms had events (Web Table~\ref{tab:empty-arm-split}). In these scenarios the excess type I error therefore comes from the datasets with an empty arm. As a check, we tested the treatment effect with the pooled estimator $\hat V_{\mathrm{Pan}}$, which draws on the residuals of all subjects, whenever a dataset had an empty arm, and with $\VarAR$ otherwise. We did this in the $130$ null scenarios with equal numbers of observations per subject, where $\hat V_{\mathrm{Pan}}$ can be computed. The number of scenarios with a type I error above the cutoff then fell from $22$ to $1$ (Web Appendix~\ref{app:empty-arm-rule}).

Guidance developed for ordinary GEE does not carry over directly to PGEE in very small samples. Austin\cite{austin2024comparison} suggested Kauermann--Carroll ($\hat V_{KC}$) for standard errors with fewer than $15$ clusters and Mancl--DeRouen ($\hat V_{MD}$) with a $t$-distribution for confidence intervals in ordinary logistic GEE. For PGEE with ten subjects, $\hat V_{KC}$ exceeded the cutoff in $16$ of the $24$ main scenarios and $\hat V_{MD}$ in $8$. Our simulations also cover a range that has been little studied. Hemming et al.\cite{hemming2025performance} found that only one of fourteen simulation studies with binary outcomes considered a prevalence below $10\%$, whereas the scenarios here include event rates of $5\%$ and $10\%$.

We study Wald inference with fixed $p$ and bounded numbers of repeated observations per subject under growing-$N$ GEE asymptotics. However, growing-dimensional covariates, irregular visit processes, missingness mechanisms, and multilevel or crossed clustering require separate theory. In all our simulations the subjects share one true within-subject correlation structure. The pooled estimators were therefore evaluated only where their assumption holds, and their behavior when it fails was not studied. The treatment covariate was binary and fixed for each subject, so the recommendations below apply to that setting. Exact, permutation, and bootstrap inference may be preferable in the most sparse designs, but they raise different questions about exchangeability, nuisance-parameter refitting, and computational stability under PGEE.

\phantomsection\label{sec:recommendations}
Based on these results, we suggest choosing the variance estimator for treatment-effect inference with PGEE first by whether pooling is justified, and then by the number of subjects and the number of events. When subjects have equal numbers of observations and a common within-subject correlation structure is a reasonable assumption, Pan's pooled estimator ($\hat V_{\mathrm{Pan}}$) with the $N-p$ Wald $t$-test is a suitable default. When pooling is unavailable and the study has about ten subjects, we recommend $\VarAR$ with $N-p$ degrees of freedom. Morel's estimator ($\hat V_{MBN}$) also controls the type I error there but is more conservative. When pooling is unavailable, the study has at least a moderate number of subjects and events are not rare, the leverage-corrected estimators, Morel's estimator and $\VarAR$ behave alike, with $\hat V_{FG}$ the least conservative. With very few events, and in particular when a treatment arm has no events, a leverage correction alone may not make the Wald test reliable, and when pooling is unavailable these simulations do not identify a reliable Wald-test choice. When pooling is justified in that regime, $\hat V_{GST}$ is a more conservative option than $\hat V_{\mathrm{Pan}}$ and remained below the cutoff in the studied scenarios. Together with the theory, these comparisons give practitioners a basis for choosing a variance estimator for PGEE in small studies with rare events.

\section*{Author contributions}

\textbf{AA}: Conceptualization, Methodology, Formal Analysis, Investigation, Software,
Writing - Original Draft, Writing - Review \& Editing, Visualization; \textbf{MSR}:
Conceptualization, Methodology, Supervision, Writing - Review \& Editing.

\section*{Acknowledgments}

Part of this work was completed as part of Awan Afiaz's thesis research. We
thank the authors of the \texttt{geefirthr} and \texttt{binarySimCLF} packages
for making their R scripts available.

\noindent\textbf{Funding.} No funding was received for this work.

\noindent\textbf{Conflict of interest.} The authors declare no conflict of interest.

\noindent\textbf{Generative AI disclosure.} We utilized Generative AI (Anthropic's Claude Opus 4.5/4.6 and Claude Fable 5.2, and OpenAI's
GPT-5.4) in the production of this manuscript in the following ways: (1)
iteratively improving the concision and clarity of the writing, and (2)
producing and debugging computer code and figure generation. We have carefully
reviewed all aspects of the manuscript for accuracy and coherence. All
scientific insights, analysis, interpretation of data, and conclusions are made
solely by the authors. All errors are our own. This disclosure is adapted from
\href{https://www.econstor.eu/bitstream/10419/307214/1/dp17390.pdf}{Tyler
Ransom (2025)}.

\section*{Data and code availability}

All simulated data in this paper can be generated with the public
\texttt{pgeeVar} package \cite{afiaz2026pgeevar} (\url{https://github.com/awanafiaz/pgeeVar}). The
real-data examples use the toenail dataset from the \texttt{geesmv} package \cite{wang2015geesmv} and
the CBPP dataset from the \texttt{lme4} package \cite{bates2015lme4}. Code for the package, the
worked examples, and the accompanying vignettes is available in the same public
repository.

\section*{Supplementary Material}
Web Appendices~\ref{app:existing} to~\ref{app:empty-arm-rule} and the Web Tables and Web Figures referenced in the main text are included at the end of this preprint. The public \texttt{pgeeVar} R package provides the PGEE implementation, worked examples, and vignettes at \url{https://github.com/awanafiaz/pgeeVar}.

\bibliographystyle{unsrtnat}
\bibliography{refs}

\clearpage
\begin{center}
{\Large\bfseries Web Appendices for\par}
{\large\bfseries Overstuffed sandwiches and separation anxiety: finite-sample variance estimation for penalized GEE with near-separated binary data\par}
\vspace{1em}
Awan Afiaz and M.\ Shafiqur Rahman
\end{center}
\clearpage
\setcounter{section}{0}
\setcounter{table}{0}
\setcounter{figure}{0}
\renewcommand{\thesection}{\Alph{section}}
\newcommand{\webappendix}[1]{\clearpage\refstepcounter{section}\setcounter{table}{0}\setcounter{figure}{0}\section*{Web Appendix~\thesection. #1}}
\renewcommand{\theHsection}{supp.\arabic{section}}
\renewcommand{\theHsubsection}{supp.\arabic{section}.\arabic{subsection}}
\renewcommand{\theHsubsubsection}{supp.\arabic{section}.\arabic{subsection}.\arabic{subsubsection}}
\renewcommand{\tablename}{Web Table}
\renewcommand{\figurename}{Web Figure}
\renewcommand{\thetable}{\thesection\arabic{table}}
\renewcommand{\thefigure}{\thesection\arabic{figure}}
\renewcommand{\theHtable}{web.\thesection.\arabic{table}}
\renewcommand{\theHfigure}{web.\thesection.\arabic{figure}}
\numberwithin{equation}{section}
\renewcommand{\theequation}{\thesection-\arabic{equation}}
\renewcommand{\theHequation}{web.\thesection.\arabic{equation}}
\webappendix{Existing variance estimators}
\label{app:existing}

This appendix gives the formulas and correction mechanisms for the thirteen existing estimators reviewed in the Methods of the main paper, in the order of Table~\ref{tab:assumptions} there. Each entry starts from the GEE sandwich in equation~\eqref{eq:lz} of the main paper and states the correction that the estimator applies. The notation is that of the main paper. Write $I_1 = \sum_{i=1}^N d_i d_i^T$, where $d_i = D_i^T V_i^{-1} r_i$ is the score contribution of subject $i$.

\paragraph{A.1\quad $\hat{V}_{LZ}$ (Liang--Zeger)}
The base sandwich estimator is
\[
\hat{V}_{LZ} = \Delta\bigl[\sum_{i=1}^N d_i d_i^T\bigr]\Delta .
\]
It applies no finite-sample correction. It is consistent as $N \to \infty$ under (R0)--(R7) but downward biased in finite samples because $\E[r_i r_i^T] \neq \Cov(y_i)$.

\paragraph{A.2\quad $\hat{V}_{DF}$ (MacKinnon--White degrees of freedom)}
\begin{equation}
\hat{V}_{DF} = \frac{N}{N-p}\,\hat{V}_{LZ}.
\end{equation}
This is a scalar inflation analogous to Bessel's correction for the sample variance. It removes the $O(p/N)$ leading bias term but does not account for subjects having different leverages.

\paragraph{A.3\quad $\hat{V}_{KC}$ (Kauermann--Carroll)}
The residual-level leverage family can be written as
\begin{equation}
\hat{V}(c) = \Delta \left[\sum_{i=1}^N D_i^T V_i^{-1} (I-H_{ii})^{-c} r_i r_i^T (I-H_{ii}^T)^{-c} V_i^{-1} D_i\right] \Delta, \qquad c \in [0,1].
\label{eq:vfamily}
\end{equation}
The Kauermann--Carroll estimator \cite{kauermann2001note} takes $c=1/2$. It corrects each subject's residual outer product by the matrix square root of $(I - H_{ii})^{-1}$, and it requires $(I - H_{ii})$ to be invertible, which can fail when eigenvalues of $H_{ii}$ approach one.

\paragraph{A.4\quad $\hat{V}_{MD}$ (Mancl--DeRouen)}
The Mancl--DeRouen estimator \cite{mancl2001covariance} takes $c=1$ in \eqref{eq:vfamily}, applying the full $(I-H_{ii})^{-1}$ correction. Theorem~\ref{thm:overcorrection} of the main paper gives its first-order overcorrection matrix under correctly specified working covariance.

\paragraph{A.5\quad $\hat{V}_{FG}$ (Fay--Graubard)}
\begin{equation}
\hat{V}_{FG} = \Delta\left[\sum_{i=1}^N F_i\, d_i\, d_i^T\, F_i\right]\Delta,
\end{equation}
where $F_i = \diag\bigl\{(1 - \min(b,\, [L_i]_{ss}))^{-1/2}\bigr\}$ is a $p \times p$ diagonal matrix, $L_i = D_i^T V_i^{-1} D_i\,\Delta$, and $b = 0.75$ is the default clipping threshold \cite{fay2001small}. Unlike $\hat V_{KC}$ and $\hat V_{MD}$, the correction operates at the score level ($p \times p$) and not at the residual level ($n_i \times n_i$). Each diagonal entry of $F_i$ inflates the corresponding score component by its estimated leverage.

\paragraph{A.6\quad $\hat{V}_{MBN}$ (Morel--Bokossa--Neerchal)}
The estimator of Morel, Bokossa and Neerchal \cite{morel2003small} combines a centered, scaled score covariance with a positive-definite additive stabilizer,
\begin{align}
I_{1,\mathrm{MBN}}^c
&=
\frac{n^*-1}{n^*-p} \cdot \frac{N}{N-1}
\sum_{i=1}^N (d_i - \bar{d})(d_i - \bar{d})^T, \notag\\
\hat{V}_{MBN}
&=
\Delta I_{1,\mathrm{MBN}}^c \Delta + \kappa \, \delta_n \, \Delta.
\label{eq:vmorel}
\end{align}
Here $\bar{d} = N^{-1}\sum_{i=1}^N d_i$, $\kappa = \max\{1, p^{-1} \tr(\Delta I_{1,\mathrm{MBN}}^c)\}$ is a design-effect estimate, $\delta_n = \min\{0.5, p/(N-p)\}$ is a shrinkage factor, and $n^* = \sum_i n_i$. The first term applies Morel's finite-population and Bessel-like correction to the centered score covariance. The second is an additive stabilizer that vanishes as $N \to \infty$. Morel et al. discuss both trace-based and determinant-based design-effect versions of $\kappa$, and we write the trace form used in their main implementation. Following Morel's regression-dimension convention, we use $p=\dim(\beta)$ in \eqref{eq:vmorel}, including when $\phi$ is estimated. Proposition~\ref{prop:plug-in-phi-morel} in Web Appendix~\ref{app:proofs} describes the scale dependence at fixed regression and working-correlation parameters. Neither the finite-sample factor nor the additive term removes the leverage bias of the sandwich. Both are forms of heuristic regularization.

\paragraph{A.7\quad $\hat{V}_{\mathrm{Pan}}$ (Pan)}
Pan\cite{pan2001robust} introduced the pooled unscaled correlation
\begin{equation}
\hat{R}_u = \frac{1}{N} \sum_{i=1}^N W_i^{-1/2} r_i r_i^T W_i^{-1/2}, \label{eq:Ru}
\end{equation}
and the variance estimator substitutes it for the per-subject residual outer products,
\begin{equation}
\hat{V}_{\mathrm{Pan}} = \Delta\left[\sum_{i=1}^N D_i^T V_i^{-1} W_i^{1/2}\,\hat{R}_u\, W_i^{1/2} V_i^{-1} D_i\right]\Delta.
\end{equation}
This requires a common within-subject correlation structure across subjects. Direct unstructured pooling also requires equal numbers of repeated observations per subject, because the residual outer products must have the same dimensions. These requirements can fail in clinical trials with variable follow-up or in cohorts in which subjects are observed for different lengths of time.

\paragraph{A.8\quad $\hat{V}_{GST}$ (Gosho--Sato--Takeuchi)}
This estimator \cite{gosho2014robust} has the same structure as Pan's, with a degrees-of-freedom adjustment in the pooled correlation,
\begin{equation}
\hat{R}_u^{(DF)} = \frac{1}{N-p}\sum_{i=1}^N W_i^{-1/2}\, r_i\, r_i^T\, W_i^{-1/2}.
\end{equation}
The variance estimator replaces $\hat{R}_u$ with $\hat{R}_u^{(DF)}$ in Pan's sandwich.

\paragraph{A.9\quad $\hat{V}_{RS}$ (Rogers--Stoner)}
\begin{equation}
\hat{V}_{RS} = \hat{V}_{\mathrm{Pan}} + \delta_n \cdot d_{\det} \cdot \Delta,
\end{equation}
where $d_{\det} = \max\bigl\{1,\, |\det(\Delta\, \hat{M}_{\mathrm{Pan}})|^{1/p}\bigr\}$ and $\hat{M}_{\mathrm{Pan}} = \sum_i D_i^T V_i^{-1} W_i^{1/2}\,\hat{R}_u\, W_i^{1/2} V_i^{-1} D_i$ is Pan's middle matrix \cite{rogers2015modification}. This is the determinant-based version of Morel's stabilizer applied to the pooled estimator.

\paragraph{A.10\quad $\hat{V}_{WL}$ (Wang--Long)}
This estimator \cite{wang2011modified} combines Pan's pooling with the Mancl--DeRouen leverage correction inside the pooled average,
\begin{equation}
\hat{R}_u^{(MD)} = \frac{1}{N}\sum_{i=1}^N W_i^{-1/2}(I - H_{ii})^{-1}\, r_i\, r_i^T\, (I - H_{ii}^T)^{-1} W_i^{-1/2}.
\end{equation}
The variance estimator uses $\hat{R}_u^{(MD)}$ in place of $\hat{R}_u$ in Pan's sandwich. It inherits the common-correlation assumption, the equal-size requirement of the direct unstructured implementation, and the invertibility requirement on $(I - H_{ii})$.

\paragraph{A.11\quad $\hat{V}_{WB}$ (Westgate--Burchett)}
This estimator \cite{westgate2016improving} generalizes the pooled estimators with a leverage-correction exponent $c \in [0,1]$,
\begin{equation}
\hat{R}_u^{(c)} = \frac{1}{N}\sum_{i=1}^N W_i^{-1/2}(I - H_{ii})^{-c}\, r_i\, r_i^T\, (I - H_{ii}^T)^{-c}\, W_i^{-1/2}.
\end{equation}
Pan's estimator corresponds to $c = 0$ and Wang--Long to $c = 1$, and $c = 1/2$ gives a pooled estimator of the Kauermann--Carroll type. The variance estimator uses $\hat{R}_u^{(c)}$ in Pan's sandwich.

\paragraph{A.12\quad $\hat{V}_{FW}$ (Ford--Westgate)}
\begin{equation}
\hat{V}_{FW} = \frac{1}{2}\bigl(\hat{V}_{KC} + \hat{V}_{MD}\bigr).
\end{equation}
This is a heuristic average of the two residual-level leverage corrections, studied by Ford and Westgate\cite{ford2018comparison}.

\paragraph{A.13\quad $\hat{V}_{FZ}$ (Fan--Zhang--Zhang)}
\begin{equation}
\hat{V}_{FZ} = \Delta\left[\sum_{i=1}^N D_i^T V_i^{-1}(I - H_{ii})^{-1}
\Bigl\{r_i\, r_i^T - \sum_{j \neq i} H_{ij}\, r_j\, r_j^T\, H_{ij}^T\Bigr\}
(I - H_{ii}^T)^{-1} V_i^{-1} D_i\right]\Delta,
\end{equation}
where $H_{ij} = D_i\,\Delta\, D_j^T V_j^{-1}$ is the cross-subject hat block \cite{fan2012robust}. The inner braces subtract an estimated cross-subject contamination before the $(I - H_{ii})^{-1}$ correction is applied, which gives an estimator approximately unbiased for $\Cov(y_i)$ in Gaussian linear models. The cross-subject summation has $O(N^2)$ cost, and the middle matrix can be indefinite. The estimator was derived for Gaussian linear models and is applied to GEE through the sandwich middle term.

\paragraph{Groups by correction mechanism.}
The main paper groups these estimators by mechanism. $\hat{V}_{LZ}$ and $\hat{V}_{DF}$ are the baseline estimators. $\hat{V}_{KC}$ and $\hat{V}_{MD}$ are residual-level leverage corrections and $\hat{V}_{FG}$ is a score-level leverage correction. $\hat{V}_{MBN}$ adds a positive-definite stabilizer, and $\hat{V}_{\mathrm{Pan}}$ and $\hat{V}_{GST}$ pool across subjects. Among the hybrids, $\hat{V}_{RS}$ applies a stabilizer of Morel's type to a pooled middle matrix, $\hat{V}_{WL}$ and $\hat{V}_{WB}$ combine pooling with leverage correction, and $\hat{V}_{FW}$ and $\hat{V}_{FZ}$ combine leverage corrections without pooling. The proposed $\VarAR$ combines full inverse leverage correction, score centering and finite-sample scaling.

\webappendix{Regularity conditions and proofs}
\label{app:proofs}

\subsection{Regularity conditions}
\label{sec:regularity}

We state conditions for fixed $p$ and growing $N$. These adapt the standard GEE conditions of Xie and Yang\cite{xie2003asymptotics}, Balan and Schiopu-Kratina\cite{balan2005asymptotic}, and Wang and Long\cite{wang2011modified} to our setting.

\begin{enumerate}
\item[(R0)] The numbers of repeated observations per subject satisfy $1 \leq n_i \leq n_{\max}$ for a fixed constant $n_{\max}$. Covariates are uniformly bounded: $\sup_{i,j}\|x_{ij}\| \leq C_x$. Marginal probabilities are bounded: $0 < b_1 \leq \mu_{ij}(\beta_0) \leq b_2 < 1$.
\item[(R1)] The parameter space $\Theta \subset \mathbb{R}^p$ is compact and $\beta_0 \in \mathrm{int}(\Theta)$.
\item[(R2)] $N^{-1}I_0(\beta)$ converges uniformly to a positive definite limit $\Gamma(\beta)$ on $\Theta$. Write $\Gamma_0 = \Gamma(\beta_0)$. In addition, the limiting mean estimating function $\bar U(\beta)=\lim_{N\to\infty}N^{-1}\sum_i \E_{\beta_0}\{U_i(\beta)\}$ has a well-separated zero at $\beta_0$: for every $\epsilon>0$, $\inf_{\beta\in\Theta:\|\beta-\beta_0\|\ge\epsilon}\|\bar U(\beta)\|>0$.
\item[(R3)] $N^{-1}\sum_i D_i^TV_i^{-1}\Cov(y_i)V_i^{-1}D_i$ converges uniformly to a positive definite limit $\Sigma(\beta)$ on $\Theta$. Write $\Sigma_0 = \Sigma(\beta_0)$.
\item[(R4)] The maps $\beta \mapsto D_i(\beta)$, $\beta \mapsto V_i(\beta,\alpha)$, and $\beta \mapsto \mu_i(\beta)$ are twice continuously differentiable on $\Theta$. The derivatives through order two are uniformly bounded and locally equicontinuous in the admissible $(\beta,\alpha)$, uniformly in $i$. The eigenvalues of $V_i(\beta,\alpha)$ are uniformly bounded above and away from zero on this set.
\item[(R5)] $\max_{1 \leq i \leq N} \E[\|U_i(\beta_0)\|^{2+\delta}] \leq C$ for some $\delta > 0$.
\item[(R6)] The final working-correlation estimate $\hat\alpha$ used in the PGEE fit satisfies $\sqrt{N}(\hat\alpha - \alpha_0) = O_p(1)$ as $N \to \infty$, and the map $(\beta, \alpha) \mapsto V_i(\beta, \alpha)^{-1}$ is continuously differentiable with uniformly bounded, locally equicontinuous derivatives on $\Theta \times \mathcal{A}$ for a compact $\mathcal{A}$ containing $\alpha_0$ in its interior.
\item[(R7)] $\max_{1 \leq i \leq N} \tr(H_{ii}(\beta_0)) \to 0$ as $N \to \infty$.
\end{enumerate}

Conditions (R0)--(R5) are standard. For the canonical logistic mean model, bounded covariates and compact $\Theta$ give the mean-derivative bounds in (R4). Its covariance bounds require the working-correlation matrices to remain uniformly positive definite over the admissible parameter set, with scalar dispersion bounded above and away from zero. Under (R0) and (R4), (R5) holds for any $\delta > 0$. The second part of (R2) is the root-identification condition used in Z-estimator consistency arguments. Under working independence it follows from strict concavity of the logistic quasi-likelihood. Under general working correlation we state it explicitly. Condition (R3) is a mild requirement that holds under standard GEE designs with canonical links and bounded covariates. Condition (R6) is the standard plug-in regularity condition for a working-correlation estimator in GEE asymptotics, and we treat it as an assumption here. Condition (R7) is our subject-level formulation of the per-observation leverage condition of Balan and Schiopu-Kratina\cite{balan2005asymptotic}. It ensures $(I - H_{ii})^{-1} \to I$ uniformly and is needed for Proposition~\ref{prop:consistency}. Under the uniform bounds above, $A_i=O(1)$ and (R2) gives $\Delta=O(N^{-1})$, so $\max_i\tr(H_{ii})=O(N^{-1})$ for fixed $p$. We retain (R7) to identify the leverage condition used in the consistency proof.

\subsection{Residual bias decomposition}
\label{sec:bias-problem}

A first-order Taylor expansion of $\hat{\mu}_i$ around $\beta_0$, using \eqref{eq:gee} and the definition of $H_{ii}$, gives
\begin{equation}
r_i \approx (I_{n_i} - H_{ii})e_i - \sum_{l \neq i} H_{il} e_l,
\label{eq:residual-expansion}
\end{equation}
where $H_{il} = D_i \Delta D_l^T V_l^{-1}$ is the cross-subject hat matrix block and $e_i = y_i - \mu_i(\beta_0)$ is the true residual.

Under PGEE, the expansion acquires an additional deterministic term. Because $\hat{\beta}^* - \beta_0 \approx I_0^{-1}[U(\beta_0) + b(\beta_0)]$, we have
\begin{equation}
r_i \approx (I_{n_i} - H_{ii})e_i - \sum_{l \neq i} H_{il}e_l - D_i I_0^{-1} b(\beta_0).
\label{eq:pgee-residual}
\end{equation}
The third term is deterministic and $O(N^{-1})$ (since $I_0^{-1} = O(N^{-1})$ and $\|b\| = O(1)$ by Lemma~\ref{lem:penalty-order}). In the linearized residual second moment, cross-terms between the mean-zero residual terms and $D_iI_0^{-1}b$ vanish, and the quadratic term $D_iI_0^{-1}bb^TI_0^{-1}D_i^T$ is $O(N^{-2})$. This contribution to the second moment is the same order as the hat-matrix approximation $H_{ii}^* - H_{ii} = O(N^{-2})$ noted in the Methods of the main paper. The GEE-based expansion \eqref{eq:residual-expansion} remains valid for PGEE at first order.

Using the first-order expansion with matrices evaluated at $\beta_0$, the residual second moment is approximated by
\begin{equation}
\E[r_i r_i^T] \approx (I - H_{ii})\Cov(y_i)(I - H_{ii})^T + \sum_{l \neq i} H_{il}\Cov(y_l)H_{il}^T.
\label{eq:Errt}
\end{equation}
The first term is self-shrinkage. The factor $(I - H_{ii})$ contracts the residual covariance, making $r_ir_i^T$ systematically smaller than $\Cov(y_i)$. This is the source of the downward bias in $\hat{V}_{LZ}$. The second term is cross-subject contamination. Subject $i$'s fitted values depend on other subjects through the shared parameter estimate. This term is positive semidefinite and partially offsets the self-shrinkage.

Leverage-correction estimators address the self-shrinkage by applying $(I - H_{ii})^{-c}$ to the residuals. Setting $c = 1$ (Mancl-DeRouen) inverts the self-shrinkage exactly, but this correction treats the cross-subject term as zero. Since the dropped term is positive semidefinite, the corrected estimator overshoots. The term $(I-H_{ii})^{-1}r_ir_i^T(I-H_{ii}^T)^{-1}$ has expectation $\Cov(y_i)$ plus a positive semidefinite remainder. This omitted cross-subject term has been noted before. Mancl and DeRouen treated it as negligible under small leverages, and Fay and Graubard later pointed out a typographical error in the same expression \cite{mancl2001covariance,fay2001small}. Our contribution is to quantify the resulting positive semidefinite remainder under PGEE, where high leverage is the regime of interest.

\subsection{Proofs}
\label{sec:proofs}

Each result is restated here for self-contained reading.

\proofheading{Lemma~\ref{lem:penalty-order} (Firth penalty order)}

\textbf{Statement.} Under regularity conditions (R0)--(R4) in Section~\ref{sec:regularity}, the Firth penalty satisfies $\|b(\beta)\| = O(1)$ and $\|\partial b / \partial \beta^T\|_{\text{sp}} = O(1)$ uniformly on $\Theta$, where $\|\cdot\|_{\text{sp}} = \sigma_{\max}(\cdot)$ denotes the spectral norm.

\begin{proof}
By (R2), $\inf_{\beta \in \Theta}\lambda_{\min}(N^{-1}I_0(\beta)) \geq c > 0$ for large $N$, so
\[
\|I_0(\beta)^{-1}\|_{\text{sp}} \leq (cN)^{-1}.
\]
By (R4), the first and second derivatives of $I_0$ are both $O(N)$ uniformly on $\Theta$, so
\[
\left\|\frac{\partial I_0}{\partial\beta_r}\right\|_{\text{sp}} \leq CN,
\qquad
\left\|\frac{\partial^2 I_0}{\partial\beta_s\partial\beta_r}\right\|_{\text{sp}} \leq CN.
\]
For the penalty component $b_r(\beta)$,
\[
\begin{aligned}
|b_r(\beta)|
&= \frac{1}{2}\left|\tr\!\left(I_0^{-1}\frac{\partial I_0}{\partial\beta_r}\right)\right| \\
&\leq \frac{p}{2}\left\|I_0^{-1}\frac{\partial I_0}{\partial\beta_r}\right\|_{\text{sp}} \\
&\leq \frac{p}{2}\|I_0^{-1}\|_{\text{sp}}
\left\|\frac{\partial I_0}{\partial\beta_r}\right\|_{\text{sp}} \\
&\leq \frac{pC}{2c}
= O(1).
\end{aligned}
\]
For $\partial b/\partial\beta^T$, differentiate componentwise:
\[
\begin{aligned}
\frac{\partial b_r}{\partial\beta_s}
&= \frac{1}{2}\tr\!\left[
\frac{\partial}{\partial\beta_s}
\left(I_0^{-1}\frac{\partial I_0}{\partial\beta_r}\right)
\right] \\
&= \frac{1}{2}\tr\!\left(
-I_0^{-1}\frac{\partial I_0}{\partial\beta_s}
I_0^{-1}\frac{\partial I_0}{\partial\beta_r}
+ I_0^{-1}\frac{\partial^2 I_0}{\partial\beta_s\partial\beta_r}
\right).
\end{aligned}
\]
The two terms are bounded separately:
\[
\begin{aligned}
\left\|I_0^{-1}\frac{\partial I_0}{\partial\beta_s}
I_0^{-1}\frac{\partial I_0}{\partial\beta_r}\right\|_{\text{sp}}
&\leq \|I_0^{-1}\|_{\text{sp}}^2
\left\|\frac{\partial I_0}{\partial\beta_s}\right\|_{\text{sp}}
\left\|\frac{\partial I_0}{\partial\beta_r}\right\|_{\text{sp}}
= O(1), \\
\left\|I_0^{-1}\frac{\partial^2 I_0}{\partial\beta_s\partial\beta_r}\right\|_{\text{sp}}
&\leq \|I_0^{-1}\|_{\text{sp}}
\left\|\frac{\partial^2 I_0}{\partial\beta_s\partial\beta_r}\right\|_{\text{sp}}
= O(1).
\end{aligned}
\]
Hence each entry $\partial b_r/\partial\beta_s$ is $O(1)$ uniformly on $\Theta$. Since $p$ is fixed, $\|\partial b/\partial\beta^T\|_{\text{sp}} = O(1)$.
\end{proof}

\proofheading{Theorem~\ref{thm:pgee-asymptotics} (PGEE asymptotic distribution)}

\textbf{Statement.} Under (R0)--(R7) in Section~\ref{sec:regularity}, let $\{\betastar_N\}$ be any sequence of PGEE roots with $\betastar_N \in \mathrm{int}(\Theta)$ for all sufficiently large $N$. Then
\begin{enumerate}
\item[(a)] $\betastar_N \plim \beta_0$;
\item[(b)] $\sqrt{N}(\betastar_N - \beta_0) \dlim N(0, \Gamma_0^{-1}\Sigma_0\Gamma_0^{-1})$,
\end{enumerate}
where $\Gamma_0 = \lim_{N \to \infty} N^{-1} I_0$ and $\Sigma_0 = \lim_{N \to \infty} N^{-1} M_0$. Consistency of the selected root follows from the Z-estimator argument of van~der Vaart\cite{vandervaart1998asymptotic} (Theorem 5.9), applied to the penalized score $U^\star$ under (R1)--(R2) and Lemma~\ref{lem:penalty-order}. The plug-in error from $\hat\alpha$ is absorbed into the first-order expansion under (R6) by the standard GEE plug-in argument. Xie and Yang\cite{xie2003asymptotics} provides the relevant asymptotic template.

\begin{proof}
\textit{Part (a): Consistency.} Define
\[
\bar{U}(\beta) = \lim_{N \to \infty} N^{-1}\sum_i \E[U_i(\beta)].
\]
At $\beta_0$, $\bar{U}(\beta_0) = 0$ with Jacobian $-\Gamma_0 \prec 0$. By the identification part of (R2), $\beta_0$ is a well-separated zero of $\bar U$ on $\Theta$.

By (R0)--(R2), $N^{-1}U(\beta) \plim \bar{U}(\beta)$ pointwise by Chebyshev's inequality. A finite covering argument upgrades the convergence to uniform on $\Theta$. For any $\epsilon > 0$, cover $\Theta$ with finitely many balls of radius $\delta$ (compactness). At each center, pointwise convergence holds. The oscillation
\[
\|N^{-1}U(\beta) - N^{-1}U(\beta_k)\| \leq L\|\beta - \beta_k\|
\]
is controlled by the uniform Lipschitz constant from (R4). A union bound gives
\[
\sup_{\beta \in \Theta}\|N^{-1}U(\beta) - \bar{U}(\beta)\| \plim 0.
\]
The bounded nuisance derivatives in (R6) also give $\sup_{\beta\in\Theta}\|N^{-1}U(\beta,\hat\alpha)-N^{-1}U(\beta,\alpha_0)\|=O_p(\|\hat\alpha-\alpha_0\|)=o_p(1)$. Thus the same uniform limit holds with the final working-correlation estimate. By Lemma~\ref{lem:penalty-order}, whose bounds are uniform over the admissible nuisance parameters, $\|N^{-1}b(\beta)\| \to 0$ uniformly. So $N^{-1}U^*(\beta) \plim \bar{U}(\beta)$ uniformly, and $\betastar_N \plim \beta_0$ by the Z-estimator consistency theorem van~der Vaart\cite{vandervaart1998asymptotic} (Theorem~5.9).

\textit{Part (b): Asymptotic normality.} Hold the final $\hat\alpha$ fixed and linearize $U^*(\betastar_N,\hat\alpha) = 0$ around $\beta_0$ via the mean value theorem, suppressing $\hat\alpha$ in the notation:
\[
0 = U(\beta_0) + b(\beta_0) + J^*(\tilde{\beta})(\betastar_N - \beta_0),
\]
where $J^*(\beta) = \partial U^*/\partial\beta^T$ and $\tilde{\beta}$ denotes componentwise mean-value points between $\betastar_N$ and $\beta_0$, each converging to $\beta_0$ by consistency. Rearranging:
\[
\sqrt{N}(\betastar_N - \beta_0) = \bigl[-N^{-1}J^*(\tilde{\beta})\bigr]^{-1}\bigl[N^{-1/2}U(\beta_0) + N^{-1/2}b(\beta_0)\bigr].
\]

\textit{Leading matrix.} Put $Q_i=D_i^TV_i^{-1}$. At $(\beta_0,\alpha_0)$, column $r$ of the score Jacobian is
\[
[J_U]_{\cdot r}=-[I_0]_{\cdot r}+\sum_i(\partial_r Q_i)e_i.
\]
The second term has mean zero, and its variance after division by $N$ is $O(N^{-1})$ under (R0) and (R4). The uniform derivative bounds and local equicontinuity extend this convergence to the consistent componentwise intermediate parameters and final $\hat\alpha$. Together with (R2) and $N^{-1}\partial b/\partial\beta^T=o(1)$ from Lemma~\ref{lem:penalty-order}, this gives
\[
-N^{-1}J^*(\tilde{\beta}) \plim \Gamma_0.
\]

\textit{Score term.} At the fixed nuisance value $\alpha_0$, $N^{-1/2}U(\beta_0,\alpha_0)$ is a sum of $N$ independent mean-zero vectors. Suppressing $\alpha_0$ in this calculation, the Lyapunov condition holds:
\[
\begin{aligned}
\sum_i \E\bigl[\|N^{-1/2}U_i\|^{2+\delta}\bigr]
&= N^{-(2+\delta)/2}\sum_i \E\bigl[\|U_i\|^{2+\delta}\bigr] \\
&\leq CN^{-\delta/2} \to 0,
\end{aligned}
\]
by (R5). So $N^{-1/2}U(\beta_0) \dlim N(0, \Sigma_0)$ by the Cram\'{e}r--Wold device and (R3).

\textit{Working-correlation plug-in.} A Taylor expansion in $\alpha$ gives
\[
U(\beta_0,\hat\alpha)
= U(\beta_0,\alpha_0)
+ \dot U_\alpha(\beta_0,\tilde\alpha)(\hat\alpha-\alpha_0),
\]
for intermediate nuisance values $\tilde\alpha$. For each fixed $\alpha$ near $\alpha_0$, the summands of $\dot U_\alpha(\beta_0,\alpha)$ have mean zero. Their uniform bounds and local equicontinuity in (R6), together with a finite covering argument, give
\[
\sup_{\alpha\in\mathcal A_0}\|N^{-1}\dot U_\alpha(\beta_0,\alpha)\|\plim0
\]
for a compact neighborhood $\mathcal A_0$ of $\alpha_0$. Since $\tilde\alpha\in\mathcal A_0$ with probability tending to one and $\sqrt N(\hat\alpha-\alpha_0)=O_p(1)$, the normalized plug-in error satisfies
\[
N^{-1/2}\dot U_\alpha(\beta_0,\tilde\alpha)(\hat\alpha-\alpha_0)=o_p(1).
\]
The plug-in argument therefore applies to the final estimate specified in (R6).

\textit{Penalty term.} By Lemma~\ref{lem:penalty-order}:
\[
N^{-1/2}b(\beta_0) = O(N^{-1/2}) \to 0.
\]
By Slutsky's lemma, the result follows.
\end{proof}

\proofheading{Theorem~\ref{thm:overcorrection} (Direction-dependent overcorrection)}

\textbf{Statement.} Under correctly specified working covariance ($V_i = \Cov(y_i)$) and the first-order Taylor expansion \eqref{eq:residual-expansion}, the overcorrection matrix of any estimator applying the full $(I - H_{ii})^{-1}$ correction to the score contributions is
\[
B_{\mathrm{lev}} = \sum_{i=1}^N A_i (I_0 - A_i)^{-1} A_i,
\]
where $A_i = D_i^T V_i^{-1} D_i$ is the per-subject contribution to $I_0 = \sum_i A_i$.

\begin{proof}
All residuals and expectations in this proof refer to the first-order expansion with matrices evaluated at $\beta_0$. From \eqref{eq:residual-expansion}, the corrected score $f_i = D_i^TV_i^{-1}(I-H_{ii})^{-1}r_i$ satisfies
\[
\begin{aligned}
f_i
&\approx D_i^TV_i^{-1}e_i \\
&\quad - D_i^TV_i^{-1}(I-H_{ii})^{-1}\sum_{l \neq i}H_{il}e_l,
\end{aligned}
\]
where the first term is the true score (self-shrinkage cancelled) and the second is a cross-subject remainder. Taking expectation:
\[
\E[f_if_i^T] = D_i^T V_i^{-1}\Sigma_i V_i^{-1}D_i + \text{cross-subject PSD term}.
\]
Under $V_i = \Sigma_i$, the first term simplifies to $D_i^TV_i^{-1}\Cov(y_i)V_i^{-1}D_i$ (the exact target). The overcorrection from cluster $i$ is
\[
\begin{aligned}
[B_{\mathrm{lev}}]_i
&= \sum_{l \neq i} D_i^TV_i^{-1}(I-H_{ii})^{-1}H_{il}\Sigma_lH_{il}^T \\
&\qquad \times (I-H_{ii}^T)^{-1}V_i^{-1}D_i.
\end{aligned}
\]
Apply the push-through identity $(I - UW)^{-1}U = U(I-WU)^{-1}$ with $U = D_i$ and $W = \Delta D_i^TV_i^{-1}$:
\begin{equation}
(I - H_{ii})^{-1}D_i = D_i(I_0 - A_i)^{-1}I_0.
\label{eq:pushthrough}
\end{equation}
Substitution of \eqref{eq:pushthrough} into $[B_{\mathrm{lev}}]_i$, using $H_{il} = D_i\Delta D_l^TV_l^{-1}$ and $V_l = \Sigma_l$:
\[
\begin{aligned}
[B_{\mathrm{lev}}]_i
&= A_i(I_0-A_i)^{-1}I_0\,\Delta
\left(\sum_{l \neq i}A_l\right)\Delta\, I_0(I_0-A_i)^{-1}A_i.
\end{aligned}
\]
The key summation identity is $\sum_{l \neq i}A_l = I_0 - A_i$. Since $I_0\Delta = \Delta I_0 = I_p$,
\[
\begin{aligned}
[B_{\mathrm{lev}}]_i
&= A_i(I_0-A_i)^{-1}(I_0-A_i)\Delta I_0(I_0-A_i)^{-1}A_i \\
&= A_i(I_0-A_i)^{-1}(I_0-A_i)(I_0-A_i)^{-1}A_i \\
&= A_i(I_0 - A_i)^{-1}A_i.
\end{aligned}
\]
The sum over $i = 1, \ldots, N$ yields the result.
\end{proof}

\proofheading{Corollary~\ref{cor:overcorrection} (Scalar overcorrection bound)}

\textbf{Statement.} Under the conditions of Theorem~\ref{thm:overcorrection} with $p = 2$ (intercept and a subject-level binary covariate), $N_{\min} \ge 2$, and common per-subject information contributions within each treatment group (as when subjects within each group have the same number of repeated observations and the same working covariance), the maximum eigenvalue of $M_0^{-1}\E[\hat{M}_{MD}]$ is
\[
\lambda_{\max}\!\left(M_0^{-1} \E\!\left[\hat{M}_{MD}\right]\right) = \frac{N_{\min}}{N_{\min} - 1},
\]
where $N_{\min} = \min(N_0, N_1)$. With $N_{\min} = 3$ (a common clinical trial configuration under PGEE), the overcorrection is 50\%.

\begin{proof}
The middle matrix of $\hat{V}_{MD}$ is
\[
\hat{M}_{MD}
= \sum_i D_i^TV_i^{-1}(I-H_{ii})^{-1}r_ir_i^T(I-H_{ii}^T)^{-1}V_i^{-1}D_i
= \sum_i f_if_i^T,
\]
where $f_i$ is the corrected score \eqref{eq:corrected-score}. Hence $\E[\hat{M}_{MD}] - M_0 = B_{\mathrm{lev}}$ from Theorem~\ref{thm:overcorrection}, and under correct specification $M_0 = I_0$.

Consider $X_i = (1, x_i)$ where $x_i \in \{0,1\}$ with $N_0$ controls and $N_1$ treated subjects. Let $A_0$ and $A_1$ denote the common per-subject information contributions in the control and treated groups. Then $I_0 = N_0A_0 + N_1A_1$ and, by Theorem~\ref{thm:overcorrection},
\[
B_{\mathrm{lev}} = N_0A_0(I_0 - A_0)^{-1}A_0 + N_1A_1(I_0 - A_1)^{-1}A_1.
\]
In the intercept-treatment basis,
\[
A_0 = a_0\begin{pmatrix}1 & 0 \\ 0 & 0\end{pmatrix},
\qquad
A_1 = a_1\begin{pmatrix}1 & 1 \\ 1 & 1\end{pmatrix},
\]
so
\[
A_0(I_0 - A_0)^{-1}A_0 = \frac{A_0}{N_0-1},
\qquad
A_1(I_0 - A_1)^{-1}A_1 = \frac{A_1}{N_1-1}.
\]
Hence
\[
B_{\mathrm{lev}} = \frac{N_0}{N_0-1}A_0 + \frac{N_1}{N_1-1}A_1.
\]
Substituting these forms into $I_0^{-1}B_{\mathrm{lev}}$ gives the lower-triangular matrix
\[
I_0^{-1}B_{\mathrm{lev}} = \begin{pmatrix} \dfrac{1}{N_0 - 1} & 0 \\[6pt] \dfrac{1}{N_1 - 1} - \dfrac{1}{N_0 - 1} & \dfrac{1}{N_1 - 1} \end{pmatrix}.
\]
Because this matrix is lower-triangular, its eigenvalues are exactly its diagonal entries, $1/(N_0-1)$ and $1/(N_1-1)$. Therefore the eigenvalues of $I_0^{-1}\E[\hat{M}_{MD}] = I_p + I_0^{-1}B_{\mathrm{lev}}$ are
\[
\frac{N_0}{N_0 - 1} \quad \text{and} \quad \frac{N_1}{N_1 - 1}.
\]
The maximum is achieved by the minority group:
\[
\lambda_{\max} = \frac{N_{\min}}{N_{\min}-1}.
\]
\end{proof}

\proofheading{Proposition~\ref{prop:bias} (Finite-sample bias comparison)}

\textbf{Statement.} Under (R0)--(R7) in Section~\ref{sec:regularity}, with fixed $p$ and correctly specified working covariance $V_i = \Cov(y_i)$, let $\hat M$ denote a linearized middle matrix based on the first-order residual expansions \eqref{eq:residual-expansion}--\eqref{eq:pgee-residual}, with matrices evaluated at $\beta_0$. Define $B(\hat{M}) = \E[\hat{M}(\beta_0)] - M_0$ and $B_{\mathrm{lev}}^{(0)} = \sum_i A_i\Delta A_i$. Then
\begin{enumerate}
\item[(i)] $\hat{V}_{LZ}$: $\E[\hat{M}_{LZ}] - M_0 = -B_{\mathrm{lev}}^{(0)} + O(N^{-1})$, negative semidefinite at leading order.
\item[(ii)] $\hat{V}_{MD}$: $\E[\hat{M}_{MD}] - M_0 = +B_{\mathrm{lev}}^{(0)} + O(N^{-1})$, positive semidefinite at leading order.
\item[(iii)] $\VarAR$: $\E[\hat{M}_{AR}] - \E[\hat{M}_{MD}] = (c_N - 1)\E[\hat{M}_{MD}] + O(N^{-1})$, positive definite at leading order, where $c_N = \frac{n^\star - 1}{n^\star - p}\cdot\frac{N}{N-1}$ is the product of the FPC and Bessel factors.
\end{enumerate}
In each case, $B = O(1)$ while $M_0 = O(N)$, so the relative bias is $O(N^{-1})$.

\begin{proof}
We use the frozen first-order residuals specified in the proposition. The deterministic PGEE shift contributes $O(N^{-2})$ per subject, hence $O(N^{-1})$ to the middle-matrix expectation. For the mean-zero part of the expansion, we work in the standardised basis $\tilde D_i = V_i^{-1/2}D_i$, in which $A_i = \tilde D_i^T\tilde D_i$ and $\tilde H_{ii} = \tilde D_i\Delta\tilde D_i^T$ is symmetric PSD with eigenvalues in $[0,1)$. Under $V_i = \Sigma_i$, equation~\eqref{eq:Errt} in the standardised basis becomes
\[
\Sigma_i^{-1/2}\E[r_ir_i^T]\Sigma_i^{-1/2} = (I - \tilde H_{ii})^2 + \sum_{l\neq i}\tilde H_{il}\tilde H_{il}^T,
\]
using symmetry of $\tilde H_{ii}$ and the cross-subject block $\tilde H_{il} = \tilde D_i\Delta\tilde D_l^T$.

\textit{(i) $\hat V_{LZ}$.} The per-subject LZ middle contribution, transformed to the standardised basis, equals $\tilde D_i^T[(I - \tilde H_{ii})^2 + \sum_{l\neq i}\tilde H_{il}\tilde H_{il}^T]\tilde D_i$. For the self-shrinkage piece, expand $(I - \tilde H_{ii})^2 = I - 2\tilde H_{ii} + \tilde H_{ii}^2$. Using $\tilde D_i^T\tilde H_{ii}\tilde D_i = A_i\Delta A_i$ and $\tilde D_i^T\tilde H_{ii}^2\tilde D_i = A_i\Delta A_i\Delta A_i$,
\[
\tilde D_i^T(I - \tilde H_{ii})^2\tilde D_i = A_i - 2A_i\Delta A_i + A_i\Delta A_i\Delta A_i.
\]
For the cross-subject piece, $\tilde D_i^T\tilde H_{il}\tilde H_{il}^T\tilde D_i = A_i\Delta A_l\Delta A_i$, so by $\sum_{l\neq i}A_l = I_0 - A_i$ and $\Delta I_0 = I_p$,
\[
\begin{aligned}
\tilde D_i^T\!\sum_{l\neq i}\tilde H_{il}\tilde H_{il}^T\tilde D_i
&= A_i\Delta(I_0 - A_i)\Delta A_i \\
&= A_i\Delta A_i - A_i\Delta A_i\Delta A_i.
\end{aligned}
\]
Adding the two contributions gives $A_i - A_i\Delta A_i$ per subject, and summing,
\[
\E[\hat M_{LZ}] = \sum_i A_i - \sum_i A_i\Delta A_i = I_0 - B_{\mathrm{lev}}^{(0)}.
\]
Since $M_0 = I_0$ under correct specification, the first-order bias is $-B_{\mathrm{lev}}^{(0)}$, negative semidefinite because each $A_i\Delta A_i$ is PSD. The quadratic leverage terms cancel in this linearized calculation. Restoring the deterministic PGEE shift adds $O(N^{-1})$ to the summed expectation, as noted above.

\textit{(ii) $\hat V_{MD}$.} The correction $(I - \tilde H_{ii})^{-1}$ cancels the self-shrinkage factor exactly:
\[
(I - \tilde H_{ii})^{-1}(I - \tilde H_{ii})^2(I - \tilde H_{ii})^{-1} = I.
\]
Therefore
\[
\sum_i\tilde D_i^T(I - \tilde H_{ii})^{-1}(I - \tilde H_{ii})^2(I - \tilde H_{ii})^{-1}\tilde D_i
= \sum_i A_i
= I_0.
\]
The cross-subject piece uses the push-through identity $\tilde D_i^T(I - \tilde H_{ii})^{-1} = (I - A_i\Delta)^{-1}\tilde D_i^T$ (Theorem~\ref{thm:overcorrection}, equation~\eqref{eq:pushthrough}):
\[
\begin{aligned}
\tilde D_i^T(I - \tilde H_{ii})^{-1}\!\sum_{l\neq i}\tilde H_{il}\tilde H_{il}^T(I - \tilde H_{ii})^{-1}\tilde D_i
&= (I - A_i\Delta)^{-1}A_i\Delta(I_0 - A_i)\Delta A_i \\
&\qquad \times (I - \Delta A_i)^{-1}.
\end{aligned}
\]
At leading order $(I - A_i\Delta)^{-1} = I + O(N^{-1})$, so this reduces to $A_i\Delta(I_0 - A_i)\Delta A_i + O(N^{-2}) = A_i\Delta A_i + O(N^{-2})$ per subject. Summing and including the deterministic PGEE shift,
\[
\E[\hat M_{MD}] = I_0 + B_{\mathrm{lev}}^{(0)} + O(N^{-1}),
\]
so the first-order bias is $+B_{\mathrm{lev}}^{(0)}$, positive semidefinite. Corollary~\ref{cor:overcorrection} gives the exact worst-case eigenvalue $N_{\min}/(N_{\min}-1)$ for its mean-zero linearized residuals under the stated $p = 2$ conditions. The general form is $B_{\mathrm{lev}} = \sum_i A_i(I_0 - A_i)^{-1}A_i$ from Theorem~\ref{thm:overcorrection}.

\textit{(iii) $\VarAR$.} From the algebraic identity $\sum_i(f_i - \bar f)(f_i - \bar f)^T = \sum_i f_i f_i^T - N\bar f\bar f^T$ and the definition $\hat M_{MD} = \sum_i f_if_i^T$ with $f_i = D_i^TV_i^{-1}(I - H_{ii})^{-1}r_i$,
\[
\begin{aligned}
\hat M_{AR} &= c_N\bigl[\hat M_{MD} - N\bar f\bar f^T\bigr], \\
\E[\hat M_{AR}] &= c_N\E[\hat M_{MD}] - (c_N/N)\E[S_f S_f^T],
\end{aligned}
\]
where $S_f = \sum_i f_i$. The calculation below establishes both $\E[S_f]=O(1)$ and $\Var(S_f)=O(N^{-1})$ under the PGEE expansion. Together these imply $\E[S_fS_f^T]=O(1)$, so the centering term is $O(N^{-1})$ and negligible relative to $\E[\hat M_{MD}]=O(N)$.

Substituting the PGEE residual expansion \eqref{eq:pgee-residual} for $r_i$ directly into $f_i$,
\begin{align*}
f_i
&= D_i^TV_i^{-1}(I - H_{ii})^{-1}
   \!\left[(I - H_{ii})e_i - \sum_{l\neq i}H_{il}e_l\right] \\
&\quad - D_i^TV_i^{-1}(I-H_{ii})^{-1}D_i\Delta b(\beta_0) \\
&= \tilde d_i - (I - A_i\Delta)^{-1}A_i\Delta\!\sum_{l\neq i}\tilde d_l - Q_i b(\beta_0),
\end{align*}
where $\tilde d_l = D_l^TV_l^{-1}e_l$ are independent subject-level scores with $\Cov(\tilde d_l) = A_l$, $b(\beta_0)$ is the Firth penalty evaluated at $\beta_0$, and we used the push-through identity and $D_i^TV_i^{-1}H_{il} = A_i\Delta D_l^TV_l^{-1}$. Let $Q_i = (I - A_i\Delta)^{-1}A_i\Delta$, $Q = \sum_i Q_i$, and $S_{\tilde d} = \sum_l\tilde d_l$. Then
\[
\begin{aligned}
S_f
&= \sum_i\tilde d_i - \sum_i Q_i\!\sum_{l\neq i}\tilde d_l - Qb(\beta_0) \\
&= (I - Q)S_{\tilde d} + \sum_l Q_l\tilde d_l - Qb(\beta_0).
\end{aligned}
\]
Expanding $Q_l = A_l\Delta + A_l\Delta A_l\Delta + O(N^{-3})$ and summing over $l$ gives $Q = I_p + B_{\mathrm{lev}}^{(0)}\Delta + O(N^{-2}) = I_p + O(N^{-1})$, so $I - Q = O(N^{-1})$. Under independence of subjects,
\begin{align*}
\Var\bigl((I - Q)S_{\tilde d}\bigr) &= (I - Q)I_0(I - Q)^T = O(N^{-1}),\\
\Var\bigl(\textstyle\sum_l Q_l\tilde d_l\bigr) &= \sum_l Q_l A_l Q_l^T = O(N^{-1}),
\end{align*}
and the cross-covariance is $O(N^{-1})$. The penalty term is deterministic, so it does not change the variance. It does change the mean:
\[
\E[S_f] = -Qb(\beta_0) = -b(\beta_0) + O(N^{-1}) = O(1),
\]
by Lemma~\ref{lem:penalty-order}. Therefore
\[
\E[S_f S_f^T] = \Var(S_f) + \E[S_f]\E[S_f]^T = O(N^{-1}) + O(1) = O(1).
\]

Combining,
\[
\E[\hat M_{AR}] = c_N\E[\hat M_{MD}] + O(N^{-1}),
\]
and the bias excess over $\hat M_{MD}$ is
\[
\E[\hat M_{AR}] - \E[\hat M_{MD}] = (c_N - 1)\E[\hat M_{MD}] + O(N^{-1}),
\]
which is strictly positive definite at leading order because $c_N > 1$ (FPC $> 1$ for $p \ge 2$ and Bessel $> 1$ for $N \ge 2$) and $\E[\hat M_{MD}] \succ 0$. Sandwiching by $\Delta$,
\[
\E[\VarAR] - \E[\hat V_{MD}] = (c_N - 1)\E[\hat V_{MD}] + O(N^{-3}) \succ 0.
\]
At $N = 10$, $n_i = 4$, $p = 3$, the explicit value is $c_N - 1 = (39/37)(10/9) - 1 \approx 0.171$, corresponding to a leading-order SE inflation of $\sqrt{c_N} - 1 \approx 0.082$.
\end{proof}

\proofheading{Proposition~\ref{prop:consistency} (Consistency)}

\textbf{Statement.} Under (R0)--(R7) in Section~\ref{sec:regularity}, with $p$ fixed, additionally assume equal subject sizes and a common true standardized covariance $W_i(\beta_0)^{-1/2}\Cov(y_i)W_i(\beta_0)^{-1/2}=C_0$ for the pooling estimators Pan, GST, WL, WB, and RS. Then, for each candidate $\hat{V} = I_0(\betahat)^{-1}\hat{M}(\betahat)I_0(\betahat)^{-1}$ under its respective conditions, $N \hat{V} \plim \Gamma_0^{-1}\Sigma_0\Gamma_0^{-1}$. In particular, $\VarAR$ is consistent under (R0)--(R7).

\begin{proof}
We first prove the result for $\VarAR$. Write $\theta_0=(\beta_0,\alpha_0)$ and $\hat\theta=(\betahat,\hat\alpha)$. Consistency follows from Theorem~\ref{thm:pgee-asymptotics} and (R6).

\textit{Step 1 (leverage correction vanishes).} With $\Delta=I_0^{-1}$, define
\[
\widetilde H_{ii}=V_i^{-1/2}D_i\Delta D_i^TV_i^{-1/2}.
\]
This matrix is symmetric positive semidefinite and $H_{ii}=V_i^{1/2}\widetilde H_{ii}V_i^{-1/2}$, so
\[
\|\widetilde H_{ii}\|_{\mathrm{sp}}\leq\tr(H_{ii}),\qquad
\|H_{ii}\|_{\mathrm{sp}}\leq C\tr(H_{ii})
\]
by the uniform covariance bounds in (R4). The bounds on $D_i$, $V_i^{-1}$, and $N\Delta$ give $\max_i\|H_{ii}\|_{\mathrm{sp}}=O(N^{-1})$ locally around $\theta_0$, including at $\hat\theta$ with probability tending to one. The Neumann series then gives
\[
\max_i\|(I-H_{ii})^{-1}-I\|_{\mathrm{sp}}=o_p(1).
\]
Consequently, the bounded scores $d_i(\theta)=D_i^TV_i^{-1}\{y_i-\mu_i(\beta)\}$ and corrected scores satisfy $\max_i\|f_i(\hat\theta)-d_i(\hat\theta)\|=o_p(1)$.

\textit{Step 2 (score moments at the truth and at the fit).} At $\theta_0$, the $d_i$ are independent with mean zero and uniformly bounded moments. By (R3) and Chebyshev's inequality,
\[
N^{-1}\sum_i d_i(\theta_0)d_i(\theta_0)^T\plim\Sigma_0,
\qquad N^{-1}\sum_i d_i(\theta_0)\plim0.
\]
The bounded derivatives in (R4) and (R6) give a local Lipschitz bound for $d_i$ and $d_id_i^T$. Thus
\[
\left\|N^{-1}\sum_i\{d_i(\hat\theta)d_i(\hat\theta)^T-d_i(\theta_0)d_i(\theta_0)^T\}\right\|
\leq C\|\hat\theta-\theta_0\|=o_p(1),
\]
and the same argument applies to their average.

\textit{Step 3 (centering and sandwich).} Steps 1--2 imply $N^{-1}\sum_i f_if_i^T\plim\Sigma_0$ and $\bar f=o_p(1)$ at the fit. Since $c_N\to1$,
\[
N^{-1}\hat M_{AR}
=c_N\left\{N^{-1}\sum_i f_if_i^T-\bar f\bar f^T\right\}\plim\Sigma_0.
\]
Together with $N^{-1}I_0(\hat\theta)\plim\Gamma_0$, this proves
\[
N\VarAR\plim\Gamma_0^{-1}\Sigma_0\Gamma_0^{-1}.
\]
The other non-pooling corrections have the same limit. Their leverage adjustments tend to the identity and their scalar multipliers tend to one. The MBN additive term is $O_p(N^{-2})$. For FZ, each cross-subject hat block is $O_p(N^{-1})$. Summing its quadratic contributions gives an $O_p(1)$ subtraction from an $O_p(N)$ middle matrix.

For the pooling estimators, the added assumptions give a common finite-dimensional target $C_0$. At the truth, the independent standardized residual outer products have expectation $C_0$. Their average converges to $C_0$. The same local smoothness argument transfers this convergence to the fitted residuals and weights. Thus $\hat R_u\plim C_0$. If $K_i=D_i^TV_i^{-1}W_i^{1/2}$, the normalized difference between the pooled middle matrix and its target is bounded by
\[
\left(N^{-1}\sum_i\|K_i\|^2\right)\|\hat R_u-C_0\|=o_p(1).
\]
The GST multiplier tends to one, the WL and WB leverage adjustments tend to the identity, and the RS additive stabilizer vanishes after multiplication by $N$. This establishes the pooled candidates' common robust limit under their additional conditions.
\end{proof}

\proofheading{Proposition~\ref{prop:plug-in-phi-morel} (Plug-in scalar dispersion estimation and Morel's correction)}

\begin{proposition}
\label{prop:plug-in-phi-morel}
Suppose the working covariance factors as
\[
V_i(\beta,\alpha,\phi)=\phi\,\widetilde V_i(\beta,\alpha), \qquad \phi>0,
\]
with scalar scale parameter $\phi$. For fixed $(\beta,\alpha)$, define
\[
I_0(\phi)=\sum_{i=1}^N D_i^T V_i(\phi)^{-1}D_i, \qquad \Delta(\phi)=I_0(\phi)^{-1},
\]
\[
H_{ii}(\phi)=D_i\Delta(\phi)D_i^T V_i(\phi)^{-1},
\]
and
\[
d_i(\phi)=D_i^T V_i(\phi)^{-1}r_i.
\]
Let
\[
I_{1,\mathrm{MBN}}^c(\phi)=
\frac{n^*-1}{n^*-p}\cdot\frac{N}{N-1}
\sum_{i=1}^N (d_i(\phi)-\bar d(\phi))(d_i(\phi)-\bar d(\phi))^T,
\qquad
\bar d(\phi)=N^{-1}\sum_{i=1}^N d_i(\phi).
\]
Then:
\begin{enumerate}
\item[(i)] $I_0(\phi)=\phi^{-1}\widetilde I_0$, $\Delta(\phi)=\phi\,\widetilde\Delta$, and $H_{ii}(\phi)=\widetilde H_{ii}$.
\item[(ii)] $d_i(\phi)=\phi^{-1}\widetilde d_i$, hence $I_{1,\mathrm{MBN}}^c(\phi)=\phi^{-2}\widetilde I_{1,\mathrm{MBN}}^c$.
\item[(iii)] The multiplicative Morel term is invariant:
\[
\Delta(\phi)\,I_{1,\mathrm{MBN}}^c(\phi)\,\Delta(\phi)
=
\widetilde\Delta\,\widetilde I_{1,\mathrm{MBN}}^c\,\widetilde\Delta.
\]
\item[(iv)] If
\[
\kappa(\phi)=\max\left\{1,\frac{1}{p}\tr\bigl(\Delta(\phi)I_{1,\mathrm{MBN}}^c(\phi)\bigr)\right\},
\qquad
\delta_n=\min\left\{0.5,\frac{p}{N-p}\right\},
\]
then
\[
\kappa(\phi)=\max\left\{1,\phi^{-1}\frac{C}{p}\right\},
\qquad
C=\tr(\widetilde\Delta\,\widetilde I_{1,\mathrm{MBN}}^c),
\]
so $\kappa(\phi)$ itself is not invariant. However,
\[
\kappa(\phi)\Delta(\phi)=\max\left\{\phi,\frac{C}{p}\right\}\widetilde\Delta.
\]
Under (R0)--(R4), the information lower bound and bounded subject contributions give $\widetilde\Delta=O(N^{-1})$ and $\widetilde I_{1,\mathrm{MBN}}^c=O(N)$. Hence $C/p=O(1)$, and therefore
\[
\delta_n\,\kappa(\phi)\Delta(\phi)=O(N^{-2})
\]
for $\phi$ in any fixed neighborhood of $\phi_0>0$.
\end{enumerate}
These identities hold at fixed $(\beta,\alpha)$ and do not determine the finite-sample parameter count. We use $p=\dim(\beta)$ following Morel's regression-dimension convention.
\end{proposition}

\begin{proof}
Since $V_i(\phi)=\phi\,\widetilde V_i$, we have
\[
V_i(\phi)^{-1}=\phi^{-1}\widetilde V_i^{-1}.
\]
Therefore
\[
I_0(\phi)=\sum_{i=1}^N D_i^T V_i(\phi)^{-1}D_i
=\phi^{-1}\sum_{i=1}^N D_i^T\widetilde V_i^{-1}D_i
=\phi^{-1}\widetilde I_0,
\]
which gives $\Delta(\phi)=I_0(\phi)^{-1}=\phi\,\widetilde\Delta$. Then
\[
H_{ii}(\phi)
=
D_i[\phi\,\widetilde\Delta]D_i^T[\phi^{-1}\widetilde V_i^{-1}]
=
D_i\widetilde\Delta D_i^T\widetilde V_i^{-1}
=
\widetilde H_{ii},
\]
proving (i).

Because $r_i=y_i-\hat\mu_i$ does not depend on $\phi$ for fixed $(\beta,\alpha)$,
\[
d_i(\phi)=D_i^T V_i(\phi)^{-1}r_i
=\phi^{-1}D_i^T\widetilde V_i^{-1}r_i
=\phi^{-1}\widetilde d_i.
\]
Centering preserves the common factor $\phi^{-1}$, so
\[
I_{1,\mathrm{MBN}}^c(\phi)=\phi^{-2}\widetilde I_{1,\mathrm{MBN}}^c.
\]
This proves (ii).

Combining (i) and (ii),
\[
\Delta(\phi)I_{1,\mathrm{MBN}}^c(\phi)\Delta(\phi)
=
(\phi\widetilde\Delta)(\phi^{-2}\widetilde I_{1,\mathrm{MBN}}^c)(\phi\widetilde\Delta)
=
\widetilde\Delta\,\widetilde I_{1,\mathrm{MBN}}^c\,\widetilde\Delta,
\]
which proves (iii).

For (iv),
\[
\tr(\Delta(\phi)I_{1,\mathrm{MBN}}^c(\phi))
=
\tr\bigl((\phi\widetilde\Delta)(\phi^{-2}\widetilde I_{1,\mathrm{MBN}}^c)\bigr)
=
\phi^{-1}\tr(\widetilde\Delta\,\widetilde I_{1,\mathrm{MBN}}^c)
=
\phi^{-1}C.
\]
Hence
\[
\kappa(\phi)=\max\left\{1,\phi^{-1}\frac{C}{p}\right\}.
\]
Multiplying by $\Delta(\phi)=\phi\widetilde\Delta$ gives
\[
\kappa(\phi)\Delta(\phi)
=
\max\left\{1,\phi^{-1}\frac{C}{p}\right\}\phi\widetilde\Delta
=
\max\left\{\phi,\frac{C}{p}\right\}\widetilde\Delta.
\]
Under (R0)--(R4), the lower eigenvalue bound for $N^{-1}\widetilde I_0$ gives $\widetilde\Delta=O(N^{-1})$. Bounded subject contributions and finite-population and Bessel factors tending to one give $\widetilde I_{1,\mathrm{MBN}}^c=O(N)$, so
\[
C=\tr(\widetilde\Delta\,\widetilde I_{1,\mathrm{MBN}}^c)=O(1)
\]
and therefore $C/p=O(1)$. Since $\delta_n=O(N^{-1})$, we obtain
\[
\delta_n\,\kappa(\phi)\Delta(\phi)=O(N^{-2})
\]
for $\phi$ in any fixed neighborhood of $\phi_0>0$.

Thus, at fixed $(\beta,\alpha)$, the multiplicative Morel term is invariant to the scalar working scale and the additive term is of higher order under the stated regularity conditions. These identities do not imply invariance of the PGEE fit when the dispersion changes.
\end{proof}

\subsection{Centering identities at a PGEE root}
\label{sec:score-centering-identities}

This section derives the identities used in the Methods of the main paper to motivate score centering. Every quantity is evaluated at the same PGEE root $\betastar$. Write $d_i=D_i^TV_i^{-1}r_i$, $f_i=D_i^TV_i^{-1}(I-H_{ii})^{-1}r_i$, $A_i=D_i^TV_i^{-1}D_i$, $I_0=\sum_i A_i$ and $\Delta=I_0^{-1}$, and assume that $I_0$ and every $I-H_{ii}$ are invertible.

\paragraph{Sum of the uncorrected score contributions.}
The PGEE estimating equation gives
\begin{equation}
\sum_{i=1}^N d_i=-b(\betastar).
\label{eq:supp-score-sum}
\end{equation}
At the root, $U(\betastar)+b(\betastar)=0$ and $r_i=y_i-\mu_i(\betastar)$, so $U(\betastar)=\sum_i d_i$, which proves \eqref{eq:supp-score-sum}. At a GEE root the penalty is absent and the uncorrected contributions sum to zero. At a PGEE root their sum is generally nonzero.

\paragraph{Corrected contributions and their mean.}
The full inverse correction satisfies
\begin{equation}
f_i-d_i=A_i\Delta f_i.
\label{eq:supp-corrected-difference}
\end{equation}
For conformable matrices $P$ and $Q$ with $I-PQ$ invertible, the push-through identity is $Q(I-PQ)^{-1}=(I-QP)^{-1}Q$. Take $Q=D_i^TV_i^{-1}$ and $P=D_i\Delta$, which gives $PQ=H_{ii}$ and $QP=A_i\Delta$. It follows that
\begin{equation}
D_i^TV_i^{-1}(I-H_{ii})^{-1}=(I-A_i\Delta)^{-1}D_i^TV_i^{-1},
\end{equation}
and hence $f_i=(I-A_i\Delta)^{-1}d_i$. Multiplication by $I-A_i\Delta$ proves \eqref{eq:supp-corrected-difference}. Summing over subjects and using \eqref{eq:supp-score-sum} gives
\begin{equation}
\bar f=-\frac{b(\betastar)}{N}+\frac{1}{N}\sum_i A_i\Delta f_i.
\label{eq:supp-corrected-mean}
\end{equation}
The first term is the penalty part and the second is the leverage part. The leverage part can also be nonzero at a GEE root. In the symmetric case $A_i=I_0/N$, equation~\eqref{eq:supp-corrected-mean} reduces to $\bar f=-b/(N-1)$ for $N>1$, which is zero when $b=0$.

\paragraph{Amount removed by centering.}
With $\pi=-b(\betastar)/N$ and $\lambda=N^{-1}\sum_i A_i\Delta f_i$, so that $\bar f=\pi+\lambda$, centering removes
\begin{equation}
\begin{aligned}
\sum_i f_if_i^T-\sum_i(f_i-\bar f)(f_i-\bar f)^T
&=N\bar f\bar f^T\\
&=N\pi\pi^T+N\lambda\lambda^T+N(\pi\lambda^T+\lambda\pi^T).
\end{aligned}
\label{eq:supp-centering-components}
\end{equation}
Expanding the centered outer products proves the first equality, and substituting $\bar f=\pi+\lambda$ proves the second. The amount removed contains the penalty term $bb^T/N$, the leverage term and their cross terms.

The uncentered middle matrix uses score second moments, while the centered matrix uses their empirical covariance, and these identities give a covariance-based motivation for centering at a PGEE fit. Because the fitted contributions share estimated parameters, the independent-observation argument for Bessel's correction does not make $N/(N-1)$ exactly unbiasing. Under the regularity conditions of Section~\ref{sec:regularity}, first-order consistency is preserved, as shown in the proof of Proposition~\ref{prop:consistency} in Section~\ref{sec:proofs}. The effect of centering on the type I error rate in small samples is measured in the simulation study of the main paper.

\webappendix{Applied examples}
\label{app:applied-full}

\subsection{Selection of the toenail samples}
\label{sec:toenail-selection}

Each toenail sample was drawn from the version of the trial data distributed in the \texttt{geesmv} R package \cite{wang2015geesmv}. With a fixed seed, we sampled patient identifiers without replacement from 1 through 224 and retained the sampled patients who had completed all seven visits. Web Table~\ref{tab:toenail-samples} gives the seed, the number of identifiers sampled and the resulting number of patients for the sample used in the main paper and for the three further samples reported in this appendix.

\begin{table}[H]
\centering\small
\caption{Selection of the toenail samples.}\label{tab:toenail-samples}
\begin{tabular}{lcccc}
\toprule
Sample & Seed & Identifiers sampled & Patients & Observations \\
\midrule
Sample of the main paper & 8791 & 20 & 10 & 70 \\
Case II & 21 & 21 & 13 & 91 \\
Case III & 79 & 96 & 60 & 420 \\
Second ten-patient sample & 7375 & 20 & 10 & 70 \\
\bottomrule
\end{tabular}
\end{table}

\subsection{Toenail sample of the main paper}

Web Table~\ref{tab:toenail-case1} gives the standard errors and $p$-values for the treatment and time coefficients in the ten-patient sample of the main paper under all fourteen estimators, including $\VarAR$. The treatment coefficient is $\hat\beta_1=-3.16$ and the time coefficient is $\hat\beta_2=-0.42$. The time coefficient is significant at the $0.05$ level under every estimator, with standard errors between $0.10$ and $0.15$, whereas the treatment coefficient is significant under four of them. Web Figure~\ref{fig:applied-intervals} shows the corresponding $95\%$ Wald intervals for the treatment coefficient, ordered by width.

\begin{table}[H]
\centering\small
\caption{Toenail sample of the main paper (ten patients). Standard errors and $p$-values for the treatment coefficient $\hat\beta_1=-3.16$ and the time coefficient $\hat\beta_2=-0.42$ under the fourteen variance estimators. $^{*}$ marks $p<0.05$.}\label{tab:toenail-case1}
\begin{tabular}{lcccc}
\toprule
 & \multicolumn{2}{c}{Treatment} & \multicolumn{2}{c}{Time} \\
\cmidrule(lr){2-3}\cmidrule(lr){4-5}
Estimator & SE & $p$ & SE & $p$ \\
\midrule
$\hat V_{LZ}$ & 0.97 & 0.014$^{*}$ & 0.11 & 0.006$^{*}$ \\
$\hat V_{DF}$ & 1.16 & 0.029$^{*}$ & 0.13 & 0.015$^{*}$ \\
$\hat V_{KC}$ & 1.25 & 0.039$^{*}$ & 0.12 & 0.009$^{*}$ \\
$\hat V_{MD}$ & 1.66 & 0.098 & 0.13 & 0.012$^{*}$ \\
$\hat V_{FG}$ & 1.26 & 0.041$^{*}$ & 0.12 & 0.009$^{*}$ \\
$\hat V_{MBN}$ & 1.62 & 0.092 & 0.15 & 0.024$^{*}$ \\
$\hat V_{\mathrm{Pan}}$ & 2.42 & 0.232 & 0.10 & 0.004$^{*}$ \\
$\hat V_{GST}$ & 2.89 & 0.310 & 0.12 & 0.010$^{*}$ \\
$\hat V_{RS}$ & 2.76 & 0.290 & 0.13 & 0.017$^{*}$ \\
$\hat V_{WL}$ & 2.81 & 0.298 & 0.12 & 0.008$^{*}$ \\
$\hat V_{WB}$ & 2.60 & 0.264 & 0.11 & 0.006$^{*}$ \\
$\hat V_{FW}$ & 1.47 & 0.068 & 0.12 & 0.010$^{*}$ \\
$\hat V_{FZ}$ & 1.46 & 0.067 & 0.12 & 0.009$^{*}$ \\
$\VarAR$ & 1.62 & 0.092 & 0.13 & 0.016$^{*}$ \\
\bottomrule
\end{tabular}
\end{table}

\begin{figure}[H]
\centering
\includegraphics[width=0.7\textwidth]{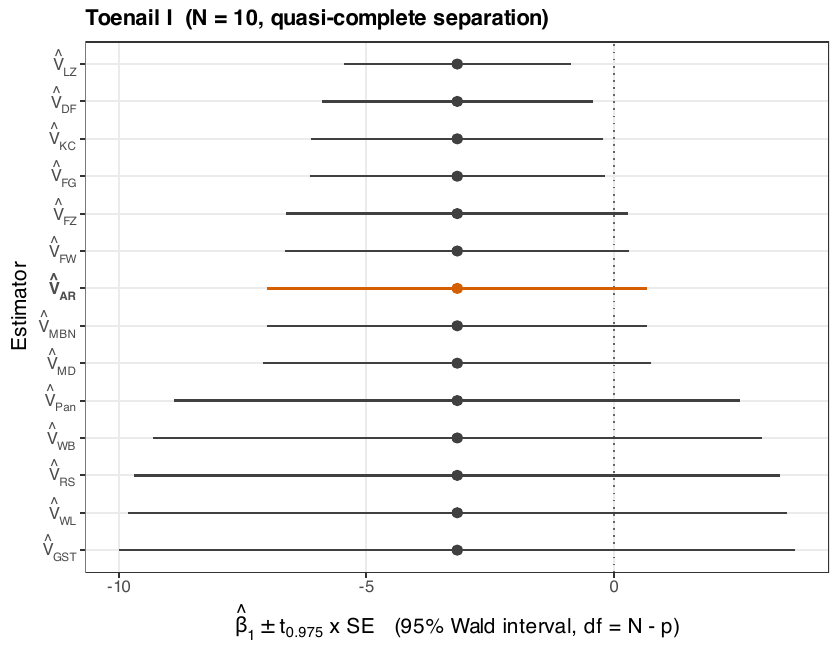}
\caption{Wald $95\%$ confidence intervals for the treatment coefficient in the toenail sample of the main paper under the fourteen variance estimators, ordered by interval width. The intervals use $N-p$ degrees of freedom, and $\VarAR$ is shown in orange.}\label{fig:applied-intervals}
\end{figure}

\subsection{Larger toenail samples}
\label{app:toenail-cases}

Web Table~\ref{tab:toenail-cases} gives the standard errors and $p$-values for the treatment coefficient in two larger samples. Case II has $13$ patients and near separation, with only three moderate or severe infections among the treated observations. Case III has $60$ patients. The range of the standard errors narrows from $1.09$ to $1.40$ in Case II to $0.56$ to $0.59$ in Case III, and no estimator gives a significant treatment effect in either sample.

\begin{table}[H]
\centering\small
\caption{Toenail Cases II and III. Standard errors and $p$-values for the treatment coefficient, which is $\hat\beta_1=-0.94$ in Case II ($13$ patients) and $\hat\beta_1=0.85$ in Case III ($60$ patients).}\label{tab:toenail-cases}
\begin{tabular}{lcccc}
\toprule
 & \multicolumn{2}{c}{Case II} & \multicolumn{2}{c}{Case III} \\
\cmidrule(lr){2-3}\cmidrule(lr){4-5}
Estimator & SE & $p$ & SE & $p$ \\
\midrule
$\hat V_{LZ}$ & 1.09 & 0.409 & 0.56 & 0.137 \\
$\hat V_{DF}$ & 1.24 & 0.467 & 0.58 & 0.147 \\
$\hat V_{KC}$ & 1.20 & 0.453 & 0.57 & 0.143 \\
$\hat V_{MD}$ & 1.33 & 0.496 & 0.58 & 0.150 \\
$\hat V_{FG}$ & 1.24 & 0.466 & 0.58 & 0.147 \\
$\hat V_{MBN}$ & 1.30 & 0.486 & 0.58 & 0.151 \\
$\hat V_{\mathrm{Pan}}$ & 1.10 & 0.416 & 0.56 & 0.133 \\
$\hat V_{GST}$ & 1.26 & 0.474 & 0.57 & 0.143 \\
$\hat V_{RS}$ & 1.21 & 0.456 & 0.57 & 0.141 \\
$\hat V_{WL}$ & 1.30 & 0.489 & 0.57 & 0.146 \\
$\hat V_{WB}$ & 1.20 & 0.453 & 0.57 & 0.139 \\
$\hat V_{FW}$ & 1.26 & 0.476 & 0.58 & 0.147 \\
$\hat V_{FZ}$ & 1.22 & 0.461 & 0.57 & 0.144 \\
$\VarAR$ & 1.40 & 0.517 & 0.59 & 0.154 \\
\bottomrule
\end{tabular}
\end{table}

\subsection{A second ten-patient toenail sample}
\label{app:toenail-seed7375}

Web Table~\ref{tab:toenail-second} gives the results for a second ten-patient sample, drawn with a different seed under the same restriction to patients with all seven visits. Its treatment-by-outcome counts are $(9, 12; 48, 1)$, so the treated observations include a single moderate or severe infection. The PGEE fit gives $\hat\beta_1=-5.41$. All thirteen existing estimators give $p<0.05$ for the treatment coefficient, and $\VarAR$ gives $p=0.054$.

\begin{table}[H]
\centering\small
\caption{Second ten-patient toenail sample. Standard errors and $p$-values for the treatment coefficient $\hat\beta_1=-5.41$. $^{*}$ marks $p<0.05$.}\label{tab:toenail-second}
\begin{tabular}{lcc}
\toprule
Estimator & SE & $p$ \\
\midrule
$\hat V_{LZ}$ & 1.52 & 0.009$^{*}$ \\
$\hat V_{DF}$ & 1.82 & 0.021$^{*}$ \\
$\hat V_{KC}$ & 1.82 & 0.021$^{*}$ \\
$\hat V_{MD}$ & 2.19 & 0.043$^{*}$ \\
$\hat V_{FG}$ & 1.89 & 0.024$^{*}$ \\
$\hat V_{MBN}$ & 1.92 & 0.026$^{*}$ \\
$\hat V_{\mathrm{Pan}}$ & 1.20 & 0.003$^{*}$ \\
$\hat V_{GST}$ & 1.44 & 0.007$^{*}$ \\
$\hat V_{RS}$ & 1.37 & 0.005$^{*}$ \\
$\hat V_{WL}$ & 1.73 & 0.017$^{*}$ \\
$\hat V_{WB}$ & 1.44 & 0.007$^{*}$ \\
$\hat V_{FW}$ & 2.02 & 0.031$^{*}$ \\
$\hat V_{FZ}$ & 1.95 & 0.027$^{*}$ \\
$\VarAR$ & 2.34 & 0.054 \\
\bottomrule
\end{tabular}
\end{table}

\subsection{Contagious bovine pleuropneumonia}

Web Table~\ref{tab:cbpp-full} gives the standard errors and $p$-values for the herd-size and period coefficients in the CBPP example of the main paper, including $\VarAR$. The five pooled estimators cannot be computed because herd sizes are unequal. The herd-size coefficient is $\hat\beta_1=-0.41$ and is not significant under any of the nine computable estimators, while the period coefficient is $\hat\beta_2=-0.51$ and is significant under all nine.

\begin{table}[H]
\centering\small
\caption{CBPP example ($15$ herds, $26$ to $96$ observations per herd). Standard errors and $p$-values for the herd-size coefficient $\hat\beta_1=-0.41$ and the period coefficient $\hat\beta_2=-0.51$. A dash marks an estimator that cannot be computed, and $^{*}$ marks $p<0.05$.}\label{tab:cbpp-full}
\begin{tabular}{lcccc}
\toprule
 & \multicolumn{2}{c}{Herd size} & \multicolumn{2}{c}{Period} \\
\cmidrule(lr){2-3}\cmidrule(lr){4-5}
Estimator & SE & $p$ & SE & $p$ \\
\midrule
$\hat V_{LZ}$ & 0.44 & 0.373 & 0.13 & 0.002$^{*}$ \\
$\hat V_{DF}$ & 0.49 & 0.424 & 0.14 & 0.003$^{*}$ \\
$\hat V_{KC}$ & 0.48 & 0.410 & 0.13 & 0.002$^{*}$ \\
$\hat V_{MD}$ & 0.52 & 0.447 & 0.14 & 0.003$^{*}$ \\
$\hat V_{FG}$ & 0.50 & 0.426 & 0.13 & 0.002$^{*}$ \\
$\hat V_{MBN}$ & 0.51 & 0.440 & 0.15 & 0.005$^{*}$ \\
$\hat V_{\mathrm{Pan}}$ & --- & --- & --- & --- \\
$\hat V_{GST}$ & --- & --- & --- & --- \\
$\hat V_{RS}$ & --- & --- & --- & --- \\
$\hat V_{WL}$ & --- & --- & --- & --- \\
$\hat V_{WB}$ & --- & --- & --- & --- \\
$\hat V_{FW}$ & 0.50 & 0.429 & 0.13 & 0.002$^{*}$ \\
$\hat V_{FZ}$ & 0.49 & 0.416 & 0.13 & 0.002$^{*}$ \\
$\VarAR$ & 0.54 & 0.462 & 0.14 & 0.004$^{*}$ \\
\bottomrule
\end{tabular}
\end{table}

\subsection{Coefficient-level leverage diagnostic}
\label{sec:coefficient-diagnostic}

Under the correctly specified covariance and frozen first-order residual expansion of Theorem~\ref{thm:overcorrection} of the main paper, $M_0=I_0$ and $\E[\hat M_{MD}]=I_0+B_{\mathrm{lev}}$. Multiplication by $\Delta=I_0^{-1}$ on both sides gives
\[
\E[\hat V_{MD}]=\Delta+\Delta B_{\mathrm{lev}}\Delta.
\]
Thus the relative excess in the variance of coefficient $s$ is
\begin{equation}
\rho_s=\frac{[\Delta B_{\mathrm{lev}}\Delta]_{ss}}{\Delta_{ss}},
\qquad s=0,\ldots,p-1.
\label{eq:rho-s}
\end{equation}
The coefficient-level variance excess is given by the diagonal entries of $\Delta B_{\mathrm{lev}}\Delta$. For a fixed contrast $a^T\beta$, the analogous ratio is $a^T\Delta B_{\mathrm{lev}}\Delta a/(a^T\Delta a)$. This expression is unchanged under a nonsingular reparameterization when the same contrast is retained. In particular, reversing a binary treatment coding leaves the treatment-coefficient ratio unchanged. The intercept then represents a different group and need not have the same ratio.

Because $B_{\mathrm{lev}}$ is positive semidefinite, the Rayleigh quotient gives
\[
0\leq\rho_s\leq\lambda_{\max}\!\left(\Delta^{1/2}B_{\mathrm{lev}}\Delta^{1/2}\right).
\]
Under the two-parameter, common-within-group information conditions of Corollary~\ref{cor:overcorrection} of the main paper, the upper bound equals $1/(N_{\min}-1)$. The models of the applied examples have three parameters, so this bound is not claimed for them.

Web Table~\ref{tab:coefficient-diagnostic} gives plug-in values of $\rho_s$ for the applied examples. In the toenail sample of the main paper the treatment ratio is $0.96$, compared with $0.14$ for time, so the full inverse correction is expected to inflate the treatment variance far more than the time variance. In the second ten-patient sample the ratios for the intercept and time are larger than the treatment ratio, and in the two larger samples all ratios are smaller. The herd-size ratio in the CBPP example is $0.17$. The diagnostic describes the leverage sensitivity of each coefficient within the working-model approximation, and it does not enter the definition of $\VarAR$.

\begin{table}[H]
\centering\small
\caption{Coefficient-level leverage diagnostic $\rho_s$ from \eqref{eq:rho-s}, evaluated at the fitted working quantities. The coefficients are the intercept ($\rho_0$), treatment ($\rho_1$) and time ($\rho_2$) for the toenail samples, and the intercept, herd size and period for CBPP.}\label{tab:coefficient-diagnostic}
\begin{tabular}{lcccc}
\toprule
Dataset & $N$ & $\rho_0$ & $\rho_1$ & $\rho_2$ \\
\midrule
Toenail, sample of the main paper & 10 & 0.14 & 0.96 & 0.14 \\
Toenail, Case II & 13 & 0.13 & 0.22 & 0.09 \\
Toenail, Case III & 60 & 0.03 & 0.03 & 0.02 \\
Toenail, second ten-patient sample & 10 & 0.44 & 0.32 & 0.41 \\
CBPP & 15 & 0.12 & 0.17 & 0.09 \\
\bottomrule
\end{tabular}
\end{table}

\webappendix{Additional simulation results}
\label{app:simulation-full}

This appendix gives the results of the simulation study of the main paper in full. As there, the type I error rate is the number of rejections divided by the number of converged fits, tests use $N-p$ degrees of freedom, and every average gives each scenario equal weight. A dash marks an estimator that cannot be computed in a scenario.

\subsection{Results for every scenario}

Web Table~\ref{tab:all-null} gives the type I error rate of the treatment test in each of the $136$ null scenarios under the fourteen estimators. The scenarios are grouped as in the design of the main paper, into the $72$ primary scenarios, $12$ with a misspecified working correlation, $6$ with unequal numbers of observations per subject, $18$ with a treated share of $20\%$ or $50\%$, $8$ with a $5\%$ event rate, $8$ with five subjects, $6$ in which the time effect is also null and $6$ with the two-parameter model. Web Table~\ref{tab:time-null} gives the type I error rate of the test of the time coefficient in the six scenarios in which both coefficients are zero. In those scenarios $\VarAR$ stays below the cutoff of $0.0562$ at every number of subjects. Web Table~\ref{tab:power-all} gives the power of the treatment test in each of the twelve power scenarios. Web Table~\ref{tab:unequal-rejection} gives the mean type I error rate in the designs with unequal numbers of observations per subject, and Web Table~\ref{tab:reference-comparison} compares $N-p$ degrees of freedom with the coefficient-specific degrees of freedom of Pan and Wall\cite{pan2002small} for the three score forms of the main paper.

\begin{landscape}
\begingroup\scriptsize\setlength{\tabcolsep}{2.6pt}
\begin{longtable}{@{}lccccccc*{14}{c}@{}}
\caption{Type I error rate of the treatment test in each of the $136$ null scenarios. $N$ is the number of subjects, $n_i$ the number of observations per subject, Ev.\ the event rate in percent, $\rho$ the true correlation level, Tr.\ the share of treated subjects in percent, and the two correlation columns give the true and the working structure.}\label{tab:all-null}\\
\toprule
Design & $N$ & $n_i$ & Ev. & $\rho$ & Tr. & True & Working & LZ & DF & KC & MD & FG & MBN & Pan & GST & RS & WL & WB & FW & FZ & AR \\
\midrule
\endfirsthead
\toprule
Design & $N$ & $n_i$ & Ev. & $\rho$ & Tr. & True & Working & LZ & DF & KC & MD & FG & MBN & Pan & GST & RS & WL & WB & FW & FZ & AR \\
\midrule
\endhead
\bottomrule
\endfoot
Primary & 10 & 4 & 10 & 0.05 & 30 & AR(1) & AR(1) & 0.147 & 0.089 & 0.097 & 0.059 & 0.049 & 0.016 & 0.031 & 0.019 & 0.016 & 0.014 & 0.023 & 0.075 & 0.082 & 0.048 \\
Primary & 10 & 4 & 10 & 0.05 & 30 & Exch. & Exch. & 0.143 & 0.100 & 0.107 & 0.061 & 0.059 & 0.030 & 0.039 & 0.020 & 0.019 & 0.013 & 0.022 & 0.077 & 0.086 & 0.035 \\
Primary & 10 & 4 & 10 & 0.10 & 30 & AR(1) & AR(1) & 0.149 & 0.089 & 0.097 & 0.059 & 0.050 & 0.017 & 0.028 & 0.015 & 0.014 & 0.011 & 0.017 & 0.074 & 0.084 & 0.048 \\
Primary & 10 & 4 & 10 & 0.10 & 30 & Exch. & Exch. & 0.146 & 0.102 & 0.105 & 0.064 & 0.052 & 0.034 & 0.045 & 0.020 & 0.020 & 0.014 & 0.022 & 0.083 & 0.090 & 0.034 \\
Primary & 10 & 4 & 10 & 0.20 & 30 & AR(1) & AR(1) & 0.150 & 0.088 & 0.095 & 0.061 & 0.034 & 0.016 & 0.030 & 0.019 & 0.019 & 0.013 & 0.021 & 0.074 & 0.082 & 0.043 \\
Primary & 10 & 4 & 10 & 0.20 & 30 & Exch. & Exch. & 0.138 & 0.105 & 0.108 & 0.068 & 0.040 & 0.033 & 0.043 & 0.017 & 0.017 & 0.011 & 0.019 & 0.083 & 0.091 & 0.031 \\
Primary & 10 & 4 & 10 & 0.30 & 30 & AR(1) & AR(1) & 0.147 & 0.089 & 0.096 & 0.063 & 0.031 & 0.015 & 0.028 & 0.015 & 0.015 & 0.008 & 0.017 & 0.077 & 0.084 & 0.040 \\
Primary & 10 & 4 & 10 & 0.30 & 30 & Exch. & Exch. & 0.145 & 0.107 & 0.109 & 0.069 & 0.025 & 0.041 & 0.054 & 0.017 & 0.018 & 0.009 & 0.017 & 0.087 & 0.095 & 0.030 \\
Primary & 10 & 4 & 20 & 0.05 & 30 & AR(1) & AR(1) & 0.095 & 0.061 & 0.067 & 0.035 & 0.057 & 0.012 & 0.031 & 0.012 & 0.007 & 0.009 & 0.018 & 0.051 & 0.057 & 0.037 \\
Primary & 10 & 4 & 20 & 0.05 & 30 & Exch. & Exch. & 0.090 & 0.064 & 0.068 & 0.037 & 0.062 & 0.017 & 0.032 & 0.015 & 0.014 & 0.011 & 0.020 & 0.048 & 0.055 & 0.035 \\
Primary & 10 & 4 & 20 & 0.10 & 30 & AR(1) & AR(1) & 0.101 & 0.062 & 0.067 & 0.037 & 0.061 & 0.015 & 0.035 & 0.016 & 0.010 & 0.012 & 0.023 & 0.053 & 0.059 & 0.038 \\
Primary & 10 & 4 & 20 & 0.10 & 30 & Exch. & Exch. & 0.097 & 0.066 & 0.071 & 0.036 & 0.061 & 0.016 & 0.031 & 0.012 & 0.011 & 0.008 & 0.016 & 0.048 & 0.056 & 0.035 \\
Primary & 10 & 4 & 20 & 0.20 & 30 & AR(1) & AR(1) & 0.115 & 0.070 & 0.078 & 0.042 & 0.066 & 0.016 & 0.037 & 0.016 & 0.014 & 0.012 & 0.023 & 0.061 & 0.067 & 0.044 \\
Primary & 10 & 4 & 20 & 0.20 & 30 & Exch. & Exch. & 0.110 & 0.071 & 0.077 & 0.038 & 0.064 & 0.021 & 0.034 & 0.016 & 0.015 & 0.012 & 0.019 & 0.052 & 0.058 & 0.035 \\
Primary & 10 & 4 & 20 & 0.30 & 30 & AR(1) & AR(1) & 0.107 & 0.060 & 0.070 & 0.034 & 0.057 & 0.011 & 0.028 & 0.013 & 0.011 & 0.008 & 0.015 & 0.051 & 0.058 & 0.033 \\
Primary & 10 & 4 & 20 & 0.30 & 30 & Exch. & Exch. & 0.109 & 0.068 & 0.072 & 0.035 & 0.056 & 0.017 & 0.029 & 0.013 & 0.013 & 0.011 & 0.015 & 0.051 & 0.056 & 0.032 \\
Primary & 10 & 4 & 30 & 0.05 & 30 & AR(1) & AR(1) & 0.058 & 0.034 & 0.037 & 0.026 & 0.034 & 0.012 & 0.032 & 0.013 & 0.007 & 0.011 & 0.020 & 0.031 & 0.034 & 0.023 \\
Primary & 10 & 4 & 30 & 0.05 & 30 & Exch. & Exch. & 0.060 & 0.035 & 0.040 & 0.026 & 0.037 & 0.018 & 0.034 & 0.017 & 0.015 & 0.012 & 0.021 & 0.033 & 0.036 & 0.023 \\
Primary & 10 & 4 & 30 & 0.10 & 30 & AR(1) & AR(1) & 0.058 & 0.035 & 0.040 & 0.026 & 0.035 & 0.014 & 0.033 & 0.016 & 0.008 & 0.012 & 0.022 & 0.032 & 0.035 & 0.024 \\
Primary & 10 & 4 & 30 & 0.10 & 30 & Exch. & Exch. & 0.061 & 0.035 & 0.040 & 0.026 & 0.035 & 0.015 & 0.031 & 0.014 & 0.013 & 0.010 & 0.017 & 0.032 & 0.034 & 0.024 \\
Primary & 10 & 4 & 30 & 0.20 & 30 & AR(1) & AR(1) & 0.064 & 0.039 & 0.045 & 0.026 & 0.040 & 0.010 & 0.029 & 0.012 & 0.006 & 0.009 & 0.017 & 0.034 & 0.038 & 0.024 \\
Primary & 10 & 4 & 30 & 0.20 & 30 & Exch. & Exch. & 0.073 & 0.048 & 0.052 & 0.029 & 0.047 & 0.017 & 0.028 & 0.013 & 0.012 & 0.010 & 0.017 & 0.038 & 0.042 & 0.029 \\
Primary & 10 & 4 & 30 & 0.30 & 30 & AR(1) & AR(1) & 0.071 & 0.043 & 0.047 & 0.030 & 0.043 & 0.013 & 0.029 & 0.014 & 0.009 & 0.010 & 0.018 & 0.039 & 0.042 & 0.030 \\
Primary & 10 & 4 & 30 & 0.30 & 30 & Exch. & Exch. & 0.080 & 0.048 & 0.053 & 0.032 & 0.046 & 0.020 & 0.030 & 0.015 & 0.013 & 0.011 & 0.017 & 0.039 & 0.043 & 0.031 \\
Primary & 30 & 4 & 10 & 0.05 & 30 & AR(1) & AR(1) & 0.067 & 0.055 & 0.057 & 0.052 & 0.053 & 0.050 & 0.038 & 0.030 & 0.028 & 0.027 & 0.034 & 0.054 & 0.057 & 0.048 \\
Primary & 30 & 4 & 10 & 0.05 & 30 & Exch. & Exch. & 0.077 & 0.065 & 0.068 & 0.061 & 0.063 & 0.061 & 0.043 & 0.034 & 0.033 & 0.030 & 0.037 & 0.063 & 0.066 & 0.058 \\
Primary & 30 & 4 & 10 & 0.10 & 30 & AR(1) & AR(1) & 0.082 & 0.067 & 0.070 & 0.061 & 0.063 & 0.061 & 0.047 & 0.036 & 0.035 & 0.030 & 0.038 & 0.065 & 0.069 & 0.058 \\
Primary & 30 & 4 & 10 & 0.10 & 30 & Exch. & Exch. & 0.081 & 0.072 & 0.074 & 0.064 & 0.069 & 0.065 & 0.040 & 0.035 & 0.034 & 0.030 & 0.036 & 0.069 & 0.073 & 0.061 \\
Primary & 30 & 4 & 10 & 0.20 & 30 & AR(1) & AR(1) & 0.088 & 0.077 & 0.079 & 0.070 & 0.073 & 0.069 & 0.043 & 0.034 & 0.032 & 0.029 & 0.036 & 0.074 & 0.078 & 0.067 \\
Primary & 30 & 4 & 10 & 0.20 & 30 & Exch. & Exch. & 0.105 & 0.095 & 0.097 & 0.088 & 0.091 & 0.087 & 0.038 & 0.030 & 0.030 & 0.026 & 0.032 & 0.092 & 0.096 & 0.085 \\
Primary & 30 & 4 & 10 & 0.30 & 30 & AR(1) & AR(1) & 0.101 & 0.087 & 0.089 & 0.082 & 0.084 & 0.080 & 0.041 & 0.032 & 0.031 & 0.027 & 0.033 & 0.085 & 0.087 & 0.079 \\
Primary & 30 & 4 & 10 & 0.30 & 30 & Exch. & Exch. & 0.129 & 0.115 & 0.116 & 0.110 & 0.110 & 0.106 & 0.046 & 0.038 & 0.038 & 0.032 & 0.040 & 0.113 & 0.115 & 0.106 \\
Primary & 30 & 4 & 20 & 0.05 & 30 & AR(1) & AR(1) & 0.057 & 0.046 & 0.047 & 0.038 & 0.043 & 0.038 & 0.044 & 0.035 & 0.034 & 0.032 & 0.039 & 0.042 & 0.046 & 0.033 \\
Primary & 30 & 4 & 20 & 0.05 & 30 & Exch. & Exch. & 0.056 & 0.046 & 0.047 & 0.037 & 0.043 & 0.039 & 0.045 & 0.033 & 0.032 & 0.029 & 0.037 & 0.041 & 0.045 & 0.033 \\
Primary & 30 & 4 & 20 & 0.10 & 30 & AR(1) & AR(1) & 0.059 & 0.047 & 0.049 & 0.040 & 0.045 & 0.041 & 0.046 & 0.036 & 0.035 & 0.033 & 0.038 & 0.043 & 0.048 & 0.035 \\
Primary & 30 & 4 & 20 & 0.10 & 30 & Exch. & Exch. & 0.054 & 0.044 & 0.045 & 0.035 & 0.041 & 0.036 & 0.042 & 0.031 & 0.031 & 0.029 & 0.033 & 0.040 & 0.044 & 0.031 \\
Primary & 30 & 4 & 20 & 0.20 & 30 & AR(1) & AR(1) & 0.057 & 0.047 & 0.048 & 0.038 & 0.045 & 0.039 & 0.044 & 0.032 & 0.031 & 0.029 & 0.036 & 0.044 & 0.047 & 0.031 \\
Primary & 30 & 4 & 20 & 0.20 & 30 & Exch. & Exch. & 0.056 & 0.044 & 0.046 & 0.037 & 0.042 & 0.038 & 0.036 & 0.028 & 0.028 & 0.025 & 0.030 & 0.041 & 0.044 & 0.032 \\
Primary & 30 & 4 & 20 & 0.30 & 30 & AR(1) & AR(1) & 0.062 & 0.048 & 0.049 & 0.041 & 0.046 & 0.042 & 0.040 & 0.030 & 0.030 & 0.028 & 0.033 & 0.045 & 0.048 & 0.036 \\
Primary & 30 & 4 & 20 & 0.30 & 30 & Exch. & Exch. & 0.069 & 0.053 & 0.055 & 0.046 & 0.051 & 0.044 & 0.035 & 0.027 & 0.027 & 0.023 & 0.029 & 0.049 & 0.054 & 0.042 \\
Primary & 30 & 4 & 30 & 0.05 & 30 & AR(1) & AR(1) & 0.054 & 0.043 & 0.045 & 0.035 & 0.041 & 0.036 & 0.042 & 0.033 & 0.032 & 0.030 & 0.037 & 0.040 & 0.044 & 0.031 \\
Primary & 30 & 4 & 30 & 0.05 & 30 & Exch. & Exch. & 0.048 & 0.039 & 0.039 & 0.032 & 0.036 & 0.033 & 0.043 & 0.035 & 0.035 & 0.033 & 0.037 & 0.035 & 0.038 & 0.029 \\
Primary & 30 & 4 & 30 & 0.10 & 30 & AR(1) & AR(1) & 0.054 & 0.042 & 0.043 & 0.036 & 0.040 & 0.036 & 0.043 & 0.035 & 0.034 & 0.033 & 0.038 & 0.038 & 0.043 & 0.032 \\
Primary & 30 & 4 & 30 & 0.10 & 30 & Exch. & Exch. & 0.054 & 0.043 & 0.044 & 0.035 & 0.042 & 0.036 & 0.042 & 0.031 & 0.031 & 0.029 & 0.035 & 0.040 & 0.043 & 0.031 \\
Primary & 30 & 4 & 30 & 0.20 & 30 & AR(1) & AR(1) & 0.054 & 0.043 & 0.044 & 0.036 & 0.041 & 0.036 & 0.045 & 0.034 & 0.033 & 0.031 & 0.037 & 0.040 & 0.043 & 0.032 \\
Primary & 30 & 4 & 30 & 0.20 & 30 & Exch. & Exch. & 0.056 & 0.043 & 0.044 & 0.037 & 0.040 & 0.037 & 0.041 & 0.034 & 0.033 & 0.030 & 0.036 & 0.038 & 0.043 & 0.033 \\
Primary & 30 & 4 & 30 & 0.30 & 30 & AR(1) & AR(1) & 0.052 & 0.038 & 0.039 & 0.032 & 0.037 & 0.034 & 0.044 & 0.032 & 0.031 & 0.029 & 0.037 & 0.036 & 0.039 & 0.030 \\
Primary & 30 & 4 & 30 & 0.30 & 30 & Exch. & Exch. & 0.047 & 0.038 & 0.039 & 0.033 & 0.036 & 0.033 & 0.036 & 0.027 & 0.027 & 0.024 & 0.030 & 0.035 & 0.037 & 0.028 \\
Primary & 50 & 4 & 10 & 0.05 & 30 & AR(1) & AR(1) & 0.058 & 0.052 & 0.053 & 0.047 & 0.051 & 0.049 & 0.050 & 0.042 & 0.042 & 0.041 & 0.044 & 0.051 & 0.053 & 0.044 \\
Primary & 50 & 4 & 10 & 0.05 & 30 & Exch. & Exch. & 0.056 & 0.050 & 0.050 & 0.043 & 0.047 & 0.043 & 0.044 & 0.039 & 0.039 & 0.036 & 0.040 & 0.046 & 0.050 & 0.039 \\
Primary & 50 & 4 & 10 & 0.10 & 30 & AR(1) & AR(1) & 0.057 & 0.049 & 0.049 & 0.042 & 0.046 & 0.043 & 0.043 & 0.037 & 0.037 & 0.036 & 0.039 & 0.045 & 0.049 & 0.039 \\
Primary & 50 & 4 & 10 & 0.10 & 30 & Exch. & Exch. & 0.064 & 0.055 & 0.056 & 0.048 & 0.052 & 0.049 & 0.047 & 0.039 & 0.039 & 0.037 & 0.042 & 0.052 & 0.056 & 0.046 \\
Primary & 50 & 4 & 10 & 0.20 & 30 & AR(1) & AR(1) & 0.063 & 0.055 & 0.056 & 0.050 & 0.053 & 0.051 & 0.045 & 0.038 & 0.038 & 0.036 & 0.040 & 0.053 & 0.056 & 0.048 \\
Primary & 50 & 4 & 10 & 0.20 & 30 & Exch. & Exch. & 0.066 & 0.057 & 0.059 & 0.051 & 0.054 & 0.051 & 0.044 & 0.035 & 0.035 & 0.033 & 0.038 & 0.054 & 0.058 & 0.048 \\
Primary & 50 & 4 & 10 & 0.30 & 30 & AR(1) & AR(1) & 0.076 & 0.068 & 0.068 & 0.063 & 0.066 & 0.063 & 0.051 & 0.044 & 0.044 & 0.043 & 0.045 & 0.065 & 0.068 & 0.059 \\
Primary & 50 & 4 & 10 & 0.30 & 30 & Exch. & Exch. & 0.081 & 0.075 & 0.075 & 0.071 & 0.074 & 0.070 & 0.043 & 0.038 & 0.038 & 0.035 & 0.039 & 0.073 & 0.075 & 0.068 \\
Primary & 50 & 4 & 20 & 0.05 & 30 & AR(1) & AR(1) & 0.048 & 0.041 & 0.041 & 0.037 & 0.040 & 0.037 & 0.042 & 0.038 & 0.038 & 0.035 & 0.039 & 0.039 & 0.041 & 0.034 \\
Primary & 50 & 4 & 20 & 0.05 & 30 & Exch. & Exch. & 0.056 & 0.050 & 0.050 & 0.042 & 0.047 & 0.042 & 0.049 & 0.040 & 0.041 & 0.037 & 0.044 & 0.045 & 0.050 & 0.037 \\
Primary & 50 & 4 & 20 & 0.10 & 30 & AR(1) & AR(1) & 0.056 & 0.051 & 0.052 & 0.046 & 0.050 & 0.046 & 0.050 & 0.043 & 0.042 & 0.040 & 0.045 & 0.049 & 0.051 & 0.041 \\
Primary & 50 & 4 & 20 & 0.10 & 30 & Exch. & Exch. & 0.050 & 0.044 & 0.045 & 0.039 & 0.042 & 0.041 & 0.044 & 0.036 & 0.036 & 0.035 & 0.039 & 0.042 & 0.044 & 0.037 \\
Primary & 50 & 4 & 20 & 0.20 & 30 & AR(1) & AR(1) & 0.051 & 0.044 & 0.045 & 0.039 & 0.043 & 0.040 & 0.045 & 0.036 & 0.036 & 0.035 & 0.039 & 0.042 & 0.044 & 0.036 \\
Primary & 50 & 4 & 20 & 0.20 & 30 & Exch. & Exch. & 0.056 & 0.046 & 0.048 & 0.041 & 0.046 & 0.042 & 0.044 & 0.037 & 0.036 & 0.035 & 0.039 & 0.044 & 0.047 & 0.038 \\
Primary & 50 & 4 & 20 & 0.30 & 30 & AR(1) & AR(1) & 0.059 & 0.050 & 0.051 & 0.046 & 0.049 & 0.046 & 0.048 & 0.041 & 0.041 & 0.040 & 0.043 & 0.048 & 0.051 & 0.042 \\
Primary & 50 & 4 & 20 & 0.30 & 30 & Exch. & Exch. & 0.055 & 0.048 & 0.049 & 0.044 & 0.047 & 0.044 & 0.040 & 0.034 & 0.033 & 0.031 & 0.035 & 0.046 & 0.049 & 0.040 \\
Primary & 50 & 4 & 30 & 0.05 & 30 & AR(1) & AR(1) & 0.050 & 0.043 & 0.044 & 0.038 & 0.042 & 0.039 & 0.046 & 0.038 & 0.038 & 0.037 & 0.041 & 0.040 & 0.044 & 0.035 \\
Primary & 50 & 4 & 30 & 0.05 & 30 & Exch. & Exch. & 0.051 & 0.043 & 0.044 & 0.039 & 0.042 & 0.040 & 0.045 & 0.039 & 0.038 & 0.037 & 0.041 & 0.041 & 0.044 & 0.036 \\
Primary & 50 & 4 & 30 & 0.10 & 30 & AR(1) & AR(1) & 0.053 & 0.047 & 0.048 & 0.042 & 0.046 & 0.044 & 0.050 & 0.042 & 0.041 & 0.040 & 0.046 & 0.045 & 0.047 & 0.040 \\
Primary & 50 & 4 & 30 & 0.10 & 30 & Exch. & Exch. & 0.056 & 0.050 & 0.050 & 0.045 & 0.048 & 0.046 & 0.049 & 0.041 & 0.041 & 0.040 & 0.044 & 0.047 & 0.050 & 0.042 \\
Primary & 50 & 4 & 30 & 0.20 & 30 & AR(1) & AR(1) & 0.051 & 0.043 & 0.044 & 0.039 & 0.042 & 0.041 & 0.047 & 0.041 & 0.040 & 0.039 & 0.042 & 0.041 & 0.044 & 0.037 \\
Primary & 50 & 4 & 30 & 0.20 & 30 & Exch. & Exch. & 0.048 & 0.042 & 0.043 & 0.039 & 0.041 & 0.040 & 0.042 & 0.037 & 0.037 & 0.036 & 0.039 & 0.040 & 0.043 & 0.036 \\
Primary & 50 & 4 & 30 & 0.30 & 30 & AR(1) & AR(1) & 0.046 & 0.039 & 0.040 & 0.034 & 0.038 & 0.035 & 0.040 & 0.034 & 0.034 & 0.033 & 0.036 & 0.037 & 0.039 & 0.031 \\
Primary & 50 & 4 & 30 & 0.30 & 30 & Exch. & Exch. & 0.049 & 0.045 & 0.046 & 0.039 & 0.044 & 0.040 & 0.042 & 0.036 & 0.036 & 0.034 & 0.039 & 0.043 & 0.046 & 0.036 \\
\addlinespace
Misspecified & 10 & 4 & 10 & 0.20 & 30 & AR(1) & Exch. & 0.138 & 0.095 & 0.100 & 0.063 & 0.045 & 0.030 & 0.042 & 0.017 & 0.017 & 0.011 & 0.020 & 0.079 & 0.088 & 0.035 \\
Misspecified & 10 & 4 & 10 & 0.20 & 30 & Exch. & AR(1) & 0.156 & 0.099 & 0.104 & 0.077 & 0.034 & 0.021 & 0.030 & 0.021 & 0.020 & 0.014 & 0.022 & 0.089 & 0.096 & 0.055 \\
Misspecified & 10 & 4 & 20 & 0.20 & 30 & AR(1) & Exch. & 0.100 & 0.067 & 0.073 & 0.035 & 0.066 & 0.015 & 0.030 & 0.014 & 0.012 & 0.012 & 0.018 & 0.050 & 0.058 & 0.035 \\
Misspecified & 10 & 4 & 20 & 0.20 & 30 & Exch. & AR(1) & 0.117 & 0.067 & 0.078 & 0.040 & 0.065 & 0.016 & 0.032 & 0.014 & 0.012 & 0.012 & 0.019 & 0.055 & 0.062 & 0.038 \\
Misspecified & 30 & 4 & 10 & 0.20 & 30 & AR(1) & Exch. & 0.084 & 0.073 & 0.075 & 0.067 & 0.069 & 0.067 & 0.040 & 0.033 & 0.033 & 0.027 & 0.034 & 0.071 & 0.074 & 0.065 \\
Misspecified & 30 & 4 & 10 & 0.20 & 30 & Exch. & AR(1) & 0.105 & 0.092 & 0.094 & 0.086 & 0.087 & 0.084 & 0.040 & 0.030 & 0.030 & 0.024 & 0.033 & 0.089 & 0.092 & 0.083 \\
Misspecified & 30 & 4 & 20 & 0.20 & 30 & AR(1) & Exch. & 0.056 & 0.045 & 0.046 & 0.038 & 0.044 & 0.038 & 0.042 & 0.035 & 0.034 & 0.031 & 0.036 & 0.042 & 0.045 & 0.034 \\
Misspecified & 30 & 4 & 20 & 0.20 & 30 & Exch. & AR(1) & 0.059 & 0.048 & 0.049 & 0.039 & 0.046 & 0.041 & 0.037 & 0.028 & 0.028 & 0.026 & 0.031 & 0.044 & 0.048 & 0.036 \\
Misspecified & 50 & 4 & 10 & 0.20 & 30 & AR(1) & Exch. & 0.058 & 0.051 & 0.052 & 0.043 & 0.049 & 0.044 & 0.040 & 0.034 & 0.033 & 0.032 & 0.036 & 0.047 & 0.051 & 0.039 \\
Misspecified & 50 & 4 & 10 & 0.20 & 30 & Exch. & AR(1) & 0.066 & 0.058 & 0.059 & 0.052 & 0.056 & 0.053 & 0.035 & 0.026 & 0.026 & 0.025 & 0.029 & 0.056 & 0.059 & 0.049 \\
Misspecified & 50 & 4 & 20 & 0.20 & 30 & AR(1) & Exch. & 0.051 & 0.047 & 0.047 & 0.041 & 0.046 & 0.043 & 0.046 & 0.039 & 0.039 & 0.037 & 0.041 & 0.044 & 0.047 & 0.039 \\
Misspecified & 50 & 4 & 20 & 0.20 & 30 & Exch. & AR(1) & 0.056 & 0.049 & 0.049 & 0.041 & 0.046 & 0.044 & 0.046 & 0.039 & 0.039 & 0.036 & 0.041 & 0.045 & 0.049 & 0.038 \\
\addlinespace
Unequal sizes & 10 & 2--6 & 20 & 0.20 & 30 & Exch. & Exch. & 0.105 & 0.061 & 0.057 & 0.028 & 0.043 & 0.021 & --- & --- & --- & --- & --- & 0.040 & 0.048 & 0.028 \\
Unequal sizes & 10 & 3--8 & 20 & 0.20 & 30 & Exch. & Exch. & 0.106 & 0.063 & 0.063 & 0.027 & 0.053 & 0.022 & --- & --- & --- & --- & --- & 0.041 & 0.050 & 0.028 \\
Unequal sizes & 30 & 2--6 & 20 & 0.20 & 30 & Exch. & Exch. & 0.061 & 0.050 & 0.049 & 0.038 & 0.045 & 0.041 & --- & --- & --- & --- & --- & 0.043 & 0.048 & 0.033 \\
Unequal sizes & 30 & 3--8 & 20 & 0.20 & 30 & Exch. & Exch. & 0.058 & 0.048 & 0.048 & 0.040 & 0.046 & 0.043 & --- & --- & --- & --- & --- & 0.044 & 0.047 & 0.036 \\
Unequal sizes & 50 & 2--6 & 20 & 0.20 & 30 & Exch. & Exch. & 0.051 & 0.046 & 0.045 & 0.040 & 0.044 & 0.041 & --- & --- & --- & --- & --- & 0.043 & 0.045 & 0.038 \\
Unequal sizes & 50 & 3--8 & 20 & 0.20 & 30 & Exch. & Exch. & 0.056 & 0.051 & 0.051 & 0.046 & 0.048 & 0.048 & --- & --- & --- & --- & --- & 0.047 & 0.050 & 0.044 \\
\addlinespace
Treated share & 10 & 4 & 10 & 0.20 & 20 & Exch. & Exch. & 0.076 & 0.060 & 0.064 & 0.034 & 0.025 & 0.031 & 0.048 & 0.028 & 0.030 & 0.011 & 0.034 & 0.043 & 0.045 & 0.027 \\
Treated share & 10 & 4 & 10 & 0.20 & 50 & Exch. & Exch. & 0.249 & 0.192 & 0.194 & 0.185 & 0.108 & 0.040 & 0.024 & 0.006 & 0.003 & 0.004 & 0.008 & 0.192 & 0.194 & 0.180 \\
Treated share & 10 & 4 & 20 & 0.20 & 20 & Exch. & Exch. & 0.077 & 0.043 & 0.039 & 0.021 & 0.029 & 0.018 & 0.032 & 0.015 & 0.014 & 0.009 & 0.021 & 0.028 & 0.029 & 0.016 \\
Treated share & 10 & 4 & 20 & 0.20 & 50 & Exch. & Exch. & 0.084 & 0.065 & 0.069 & 0.058 & 0.041 & 0.027 & 0.031 & 0.013 & 0.010 & 0.008 & 0.017 & 0.065 & 0.068 & 0.054 \\
Treated share & 10 & 4 & 30 & 0.20 & 20 & Exch. & Exch. & 0.097 & 0.051 & 0.046 & 0.024 & 0.040 & 0.021 & 0.025 & 0.013 & 0.011 & 0.009 & 0.017 & 0.030 & 0.033 & 0.021 \\
Treated share & 10 & 4 & 30 & 0.20 & 50 & Exch. & Exch. & 0.055 & 0.028 & 0.035 & 0.023 & 0.025 & 0.016 & 0.038 & 0.017 & 0.014 & 0.013 & 0.022 & 0.029 & 0.033 & 0.018 \\
Treated share & 30 & 4 & 10 & 0.20 & 20 & Exch. & Exch. & 0.178 & 0.162 & 0.160 & 0.150 & 0.153 & 0.130 & 0.044 & 0.036 & 0.035 & 0.031 & 0.037 & 0.154 & 0.158 & 0.148 \\
Treated share & 30 & 4 & 10 & 0.20 & 50 & Exch. & Exch. & 0.079 & 0.069 & 0.071 & 0.061 & 0.064 & 0.056 & 0.042 & 0.028 & 0.028 & 0.024 & 0.034 & 0.069 & 0.071 & 0.054 \\
Treated share & 30 & 4 & 20 & 0.20 & 20 & Exch. & Exch. & 0.072 & 0.057 & 0.052 & 0.042 & 0.050 & 0.049 & 0.033 & 0.027 & 0.026 & 0.024 & 0.027 & 0.046 & 0.050 & 0.039 \\
Treated share & 30 & 4 & 20 & 0.20 & 50 & Exch. & Exch. & 0.053 & 0.042 & 0.046 & 0.038 & 0.042 & 0.036 & 0.047 & 0.034 & 0.033 & 0.031 & 0.039 & 0.042 & 0.045 & 0.033 \\
Treated share & 30 & 4 & 30 & 0.20 & 20 & Exch. & Exch. & 0.066 & 0.056 & 0.053 & 0.039 & 0.051 & 0.046 & 0.034 & 0.026 & 0.026 & 0.024 & 0.028 & 0.045 & 0.051 & 0.037 \\
Treated share & 30 & 4 & 30 & 0.20 & 50 & Exch. & Exch. & 0.052 & 0.042 & 0.045 & 0.038 & 0.041 & 0.036 & 0.050 & 0.038 & 0.038 & 0.035 & 0.042 & 0.042 & 0.044 & 0.034 \\
Treated share & 50 & 4 & 10 & 0.20 & 20 & Exch. & Exch. & 0.092 & 0.083 & 0.081 & 0.074 & 0.078 & 0.078 & 0.039 & 0.034 & 0.034 & 0.031 & 0.036 & 0.077 & 0.080 & 0.071 \\
Treated share & 50 & 4 & 10 & 0.20 & 50 & Exch. & Exch. & 0.063 & 0.056 & 0.058 & 0.053 & 0.055 & 0.051 & 0.050 & 0.043 & 0.043 & 0.041 & 0.045 & 0.056 & 0.058 & 0.050 \\
Treated share & 50 & 4 & 20 & 0.20 & 20 & Exch. & Exch. & 0.061 & 0.055 & 0.052 & 0.043 & 0.049 & 0.049 & 0.037 & 0.033 & 0.033 & 0.032 & 0.034 & 0.048 & 0.051 & 0.040 \\
Treated share & 50 & 4 & 20 & 0.20 & 50 & Exch. & Exch. & 0.045 & 0.038 & 0.041 & 0.037 & 0.038 & 0.036 & 0.042 & 0.035 & 0.035 & 0.033 & 0.037 & 0.038 & 0.041 & 0.033 \\
Treated share & 50 & 4 & 30 & 0.20 & 20 & Exch. & Exch. & 0.064 & 0.057 & 0.054 & 0.045 & 0.052 & 0.052 & 0.044 & 0.038 & 0.038 & 0.036 & 0.040 & 0.049 & 0.053 & 0.043 \\
Treated share & 50 & 4 & 30 & 0.20 & 50 & Exch. & Exch. & 0.046 & 0.039 & 0.041 & 0.037 & 0.038 & 0.036 & 0.045 & 0.038 & 0.038 & 0.036 & 0.040 & 0.039 & 0.040 & 0.034 \\
\addlinespace
5\% events & 30 & 4 & 5 & 0.10 & 30 & Exch. & Exch. & 0.204 & 0.197 & 0.199 & 0.191 & 0.186 & 0.173 & 0.044 & 0.038 & 0.039 & 0.025 & 0.039 & 0.196 & 0.198 & 0.189 \\
5\% events & 30 & 8 & 5 & 0.10 & 30 & Exch. & Exch. & 0.109 & 0.098 & 0.099 & 0.090 & 0.093 & 0.089 & 0.043 & 0.035 & 0.035 & 0.033 & 0.036 & 0.094 & 0.099 & 0.088 \\
5\% events & 30 & 4 & 5 & 0.30 & 30 & Exch. & Exch. & 0.247 & 0.230 & 0.237 & 0.225 & 0.221 & 0.195 & 0.054 & 0.046 & 0.047 & 0.034 & 0.047 & 0.230 & 0.235 & 0.224 \\
5\% events & 30 & 8 & 5 & 0.30 & 30 & Exch. & Exch. & 0.207 & 0.194 & 0.197 & 0.184 & 0.182 & 0.167 & 0.059 & 0.048 & 0.048 & 0.042 & 0.049 & 0.190 & 0.196 & 0.183 \\
5\% events & 50 & 4 & 5 & 0.10 & 30 & Exch. & Exch. & 0.120 & 0.113 & 0.116 & 0.110 & 0.113 & 0.109 & 0.043 & 0.040 & 0.040 & 0.036 & 0.041 & 0.111 & 0.115 & 0.109 \\
5\% events & 50 & 8 & 5 & 0.10 & 30 & Exch. & Exch. & 0.080 & 0.071 & 0.071 & 0.065 & 0.069 & 0.066 & 0.046 & 0.039 & 0.039 & 0.037 & 0.042 & 0.068 & 0.071 & 0.062 \\
5\% events & 50 & 4 & 5 & 0.30 & 30 & Exch. & Exch. & 0.177 & 0.165 & 0.169 & 0.162 & 0.165 & 0.156 & 0.045 & 0.041 & 0.041 & 0.037 & 0.042 & 0.164 & 0.168 & 0.160 \\
5\% events & 50 & 8 & 5 & 0.30 & 30 & Exch. & Exch. & 0.131 & 0.120 & 0.122 & 0.114 & 0.117 & 0.111 & 0.051 & 0.045 & 0.045 & 0.042 & 0.046 & 0.117 & 0.121 & 0.110 \\
\addlinespace
Five subjects & 5 & 6 & 20 & 0.10 & 30 & Exch. & Exch. & 0.030 & 0.005 & 0.012 & 0.005 & 0.007 & 0.002 & 0.022 & 0.007 & 0.003 & 0.006 & 0.010 & 0.007 & 0.008 & 0.008 \\
Five subjects & 5 & 8 & 20 & 0.10 & 30 & Exch. & Exch. & 0.028 & 0.008 & 0.014 & 0.008 & 0.007 & 0.004 & 0.023 & 0.006 & 0.005 & 0.005 & 0.010 & 0.010 & 0.010 & 0.009 \\
Five subjects & 5 & 6 & 20 & 0.30 & 30 & Exch. & Exch. & 0.041 & 0.008 & 0.022 & 0.007 & 0.009 & 0.003 & 0.029 & 0.007 & 0.005 & 0.007 & 0.015 & 0.011 & 0.012 & 0.016 \\
Five subjects & 5 & 8 & 20 & 0.30 & 30 & Exch. & Exch. & 0.038 & 0.007 & 0.018 & 0.006 & 0.005 & 0.003 & 0.021 & 0.005 & 0.004 & 0.004 & 0.010 & 0.010 & 0.011 & 0.010 \\
Five subjects & 5 & 6 & 30 & 0.10 & 30 & Exch. & Exch. & 0.032 & 0.009 & 0.014 & 0.008 & 0.012 & 0.008 & 0.024 & 0.007 & 0.008 & 0.006 & 0.011 & 0.010 & 0.012 & 0.007 \\
Five subjects & 5 & 8 & 30 & 0.10 & 30 & Exch. & Exch. & 0.040 & 0.011 & 0.021 & 0.011 & 0.017 & 0.012 & 0.032 & 0.008 & 0.013 & 0.007 & 0.015 & 0.014 & 0.016 & 0.008 \\
Five subjects & 5 & 6 & 30 & 0.30 & 30 & Exch. & Exch. & 0.031 & 0.007 & 0.014 & 0.006 & 0.006 & 0.005 & 0.021 & 0.006 & 0.006 & 0.004 & 0.009 & 0.008 & 0.008 & 0.007 \\
Five subjects & 5 & 8 & 30 & 0.30 & 30 & Exch. & Exch. & 0.034 & 0.008 & 0.014 & 0.006 & 0.008 & 0.005 & 0.025 & 0.005 & 0.006 & 0.005 & 0.008 & 0.009 & 0.009 & 0.006 \\
\addlinespace
Null time effect & 10 & 4 & 20 & 0.20 & 30 & Exch. & Exch. & 0.111 & 0.074 & 0.079 & 0.042 & 0.068 & 0.019 & 0.031 & 0.017 & 0.016 & 0.012 & 0.021 & 0.057 & 0.063 & 0.041 \\
Null time effect & 10 & 4 & 30 & 0.20 & 30 & Exch. & Exch. & 0.069 & 0.043 & 0.048 & 0.028 & 0.042 & 0.017 & 0.027 & 0.013 & 0.012 & 0.009 & 0.017 & 0.036 & 0.038 & 0.027 \\
Null time effect & 30 & 4 & 20 & 0.20 & 30 & Exch. & Exch. & 0.062 & 0.052 & 0.053 & 0.045 & 0.049 & 0.045 & 0.041 & 0.030 & 0.030 & 0.028 & 0.033 & 0.049 & 0.053 & 0.040 \\
Null time effect & 30 & 4 & 30 & 0.20 & 30 & Exch. & Exch. & 0.054 & 0.042 & 0.044 & 0.034 & 0.041 & 0.036 & 0.039 & 0.032 & 0.032 & 0.029 & 0.035 & 0.040 & 0.043 & 0.031 \\
Null time effect & 50 & 4 & 20 & 0.20 & 30 & Exch. & Exch. & 0.051 & 0.044 & 0.044 & 0.039 & 0.042 & 0.040 & 0.042 & 0.035 & 0.035 & 0.034 & 0.037 & 0.042 & 0.044 & 0.036 \\
Null time effect & 50 & 4 & 30 & 0.20 & 30 & Exch. & Exch. & 0.050 & 0.043 & 0.044 & 0.038 & 0.042 & 0.040 & 0.042 & 0.037 & 0.037 & 0.036 & 0.038 & 0.041 & 0.044 & 0.035 \\
\addlinespace
Two-parameter model & 10 & 4 & 20 & 0.20 & 30 & AR(1) & AR(1) & 0.108 & 0.085 & 0.070 & 0.041 & 0.061 & 0.037 & 0.039 & 0.025 & 0.021 & 0.016 & 0.025 & 0.058 & 0.064 & 0.045 \\
Two-parameter model & 10 & 4 & 20 & 0.20 & 30 & Exch. & Exch. & 0.112 & 0.085 & 0.077 & 0.045 & 0.062 & 0.044 & 0.035 & 0.022 & 0.022 & 0.014 & 0.023 & 0.061 & 0.076 & 0.045 \\
Two-parameter model & 30 & 4 & 20 & 0.20 & 30 & AR(1) & AR(1) & 0.057 & 0.049 & 0.047 & 0.035 & 0.042 & 0.043 & 0.044 & 0.038 & 0.038 & 0.031 & 0.038 & 0.041 & 0.046 & 0.032 \\
Two-parameter model & 30 & 4 & 20 & 0.20 & 30 & Exch. & Exch. & 0.057 & 0.050 & 0.048 & 0.039 & 0.046 & 0.045 & 0.041 & 0.034 & 0.033 & 0.027 & 0.033 & 0.043 & 0.048 & 0.036 \\
Two-parameter model & 50 & 4 & 20 & 0.20 & 30 & AR(1) & AR(1) & 0.055 & 0.051 & 0.050 & 0.044 & 0.047 & 0.048 & 0.047 & 0.043 & 0.042 & 0.039 & 0.043 & 0.047 & 0.049 & 0.040 \\
Two-parameter model & 50 & 4 & 20 & 0.20 & 30 & Exch. & Exch. & 0.054 & 0.049 & 0.047 & 0.040 & 0.044 & 0.044 & 0.043 & 0.040 & 0.040 & 0.036 & 0.040 & 0.043 & 0.046 & 0.037 \\
\end{longtable}
\endgroup
\end{landscape}

\begin{table}[H]
\centering\small
\caption{Type I error rate of the test of the time coefficient in the six scenarios in which the treatment and time coefficients are both zero ($\rho=0.2$, exchangeable correlation, four observations per subject).}\label{tab:time-null}
\begin{tabular}{lcccccc}
\toprule
 & \multicolumn{3}{c}{$20\%$ events} & \multicolumn{3}{c}{$30\%$ events} \\
\cmidrule(lr){2-4}\cmidrule(lr){5-7}
Estimator & $N=10$ & $N=30$ & $N=50$ & $N=10$ & $N=30$ & $N=50$ \\
\midrule
$\hat V_{LZ}$ & 0.080 & 0.050 & 0.047 & 0.044 & 0.047 & 0.047 \\
$\hat V_{DF}$ & 0.055 & 0.041 & 0.041 & 0.021 & 0.035 & 0.039 \\
$\hat V_{KC}$ & 0.068 & 0.048 & 0.045 & 0.033 & 0.043 & 0.043 \\
$\hat V_{MD}$ & 0.059 & 0.044 & 0.043 & 0.027 & 0.040 & 0.041 \\
$\hat V_{FG}$ & 0.063 & 0.048 & 0.045 & 0.033 & 0.043 & 0.044 \\
$\hat V_{MBN}$ & 0.007 & 0.034 & 0.036 & 0.004 & 0.028 & 0.036 \\
$\hat V_{\mathrm{Pan}}$ & 0.087 & 0.050 & 0.047 & 0.045 & 0.046 & 0.047 \\
$\hat V_{GST}$ & 0.057 & 0.041 & 0.040 & 0.023 & 0.035 & 0.039 \\
$\hat V_{RS}$ & 0.020 & 0.037 & 0.040 & 0.007 & 0.033 & 0.039 \\
$\hat V_{WL}$ & 0.065 & 0.043 & 0.043 & 0.029 & 0.040 & 0.042 \\
$\hat V_{WB}$ & 0.076 & 0.047 & 0.046 & 0.036 & 0.043 & 0.045 \\
$\hat V_{FW}$ & 0.062 & 0.047 & 0.044 & 0.030 & 0.041 & 0.042 \\
$\hat V_{FZ}$ & 0.065 & 0.048 & 0.045 & 0.033 & 0.043 & 0.043 \\
$\VarAR$ & 0.052 & 0.040 & 0.040 & 0.020 & 0.033 & 0.038 \\
\bottomrule
\end{tabular}
\end{table}

\begin{table}[H]
\centering\footnotesize\setlength{\tabcolsep}{3.5pt}
\caption{Power of the treatment test in each of the twelve power scenarios ($\beta_1=\log 2$, $20\%$ baseline event rate). Columns give the true correlation level and structure, AR for AR(1) and Ex for exchangeable.}\label{tab:power-all}
\begin{tabular}{l*{12}{c}}
\toprule
 & \multicolumn{4}{c}{$N=10$} & \multicolumn{4}{c}{$N=30$} & \multicolumn{4}{c}{$N=50$} \\
\cmidrule(lr){2-5}\cmidrule(lr){6-9}\cmidrule(lr){10-13}
 & \multicolumn{2}{c}{$\rho=0.1$} & \multicolumn{2}{c}{$\rho=0.3$} & \multicolumn{2}{c}{$\rho=0.1$} & \multicolumn{2}{c}{$\rho=0.3$} & \multicolumn{2}{c}{$\rho=0.1$} & \multicolumn{2}{c}{$\rho=0.3$} \\
Estimator & AR & Ex & AR & Ex & AR & Ex & AR & Ex & AR & Ex & AR & Ex \\
\midrule
$\hat V_{LZ}$ & 0.150 & 0.131 & 0.126 & 0.115 & 0.306 & 0.269 & 0.243 & 0.200 & 0.455 & 0.414 & 0.361 & 0.311 \\
$\hat V_{DF}$ & 0.089 & 0.083 & 0.078 & 0.069 & 0.270 & 0.231 & 0.212 & 0.170 & 0.432 & 0.388 & 0.340 & 0.291 \\
$\hat V_{KC}$ & 0.103 & 0.095 & 0.089 & 0.079 & 0.277 & 0.236 & 0.217 & 0.173 & 0.435 & 0.391 & 0.342 & 0.294 \\
$\hat V_{MD}$ & 0.067 & 0.058 & 0.056 & 0.046 & 0.249 & 0.208 & 0.188 & 0.149 & 0.414 & 0.369 & 0.324 & 0.273 \\
$\hat V_{FG}$ & 0.082 & 0.071 & 0.066 & 0.051 & 0.262 & 0.221 & 0.201 & 0.162 & 0.426 & 0.380 & 0.335 & 0.283 \\
$\hat V_{MBN}$ & 0.047 & 0.051 & 0.037 & 0.041 & 0.252 & 0.216 & 0.194 & 0.156 & 0.421 & 0.374 & 0.331 & 0.278 \\
$\hat V_{\mathrm{Pan}}$ & 0.115 & 0.105 & 0.096 & 0.088 & 0.297 & 0.263 & 0.244 & 0.200 & 0.452 & 0.407 & 0.365 & 0.316 \\
$\hat V_{GST}$ & 0.071 & 0.065 & 0.055 & 0.047 & 0.264 & 0.230 & 0.214 & 0.173 & 0.431 & 0.385 & 0.342 & 0.293 \\
$\hat V_{RS}$ & 0.053 & 0.064 & 0.043 & 0.046 & 0.262 & 0.230 & 0.213 & 0.173 & 0.430 & 0.384 & 0.341 & 0.293 \\
$\hat V_{WL}$ & 0.055 & 0.050 & 0.043 & 0.035 & 0.253 & 0.219 & 0.205 & 0.164 & 0.423 & 0.376 & 0.333 & 0.288 \\
$\hat V_{WB}$ & 0.086 & 0.077 & 0.066 & 0.054 & 0.274 & 0.243 & 0.225 & 0.181 & 0.438 & 0.392 & 0.348 & 0.300 \\
$\hat V_{FW}$ & 0.082 & 0.075 & 0.072 & 0.061 & 0.263 & 0.222 & 0.202 & 0.162 & 0.425 & 0.380 & 0.335 & 0.284 \\
$\hat V_{FZ}$ & 0.090 & 0.083 & 0.077 & 0.068 & 0.273 & 0.231 & 0.214 & 0.171 & 0.434 & 0.391 & 0.341 & 0.293 \\
$\VarAR$ & 0.051 & 0.046 & 0.047 & 0.037 & 0.236 & 0.195 & 0.176 & 0.136 & 0.403 & 0.358 & 0.312 & 0.263 \\
\bottomrule
\end{tabular}
\end{table}

\begin{table}[H]
\centering\small
\caption{Mean type I error rates for the treatment effect over the two unequal-size scenarios at each sample size.}\label{tab:unequal-rejection}
\resizebox{\ifdim\width>\linewidth\linewidth\else\width\fi}{!}{%
\begin{tabular}{lccc}
\toprule
Estimator & $N=10$ & $N=30$ & $N=50$ \\
\midrule
$\hat V_{LZ}$ & 0.1055 & 0.0593 & 0.0533 \\
$\hat V_{KC}$ & 0.0602 & 0.0486 & 0.0481 \\
$\hat V_{MD}$ & 0.0279 & 0.0388 & 0.0432 \\
$\hat V_{FG}$ & 0.0482 & 0.0459 & 0.0463 \\
$\hat V_{MBN}$ & 0.0213 & 0.0416 & 0.0445 \\
$\VarAR$ & 0.0279 & 0.0347 & 0.0412 \\
\bottomrule
\end{tabular}}
\end{table}

\begin{table}[H]
\centering\small
\caption{Mean primary-grid type I error rates under $N-p$ and Pan and Wall (PW) degrees of freedom, with no added centering or scalar scaling.}\label{tab:reference-comparison}
\resizebox{\ifdim\width>\linewidth\linewidth\else\width\fi}{!}{%
\begin{tabular}{lcccccc}
\toprule
Score form & $N=10$, $N-p$ & $N=10$, PW & $N=30$, $N-p$ & $N=30$, PW & $N=50$, $N-p$ & $N=50$, PW \\
\midrule
Uncorrected & 0.1048 & 0.1423 & 0.0675 & 0.0612 & 0.0565 & 0.0481 \\
Half power & 0.0725 & 0.0989 & 0.0572 & 0.0524 & 0.0503 & 0.0424 \\
Full inverse & 0.0426 & 0.0563 & 0.0490 & 0.0445 & 0.0442 & 0.0368 \\
\bottomrule
\end{tabular}}
\end{table}

\subsection{Accuracy and variability of the standard errors}
\label{app:stability}

Web Figure~\ref{fig:se-ratio} shows the median estimated standard error of the treatment coefficient divided by its simulation standard error, by number of subjects and event rate, with a target of one. Web Figure~\ref{fig:se-ratio-rho} shows the same ratio with ten subjects by the true correlation level. Web Table~\ref{tab:stability} summarizes the distribution of the estimated standard errors across replications with ten subjects. There the median ratio is $0.73$ for $\hat V_{LZ}$ and $0.85$ for $\hat V_{KC}$, against $1.03$ for $\hat V_{MD}$, $0.99$ for $\hat V_{MBN}$ and $1.09$ for $\VarAR$. The standard errors of $\hat V_{MD}$ and $\VarAR$ have similar coefficients of variation ($0.226$ and $0.235$) and similar tail ratios.

\begin{table}[H]
\centering\small
\caption{Distribution of the estimated standard error of the treatment coefficient with ten subjects. Each entry is the median over the null scenarios with $N=10$ in which the estimator can be computed. CV is the coefficient of variation of the standard error across replications, P95/P50 and P99/P50 are the ratios of its 95th and 99th percentiles to its median, and SE/SimSE is the median standard error divided by the simulation standard error.}\label{tab:stability}
\begin{tabular}{lccccc}
\toprule
Estimator & Scenarios & CV & P95/P50 & P99/P50 & SE/SimSE \\
\midrule
$\hat V_{LZ}$ & 40 & 0.204 & 1.37 & 1.56 & 0.73 \\
$\hat V_{DF}$ & 40 & 0.204 & 1.37 & 1.56 & 0.87 \\
$\hat V_{KC}$ & 40 & 0.215 & 1.40 & 1.60 & 0.85 \\
$\hat V_{MD}$ & 40 & 0.226 & 1.42 & 1.62 & 1.03 \\
$\hat V_{FG}$ & 40 & 0.322 & 1.49 & 1.81 & 0.89 \\
$\hat V_{MBN}$ & 40 & 0.198 & 1.31 & 1.47 & 0.99 \\
$\hat V_{\mathrm{Pan}}$ & 38 & 0.297 & 1.71 & 2.06 & 0.81 \\
$\hat V_{GST}$ & 38 & 0.297 & 1.71 & 2.06 & 0.97 \\
$\hat V_{RS}$ & 38 & 0.273 & 1.68 & 1.95 & 0.98 \\
$\hat V_{WL}$ & 38 & 0.283 & 1.61 & 1.98 & 1.03 \\
$\hat V_{WB}$ & 38 & 0.289 & 1.66 & 2.01 & 0.91 \\
$\hat V_{FW}$ & 40 & 0.222 & 1.41 & 1.61 & 0.95 \\
$\hat V_{FZ}$ & 40 & 0.221 & 1.41 & 1.61 & 0.92 \\
$\VarAR$ & 40 & 0.235 & 1.45 & 1.64 & 1.09 \\
\bottomrule
\end{tabular}
\end{table}

\begin{figure}[p]
\centering
\includegraphics[width=\textwidth]{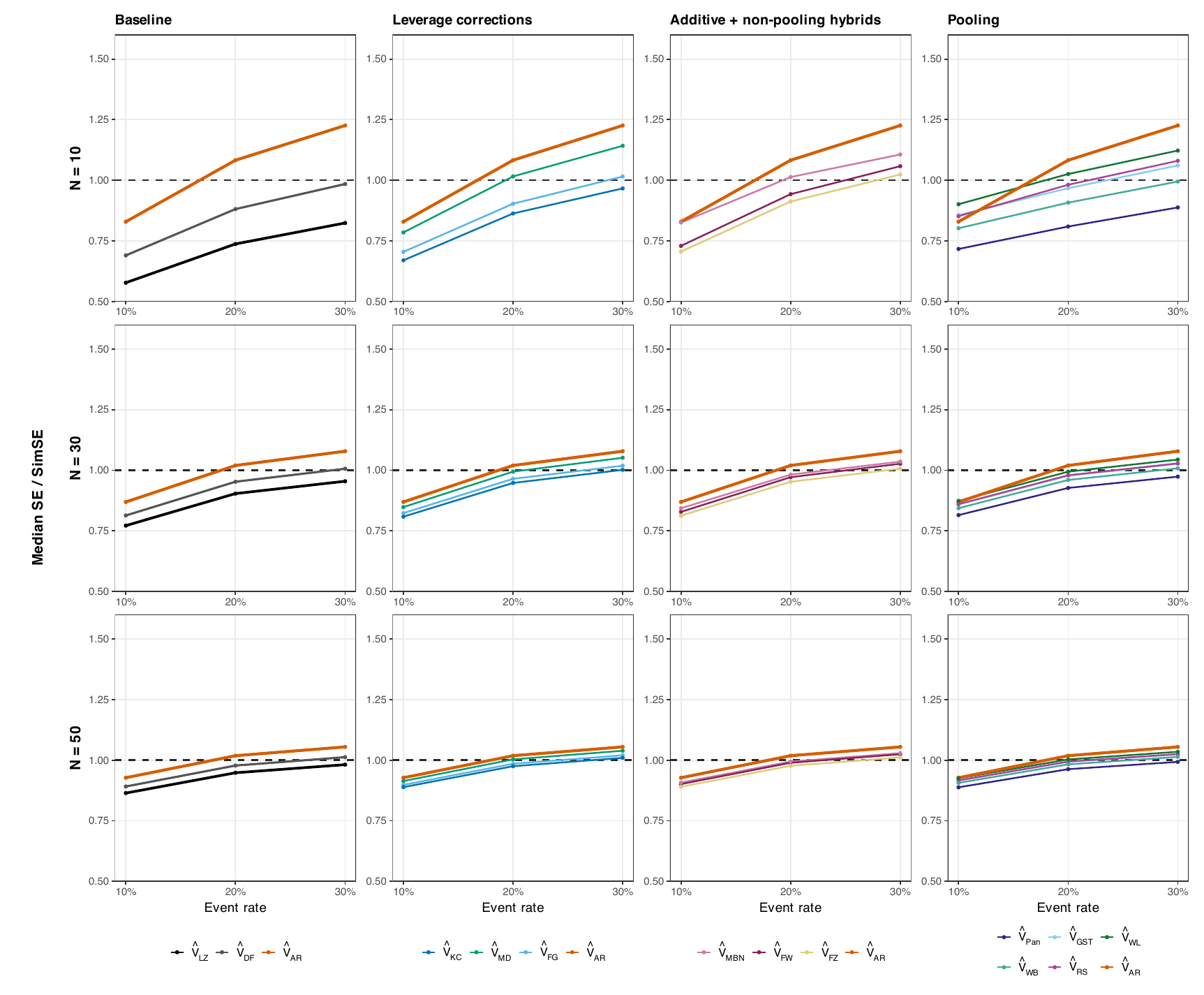}
\caption{Median standard error divided by simulation standard error for the treatment coefficient in the primary null scenarios, by number of subjects $N$ and event rate. The dashed line marks the target of one. Rows fix $N$ and columns group the estimators as in Figure~\ref{fig:sim-typeI} of the main paper. Each point averages eight scenarios with equal weight, and $\VarAR$ (orange) appears in every panel.}\label{fig:se-ratio}
\end{figure}

\begin{figure}[p]
\centering
\includegraphics[width=\textwidth]{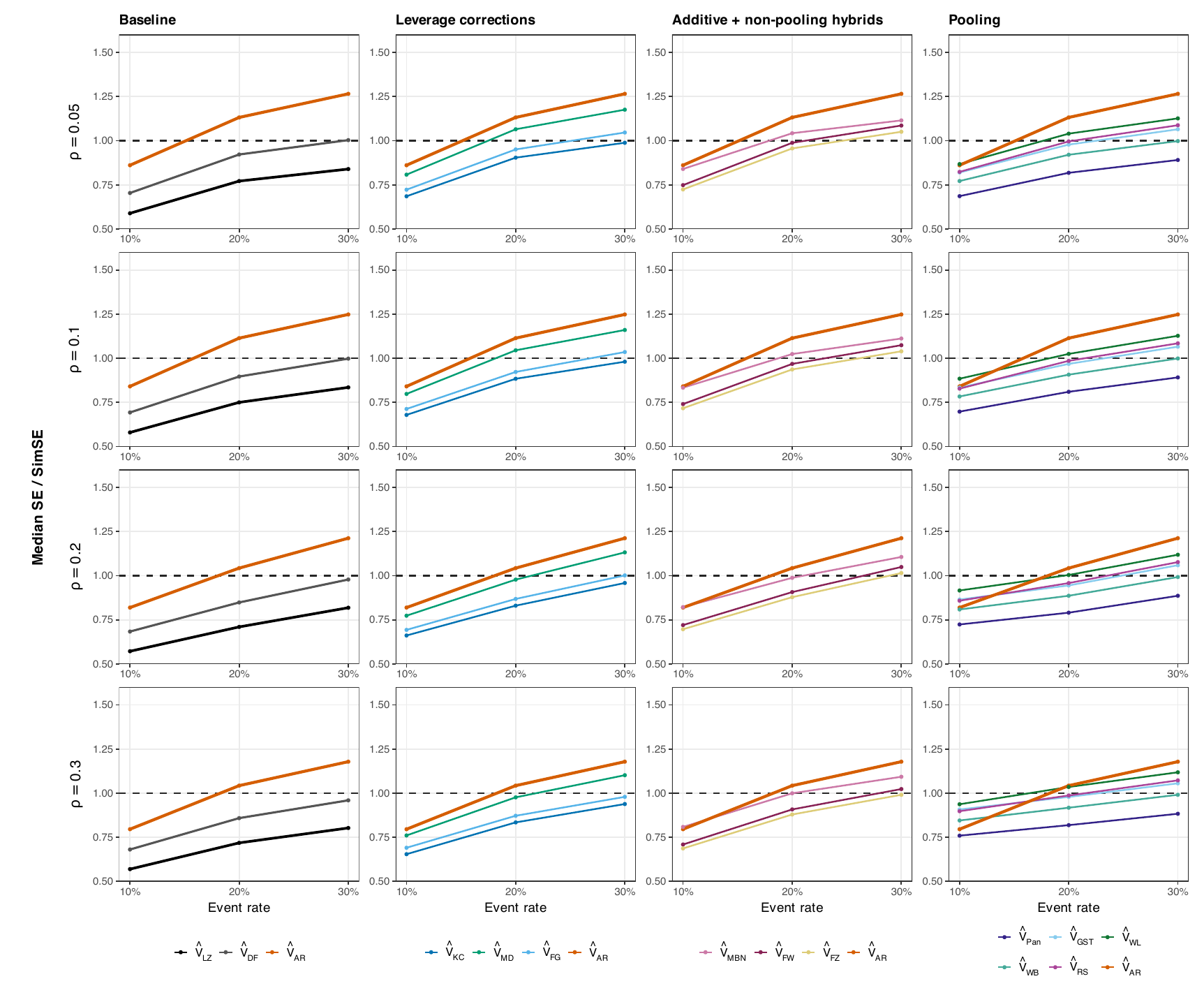}
\caption{Median standard error divided by simulation standard error for the treatment coefficient with $N=10$, by the true correlation level $\rho$. Rows fix $\rho$ and columns group the estimators as in Web Figure~\ref{fig:se-ratio}. Each point averages the two correlation structures, and the dashed line marks the target of one.}\label{fig:se-ratio-rho}
\end{figure}

\subsection{Type I error by design factor}
\label{app:supp-figures}

Web Figures~\ref{fig:typeI-structure} to~\ref{fig:typeI-unequal} show the type I error rate of the treatment test separately by design factor. Web Figure~\ref{fig:typeI-structure} separates the primary scenarios by the true correlation structure, and Web Figure~\ref{fig:typeI-misspec} shows the two directions of working-correlation misspecification. Web Figure~\ref{fig:typeI-rho} shows the primary scenarios with ten subjects by the true correlation level, and Web Figure~\ref{fig:typeI-unequal} shows the designs with unequal numbers of observations per subject. In every figure the estimators are grouped as in Figure~\ref{fig:sim-typeI} of the main paper, $\VarAR$ is drawn in orange in every panel, and the dashed line marks the nominal level of $0.05$.

\begin{figure}[p]
\centering
\includegraphics[width=0.92\textwidth,trim=0 684bp 0 0,clip]{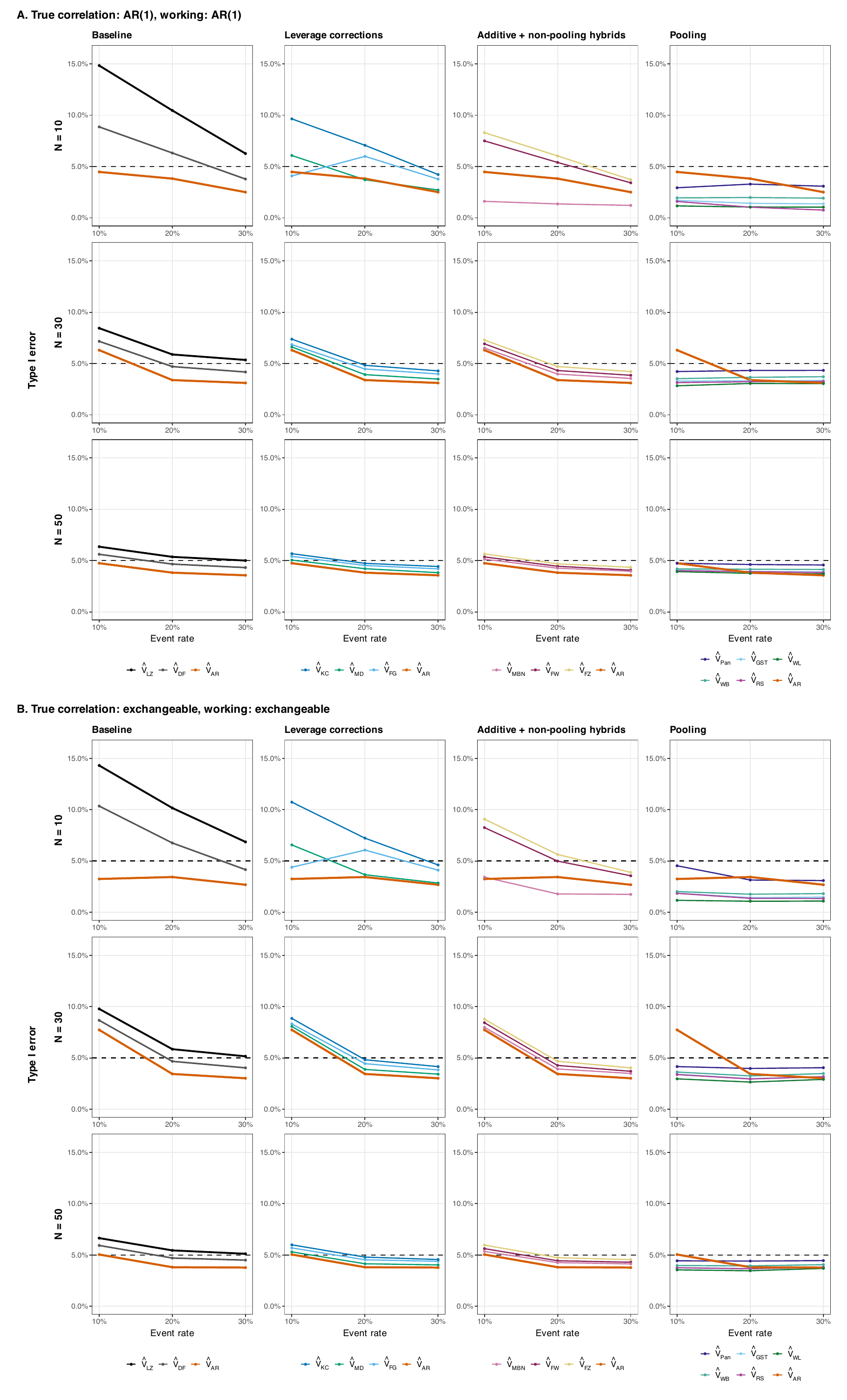}
\caption{Type I error rate of the treatment test in the primary scenarios, by true correlation structure. Panel A, on this page, has AR(1) as the true and working structure. Panel B, on the next page, has exchangeable. Rows fix the number of subjects $N$, and each point averages the four correlation levels with equal weight.}\label{fig:typeI-structure}
\end{figure}
\begin{figure}[p]
\ContinuedFloat
\centering
\includegraphics[width=0.92\textwidth,trim=0 0 0 684bp,clip]{figS1_typeI_by_corr_true.pdf}
\caption{Continued. Panel B, exchangeable true and working structure.}
\end{figure}

\begin{figure}[p]
\centering
\includegraphics[width=0.92\textwidth,trim=0 684bp 0 0,clip]{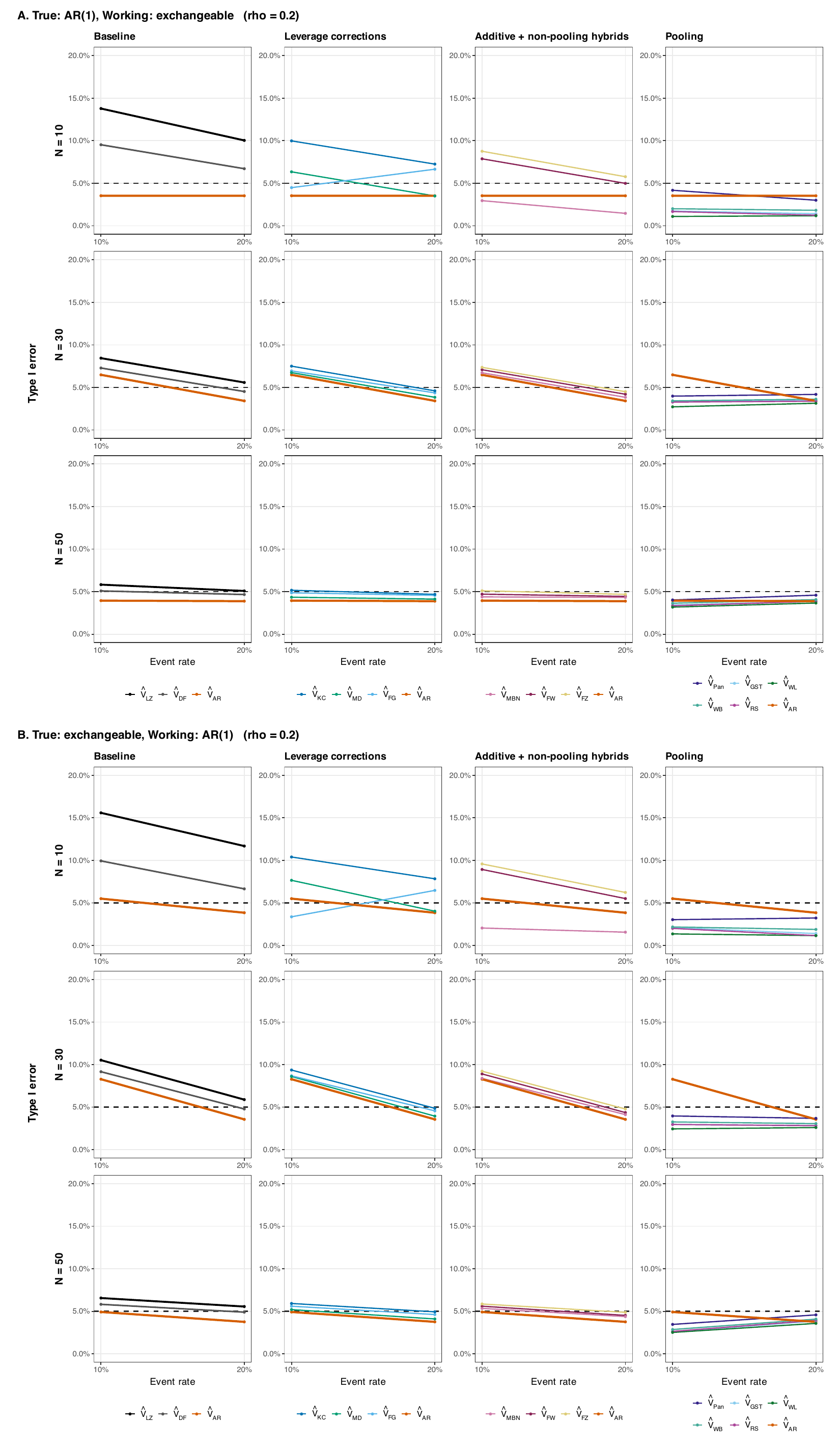}
\caption{Type I error rate of the treatment test under working-correlation misspecification, with $\rho=0.2$. In Panel A, on this page, the data have AR(1) correlation and the working structure is exchangeable. In Panel B, on the next page, the data have exchangeable correlation and the working structure is AR(1). Rows fix the number of subjects $N$.}\label{fig:typeI-misspec}
\end{figure}
\begin{figure}[p]
\ContinuedFloat
\centering
\includegraphics[width=0.92\textwidth,trim=0 0 0 684bp,clip]{figS2_typeI_misspec.pdf}
\caption{Continued. Panel B, exchangeable data with an AR(1) working structure.}
\end{figure}

\begin{figure}[p]
\centering
\includegraphics[width=\textwidth]{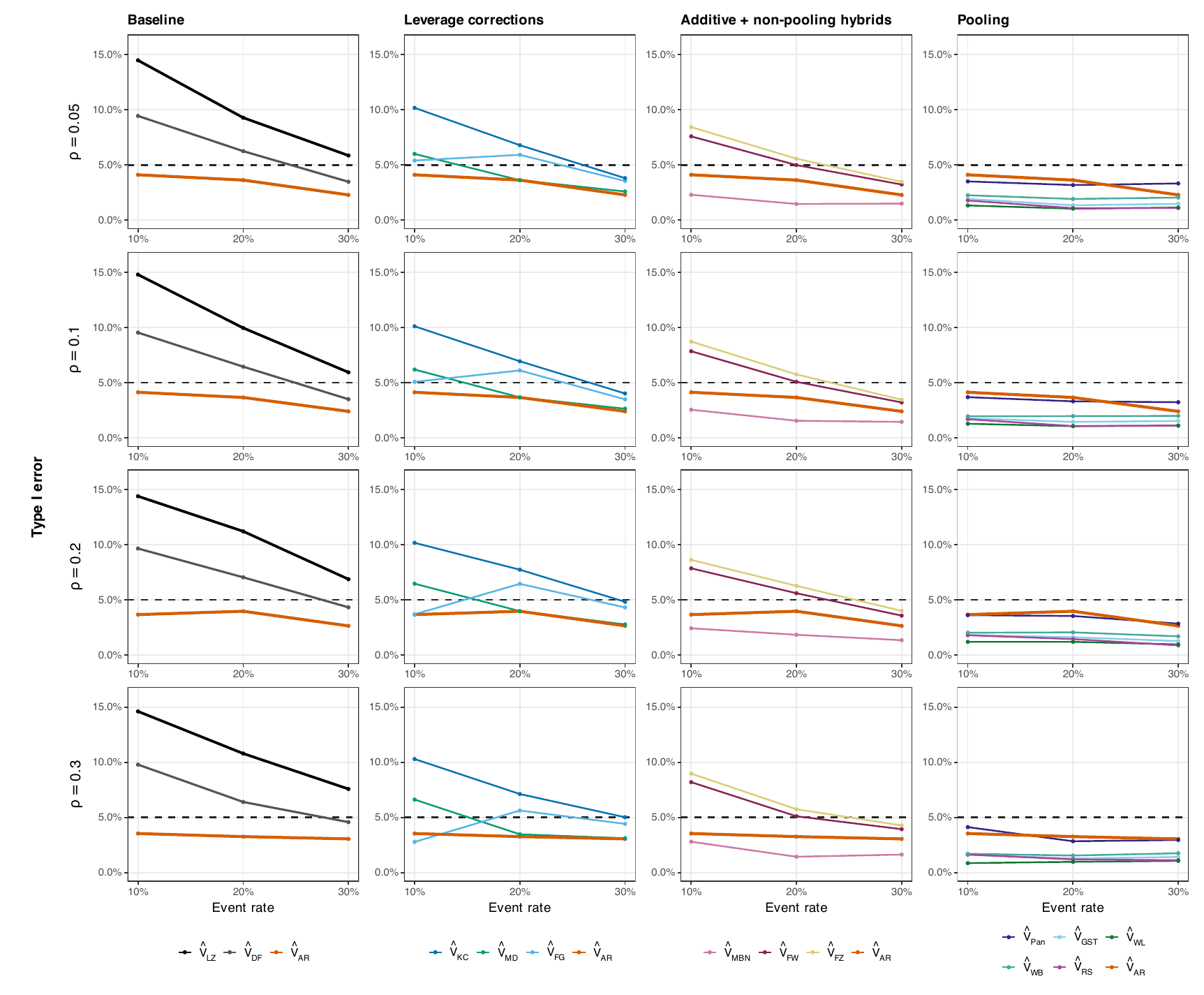}
\caption{Type I error rate of the treatment test in the primary scenarios with $N=10$, by the true correlation level $\rho$. Rows fix $\rho$, and each point averages the exchangeable and AR(1) structures with equal weight.}\label{fig:typeI-rho}
\end{figure}

\begin{figure}[p]
\centering
\includegraphics[width=\textwidth]{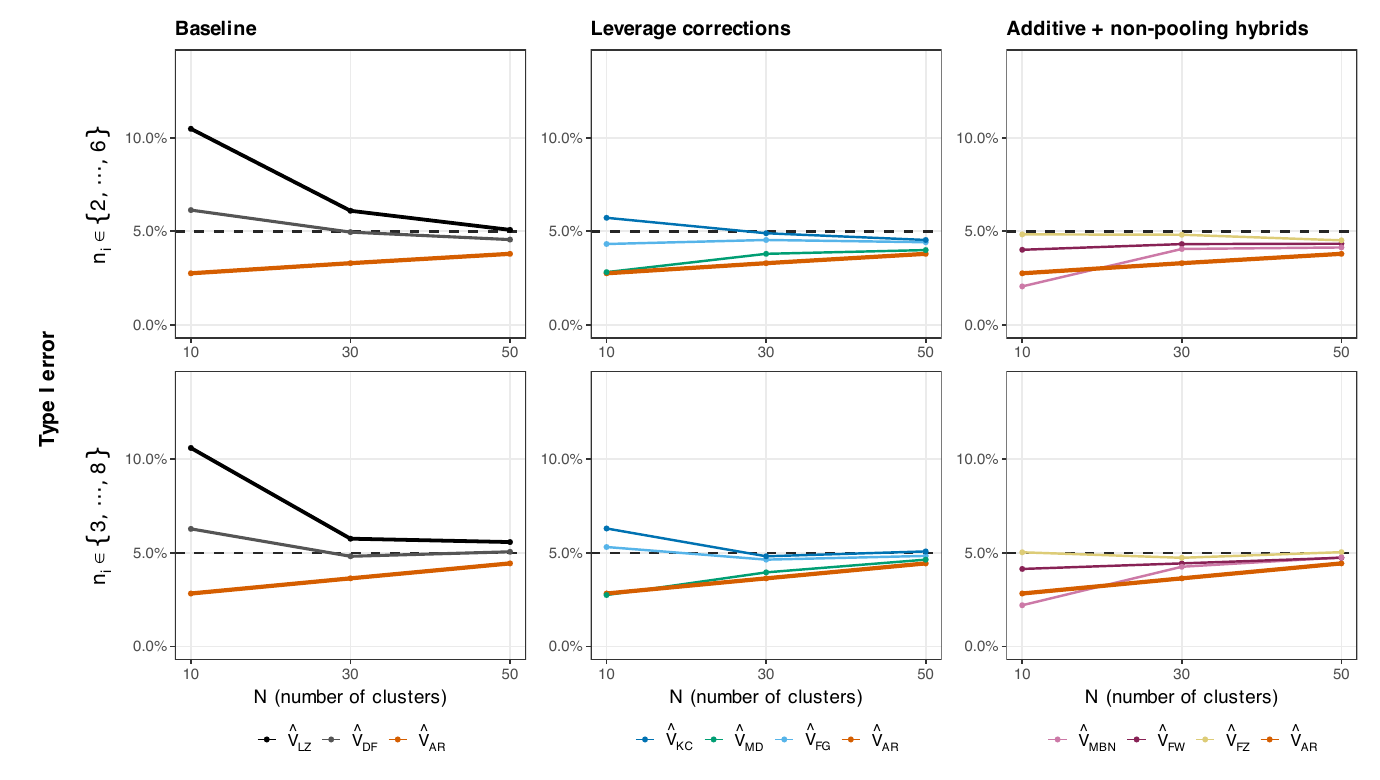}
\caption{Type I error rate of the treatment test with unequal numbers of observations per subject, by number of subjects $N$. Rows fix the range of the numbers of observations. The event rate is $20\%$ and $\rho=0.2$ with exchangeable correlation. The pooled estimators are omitted because they cannot be computed in these designs.}\label{fig:typeI-unequal}
\end{figure}

\webappendix{Datasets with an empty treatment arm}
\label{app:empty-arm-rule}

This appendix reports the exploratory analysis summarized in the Discussion of the main paper. In several scenarios with few events, the estimators that do not pool had type I error rates above the cutoff, and we examined which simulated datasets were responsible. Within each null scenario we separated the converged datasets in which one treatment arm had no events from those in which both arms had events, and we computed the type I error rate of the treatment test in each group.

Web Table~\ref{tab:empty-arm-split} gives these rates for the primary scenarios at the $10\%$ event rate, combining the eight scenarios at each number of subjects. With ten subjects, $15{,}282$ of the $38{,}696$ converged datasets ($39.5\%$) had an empty arm. With thirty subjects the share was $1{,}876$ of $39{,}991$ ($4.7\%$), and with fifty subjects $235$ of $39{,}995$ ($0.6\%$). With thirty subjects, each of the nine estimators that do not pool rejected in at least $94.9\%$ of the datasets with an empty arm and in at most $4.8\%$ of the datasets with events in both arms. The excess type I error in these scenarios therefore comes from the datasets with an empty arm. The five pooled estimators rejected in at most $7.1\%$ of those datasets at that sample size.

\begin{table}[H]
\centering\small
\caption{Type I error rate of the treatment test in the primary scenarios at the $10\%$ event rate, separately for datasets with an empty treatment arm and datasets with events in both arms.}\label{tab:empty-arm-split}
\begin{tabular}{lcccccc}
\toprule
 & \multicolumn{2}{c}{$N=10$} & \multicolumn{2}{c}{$N=30$} & \multicolumn{2}{c}{$N=50$} \\
\cmidrule(lr){2-3}\cmidrule(lr){4-5}\cmidrule(lr){6-7}
Estimator & Empty arm & Both arms & Empty arm & Both arms & Empty arm & Both arms \\
\midrule
$\hat V_{LZ}$ & 0.3341 & 0.0228 & 0.9755 & 0.0476 & 1.0000 & 0.0594 \\
$\hat V_{DF}$ & 0.2286 & 0.0094 & 0.9733 & 0.0351 & 1.0000 & 0.0521 \\
$\hat V_{KC}$ & 0.2366 & 0.0138 & 0.9733 & 0.0372 & 1.0000 & 0.0527 \\
$\hat V_{MD}$ & 0.1480 & 0.0077 & 0.9733 & 0.0291 & 1.0000 & 0.0461 \\
$\hat V_{FG}$ & 0.0912 & 0.0107 & 0.9707 & 0.0317 & 1.0000 & 0.0499 \\
$\hat V_{MBN}$ & 0.0569 & 0.0043 & 0.9488 & 0.0293 & 0.9957 & 0.0468 \\
$\hat V_{\mathrm{Pan}}$ & 0.0639 & 0.0198 & 0.0709 & 0.0405 & 0.2213 & 0.0448 \\
$\hat V_{GST}$ & 0.0338 & 0.0073 & 0.0480 & 0.0328 & 0.1830 & 0.0381 \\
$\hat V_{RS}$ & 0.0340 & 0.0063 & 0.0410 & 0.0322 & 0.1830 & 0.0380 \\
$\hat V_{WL}$ & 0.0192 & 0.0067 & 0.0453 & 0.0282 & 0.1830 & 0.0365 \\
$\hat V_{WB}$ & 0.0336 & 0.0109 & 0.0554 & 0.0348 & 0.2000 & 0.0398 \\
$\hat V_{FW}$ & 0.1830 & 0.0105 & 0.9733 & 0.0326 & 1.0000 & 0.0493 \\
$\hat V_{FZ}$ & 0.2007 & 0.0125 & 0.9733 & 0.0364 & 1.0000 & 0.0524 \\
$\VarAR$ & 0.0910 & 0.0045 & 0.9733 & 0.0257 & 1.0000 & 0.0433 \\
\bottomrule
\end{tabular}
\end{table}

The same holds at the $5\%$ event rate. Web Table~\ref{tab:empty-arm-split-5} gives the rates for the eight scenarios with a $5\%$ event rate. With thirty subjects, $3{,}822$ of the $19{,}886$ converged datasets ($19.2\%$) had an empty arm, and with fifty subjects $1{,}358$ of $19{,}991$ ($6.8\%$). With thirty subjects $\VarAR$ rejected in $75.1\%$ of the datasets with an empty arm and in $3.2\%$ of the others, and with fifty subjects in $94.7\%$ and $4.9\%$.

\begin{table}[H]
\centering\small
\caption{Type I error rate of the treatment test in the eight scenarios with a $5\%$ event rate, separately for datasets with an empty treatment arm and datasets with events in both arms.}\label{tab:empty-arm-split-5}
\begin{tabular}{lcccc}
\toprule
 & \multicolumn{2}{c}{$N=30$} & \multicolumn{2}{c}{$N=50$} \\
\cmidrule(lr){2-3}\cmidrule(lr){4-5}
Estimator & Empty arm & Both arms & Empty arm & Both arms \\
\midrule
$\hat V_{LZ}$ & 0.7758 & 0.0524 & 0.9713 & 0.0653 \\
$\hat V_{DF}$ & 0.7548 & 0.0425 & 0.9470 & 0.0567 \\
$\hat V_{KC}$ & 0.7640 & 0.0445 & 0.9595 & 0.0581 \\
$\hat V_{MD}$ & 0.7496 & 0.0351 & 0.9470 & 0.0523 \\
$\hat V_{FG}$ & 0.7399 & 0.0348 & 0.9595 & 0.0545 \\
$\hat V_{MBN}$ & 0.6627 & 0.0354 & 0.9212 & 0.0517 \\
$\hat V_{\mathrm{Pan}}$ & 0.0597 & 0.0477 & 0.0361 & 0.0471 \\
$\hat V_{GST}$ & 0.0552 & 0.0387 & 0.0287 & 0.0421 \\
$\hat V_{RS}$ & 0.0563 & 0.0388 & 0.0287 & 0.0423 \\
$\hat V_{WL}$ & 0.0345 & 0.0334 & 0.0243 & 0.0389 \\
$\hat V_{WB}$ & 0.0547 & 0.0403 & 0.0339 & 0.0433 \\
$\hat V_{FW}$ & 0.7548 & 0.0398 & 0.9470 & 0.0545 \\
$\hat V_{FZ}$ & 0.7614 & 0.0438 & 0.9595 & 0.0575 \\
$\VarAR$ & 0.7512 & 0.0324 & 0.9470 & 0.0491 \\
\bottomrule
\end{tabular}
\end{table}

As a check, we repeated the treatment tests with the pooled estimator $\hat V_{\mathrm{Pan}}$ whenever a dataset had an empty arm, and with a leverage-corrected estimator otherwise. We did this in the $130$ null scenarios with equal numbers of observations per subject, where $\hat V_{\mathrm{Pan}}$ can be computed. Web Table~\ref{tab:empty-arm-check} gives, for four estimators, the number of these $130$ scenarios with a type I error rate above the cutoff and the largest rate, first when the estimator is used for every dataset and then with the check. With $\VarAR$ the number of scenarios above the cutoff fell from $22$ to $1$, and the largest rate fell from $0.2238$ to $0.0630$.

\begin{table}[H]
\centering\small
\caption{Number of the $130$ null scenarios with equal numbers of observations per subject that have a type I error rate above $0.0562$, and the largest rate, when each estimator is used for every dataset and when $\hat V_{\mathrm{Pan}}$ is used for the datasets with an empty treatment arm.}\label{tab:empty-arm-check}
\begin{tabular}{lcccc}
\toprule
 & \multicolumn{2}{c}{Scenarios above the cutoff} & \multicolumn{2}{c}{Largest type I error rate} \\
\cmidrule(lr){2-3}\cmidrule(lr){4-5}
Estimator & Every dataset & With the check & Every dataset & With the check \\
\midrule
$\hat V_{KC}$ & 53 & 7 & 0.2372 & 0.0752 \\
$\hat V_{FG}$ & 36 & 4 & 0.2208 & 0.0706 \\
$\hat V_{MD}$ & 34 & 3 & 0.2255 & 0.0674 \\
$\VarAR$ & 22 & 1 & 0.2238 & 0.0630 \\
\bottomrule
\end{tabular}
\end{table}

\end{document}